\documentclass[11pt,onecolumn]{article}
\usepackage{enumerate}
\usepackage{float}
\usepackage{xcolor}
\definecolor{DarkGray}{rgb}{0.1,0.1,0.5}
\usepackage[colorlinks=true,breaklinks, linkcolor= DarkGray,citecolor= DarkGray,urlcolor= DarkGray]{hyperref}
\usepackage{graphicx}
\usepackage{amsmath}
\usepackage{amssymb}
\usepackage{amsthm}
\usepackage{epsfig}
\usepackage{subfigure}
\newcommand{\be}{\begin{equation}}
\newcommand{\ee}{\end{equation}}
\newcommand{\ba}{\begin{array}}
\newcommand{\ea}{\end{array}}
\newcommand{\bea}{\begin{eqnarray}}
\newcommand{\eea}{\end{eqnarray}}

\newcommand*{\textfrac}[2]{{#1}/{#2}}
\newcommand{\expeq}{\stackrel{\mathrm{exp}}{=}}

\newcommand*{\cO}{\mathcal{O}}

\newcommand{\calE}{{\cal E }}

\newcommand{\calT}{{\cal T }}

\newcommand{\calO}{{\cal O }}

\newcommand{\ZZ}{\mathbb{Z}}
\newcommand{\CC}{\mathbb{C}}
\newcommand{\RR}{\mathbb{R}}
\newcommand{\EE}{\mathbb{E}}
\newcommand{\la}{\langle}
\newcommand{\ra}{\rangle}

\newcommand{\nn}{\nonumber}

\newcommand{\trace}{\mathop{\mathrm{Tr}}\nolimits}

\newcommand{\pf}{\mathop{\mathrm{Pf}}\nolimits}

\newcommand*{\ket}[1]{|#1\rangle}
\newcommand*{\bra}[1]{\langle #1|}
\newcommand*{\proj}[1]{|#1\rangle\langle #1|}
\newcommand*{\spr}[2]{\langle #1|#2\rangle}
\newcommand{\bad}{\textrm{bad}}
\newcommand{\good}{\textrm{good}}
\newtheorem{dfn}{Definition}
\newtheorem{lemma}{Lemma}
\newtheorem{prop}{Proposition}

\newtheorem{theorem}{Theorem}

\newtheorem{corol}{Corollary}

\newtheorem{fact}{Fact}

\newcommand*{\random}{R}
\newcommand*{\uncomment}[1]{}
\newcommand*{\ExpE}[1]{\mathbb{E}{\left[#1\right]}}
\newcommand*{\fw}{\mathsf{f}} %fermionic weight
\newcommand*{\qw}{\mathsf{q}} % qubit weight

\numberwithin{equation}{section}
\newcommand*{\cwave}{c_1}
\newcommand*{\xiwave}{\xi_1}

\title{Disorder-assisted error correction in Majorana chains}

\date{\today}
\author{Sergey Bravyi and Robert K\"onig\\
\small IBM T.J. Watson Research Center, Yorktown Heights, NY 10598, USA}
\begin{document}

\maketitle

\begin{abstract}
It was recently realized that quenched disorder may enhance the reliability of topological qubits by reducing the mobility of anyons at zero temperature. Here we compute storage times with and without disorder for quantum chains with unpaired Majorana fermions --- the simplest  toy model of a quantum memory. Disorder takes the form of a random site-dependent chemical potential. The corresponding one-particle problem is a one-dimensional Anderson model with  disorder in the hopping amplitudes.  We focus on the zero-temperature storage of a qubit encoded in the ground state of the Majorana chain. Storage and retrieval are modeled by a unitary evolution under the memory Hamiltonian with an unknown weak perturbation followed by an error-correction step. Assuming dynamical localization of the one-particle problem, we show that the storage time grows exponentially with the system size. We give supporting evidence for the required localization property by estimating Lyapunov exponents of the one-particle eigenfunctions.  We also simulate the storage process for chains with a few hundred sites. Our numerical results indicate that  in the absence of disorder, the storage time grows only as a logarithm of the system size. We provide numerical evidence for the beneficial effect of disorder on storage times and show that suitably chosen pseudorandom potentials can outperform random ones.
\end{abstract}

\newpage

\tableofcontents

\newpage

\section{Introduction}

Topological features of physical systems are -- by  definition  -- robust against local imperfections. As a consequence, they are the natural object of study in the  characterization and classification of different phases of matter. A major challenge is to identify and  compute topological invariants for various systems and to connect these to experimentally accessible quantities. Examples include Chern numbers, the topological entanglement entropy~\cite{KitaevPreskill06,LevinWen06entr}, or the classification of  free fermions in terms of symmetry classes.

In contrast, topological quantum computing attempts to make  operational use of robust features of topologically ordered systems. A basic first goal is the reliable storage and readout of quantum information. Proposed toy
models for topological quantum memories include Kitaev's Majorana chains~\cite{Kitaev00} and the toric code~\cite{Kitaev1997}. These systems exhibit a `topological' ground space degeneracy which persists in the presence of local perturbations.
In both cases, the ground space constitutes a quantum error-correcting code. When superselection rules are taken into account (in the case of Majorana chains), these codes
have a macroscopic distance which scales linearly with the system size.

Unfortunately, these remarkable ground space properties are insufficient to protect the encoded information in the presence of
generic local perturbations and/or thermal fluctuations: localized excitations (or {\em anyons}) can be created at constant energy cost and once created, can propagate freely to cause logical errors. This fact was recognized early on~\cite{Dennis2002} and has since been elaborated by several authors by giving explicit examples~ (see e.g.,~\cite{Pastawskietal10,Starketal11}) and general impossibility proofs for storage at non-zero temperature.
These range from the investigation of  the robustness of topological order using finite-temperature generalizations of the topological entanglement entropy~\cite{Chamon2007} to bounds on  operational quantities such as the quantum memory relaxation rate~\cite{CLBT09}. For the 2D toric code, Alicki et al. showed that a weak coupling to a Markovian environment implies a constant lower bound on the relaxation rate of any logical state independent of the system size~\cite{AFH:qmem}.

Several approaches for combating decoherence in quantum memories have been considered in the literature. Arguably the most intriguing of these is the idea of self-correcting quantum memories. This completely circumvents the decoherence mechanism mentioned above: it is based on a Hamiltonian which has no
 point-like excitations carrying a topological charge (such as anyons). More specifically, it assumes a Hamiltonian whose ground space is a quantum code with macroscopic distance and the property that high-weight errors correspond to larger energies. As a result, encoded states are separated by (macroscopic) energy barriers, a necessary condition for self-correction as formalized in~\cite{Bravyi2008} (see also~\cite{Kay2008}). A 4D generalization of the toric code is known~\cite{Dennis2002} to provide a stable quantum memory in this fashion: it was shown in~\cite{Alicki08} that this memory is thermally stable in the sense that relaxation times increase exponentially with system size. Unfortunately, this model  cannot be embedded in a locality-preserving way in~$\mathbb{R}^3$.  The existence of self-correcting memories in 3D remains under investigation and is of fundamental theoretical interest.  On the other hand, simpler geometries, such as 2D (subspace) stabilizer codes cannot provide this kind of protection, as shown in~\cite{Bravyi2008}.

A different line of research attempts to specifically address the creation and movement of anyons responsible for decoherence in topologically ordered systems. Importantly, while these excitations are thermally suppressed at low temperatures, their propagation impacts the coherence of quantum memories even in the case of zero temperature. Limiting this movement is therefore at the core of these methods. Each of these methods has its own merits and drawbacks.

An example is active error correction which employs frequent rounds of syndrome measurements
revealing the locations of anyons at any given time step.
The syndrome information  can then  be used to undo the effect of errors caused by the propagation of anyons.
 It has been shown to permit reliable storage of quantum information
for arbitrarily long times in e.g., the 2D toric code~\cite{Dennis2002}, with errors modeled as depolarizing noise.
However, this approach requires significant experimental capabilities: the syndrome measurements need to be
sufficiently accurate and sufficiently fast  in order for the error correction to succeed~\cite{Dennis2002}.

Given the stringent experimental requirements for active error-correction, it is  natural to consider other mechanisms for reducing the propagation of anyons in topological quantum memories. One such approach is the introduction of long-range interactions between anyons.  A concrete model where such an interaction is mediated by a bosonic background field was proposed by
Hamma at al~\cite{HCC:toricboson}.
An alternative approach in which creation of anyons from the vacuum is penalized by long-range repulsive
interactions was studied by Chesi et al~\cite{Chesi2009}.
A promising 3D stabilizer model  was proposed more recently by Haah~\cite{Haah11}.
In this model anyons are bound together by energy barriers growing logarithmically with the
lattice size~\cite{Bravyi11}  even though the underlying  Hamiltonian includes only short-range  interactions
with constant strength.

\subsubsection*{Quenched disorder in topological quantum memories}
Here we focus on another potentially more realistic  approach towards limiting the detrimental effects of propagating anyons pioneered by Wootton and Pachos~\cite{WoottonPachos11}, and, independently, by
Stark et al~\cite{Starketal11}. It relies on the observation
that the mobility of anyons can be  strongly reduced  in the presence of {\em quenched disorder}.
 Such disorder may account for natural imperfections in physical realizations, or may be artificially engineered by tuning the interactions. Its stabilizing role on topological order has been investigated in terms of the topological entanglement entropy~\cite{Tsokomosetal11}. Here we are interested in how disorder enhances the reliability of a quantum memory. The underlying intuition is analogous to the reasoning used to explain why disorder suppresses
 zero-temperature electron transport in wires: the anyons, like electrons, get trapped in the many local minima present in a random potential. This phenomenon is commonly known as Anderson localization, named after the model introduced by Anderson to study the effect of randomness on quantum transport~\cite{Anderson1958}.

 Perhaps the most well-studied paradigmatic model exhibiting Anderson localization is that of a single particle hopping on $\mathbb{Z}^d$ with random on-site potential at each site. This corresponds to a tight-binding approximation of an electron moving in some random medium. The model Hamiltonian acting on states $\psi\in \ell^2(\mathbb{Z}^d)$ takes the form $H_\random=-\Delta+\eta W$, where $\Delta$ is the discrete Laplacian, $(W\psi)(x)=W(x)\psi(x)$ and where
 $\{W(x)\}_{x\in\mathbb{Z}^d}$ are identically and independently distributed random variables with support e.g., on $[-1,1]$. The parameter $\eta\geq 0$ controls the strength of the disorder in this setting. A realization $H_\random$ of such a random Schr\"odinger operator is said to have exponentially localized eigenfunctions with a localization length $\xi_s$ if every eigenfunction $\psi$ satisfies~\footnote{It is worth mentioning that
a typical realization of a random potential $W(x)$  for an infinite lattice
contains arbitrarily long intervals at which $W(x)$  is arbitrarily close
to a constant function. Such almost flat intervals would  give rise to eigenfunctions which
are delocalized at small scales. Note that the constant $C$ in  Eq.~(\ref{eq:exponentiallocalization})
could be arbitrarily large, so this bound  does not prevent $\psi(x)$ from having oscillatory behavior for small $|x-x_0|$.}
\begin{align}
|\psi(x)|\leq C |\psi(x_0)|e^{-|x-x_0|/\xi_s}\ \label{eq:exponentiallocalization}
\end{align}
 for some $x_0\in\mathbb{Z}^d$
   and some constant $C\ge 1$ depending on $\psi$. This is known to hold with high probability  for arbitrary disorder strength
   in 1D systems ($d=1$). This fact can be established using transfer matrix methods and computing Lyapunov exponents~\cite{DelyonKunzSouillard83,DelyonLevySouillard85}. It is in stark contrast to
the delocalized form of the eigenstates of the `clean' Hamiltonian~$H_C=-\Delta$, which are Bloch waves. For $d>1$, as well as for certain kinds of disorder, the situation is more complex: for example, there can be both localized and extended eigenstates separated by an energy threshold called a mobility edge.
Theoretical manifestations of localization such as~\eqref{eq:exponentiallocalization} or the so-called dynamical localization condition
 can be connected to experimentally measurable quantities, e.g., conductivity.
Powerful analytic tools for establishing such conditions for more general families of random Hamiltonians have been developed (see e.g.,~\cite{Kirsch07,Stolz11} for two recent reviews). This includes Hamiltonians with  more general forms of disorder such as randomness in the hopping terms, i.e., off-diagonal disorder (see e.g.,~\cite{Delyonetal87}).

A first step towards showing that Anderson localization increases the robustness of a topological qubit is to
prove that the given physical system indeed has the desired  localization properties.
A difficulty here is that time-evolving topological memories generically involve coherent superpositions of states with varying number of anyons --- a regime in which the standard Anderson localization theory cannot be directly applied. An approach pursued in~\cite{WoottonPachos11,Starketal11}
is to restrict the dynamics to subspaces with a fixed particle number and neglect processes creating and destroying pairs of anyons.  Within this approximation localization properties of the corresponding multi-particle states
have been established both numerically~\cite{WoottonPachos11} and analytically~\cite{Starketal11,AizenmanWarzel09,ChulaevskySuhov09}.

A second step is to connect localization properties of states to the robustness  of the quantum memory. The first analysis
along these lines starting from two-particle wavefunction localization was provided in~\cite{WoottonPachos11,Starketal11} for the toric code. By using the fact that logical errors correspond to topologically non-trivial trajectories of such pairs,  some estimates on the failure probability starting from a two-particle configuration as well as certain geometrically arranged multi-particle initial configurations were obtained. However, these estimates neglect the full many-particle dynamics, and hence only provide a somewhat qualitative picture of the effect of localization.

In order to assess the effect of disorder on the stability of topological memories more rigorously, we first need to
choose a sufficiently simple model describing encoding, storage, and readout stages
as well as  some benchmark measure to compare the robustness of a memory in different regimes.
Let us formulate this in more detail.  We will describe the zero-temperature quantum information storage
using a {\em memory Hamiltonian}~$H_0$,  a {\em perturbation}~$V$, and an
{\em error correction} operation $\Phi_{ec}$.  The memory Hamiltonian $H_0=-\sum_j S_j$
represents an ideal topological quantum memory such as the toric code model.
It is a sum of local pairwise commuting {\em stabilizers}  $S_j$ where each stabilizer
has eigenvalues $\pm 1$. Ground states of $H_0$ are $+1$ eigenvectors of any stabilizer $S_j$
while excitations correspond to flipped eigenvalues, $S_j=-1$.
The encoding stage amounts to initializing the system in some ground state $|g\ra$ of $H_0$,
see Section~\ref{subs:encoding}.
The ground state of $H_0$  must be degenerate to permit encoding of one or several qubits.
The storage is modeled by a unitary evolution under a Hamiltonian $H=H_0+V$
for some time $t$ which results in a final state $|g(t)\ra=e^{i(H_0+V)t}\, |g\ra$.
Here we assume that the perturbation $V$ is either unknown or known only partially,
so the unitary evolution cannot be undone even if we have a full control of the system.
However, if the perturbation $V$ is sufficiently weak and the evolution time~$t$ is not too long,
the final state~$|g(t)\ra$ can be regarded as a slightly corrupted version of the initial state $|g\ra$
with a small density of errors which can be corrected at the readout stage.
The readout begins by measuring a {\em syndrome}, i.e., measuring the eigenvalue $\pm1$ for every
stabilizer $S_j$. The syndrome information $s$ is fed into an error correction algorithm that
determines a unitary correcting operator~$C(s)$ returning the system back to some ground state of $H_0$.
The error-corrected final state is therefore $\Phi_{ec}(|g(t)\ra\la g(t)|)$, where
$\Phi_{ec}$ is a linear map describing how the error correction acts on states. More explicitly,
$\Phi_{ec}(\rho)=\sum_s C(s) Q_s \rho Q_s C(s)^\dag$, where $Q_s$ is the projector onto
a subspace with a fixed syndrome $s$, see Section~\ref{subs:EC} for details.
We will measure robustness of a memory by its {\em storage fidelity} --- the overlap
between the initial encoded state and the final error-corrected state:
\be
F_{\ket{g}}(t)=\la g| \Phi_{ec}(e^{i(H_0+V)t}\, |g\ra\la g| e^{-i(H_0+V)t})|g\ra, \label{eq:fidelitytime}
\ee
and the worst-case fidelity $F(t)=\min_{\ket{g}} F_{\ket{g}}(t)$ minimized over all encoded states.
By definition, $0\le F(t)\le 1$ with $F(0)=1$. A perfect recovery of encoded information corresponds to~$F(t)=1$.
Another measure we will use is the {\em storage time}~$T_{\textrm{storage}}=T_{\textrm{storage}}(F_0)$.
This is the time it takes for the storage fidelity~$F(t)$ to drop below a given threshold value $F_0$,
say, $F_0=0.99$ (we note that although~$F(t)$ tends to decrease with increasing~$t$, its behavior may be
highly non-monotone, see e.g. Fig.~\ref{figosc:subfig2} in Section~\ref{sec:numerics}).
We note that error correction is necessary here to give a fair assessment of the information recoverable from the memory. But its role  is different from the case of active error-correction: it is only applied once in this process when the information is retrieved after time~$t$. In particular, it plays no direct role in preserving coherence in the time interval~$[0,t]$.

Quenched disorder will be modeled by introducing some randomness into the perturbation~$V$. Clearly, physically realistic models of disorder depend on the system under consideration. Here we study the Majorana chain model proposed by Kitaev~\cite{Kitaev00}. It describes a quantum wire put in contact with a bulk superconductor
which enables creation and destruction of electron pairs in the wire through tunneling of Cooper pairs
(see Section~\ref{subs:themodel} for a formal definition of the model and a discussion of its properties).  Since this is a system of electrons, the simplest choice of disorder is a random site-dependent chemical potential. The corresponding perturbation is
$V=-\sum_{j=1}^N \mu_j a_j^\dag a_j$, where $a_j^\dag$ $(a_j)$ are electron creation (annihilation)
operators, $N$ is the system size, and $\mu_j$ is the chemical potential at the site $j$.
Furthermore, we assume that $\mu_j=\mu+\eta x_j$, where $x_j\in [-1,1]$ are independent
identically distributed random variables drawn from the uniform distribution, while $\eta>0$
controls the disorder strength. The constant $\mu$ represents the homogeneous part of the perturbation.
More general perturbations that might be relevant for the Majorana chain model
are described in Section~\ref{subs:themodel}.
We note that our description of quenched disorder differs from the one of~\cite{WoottonPachos11,Starketal11}
where the randomness was introduced into the memory Hamiltonian $H_0$ by allowing
site-dependent random coefficients in front of the stabilizers $S_i$.
The reason is that the latter  choice of randomness would be very artificial and hard to motivate
in the case of Majorana chains, see Section~\ref{sec:qwires}.

In the presence of disorder  the storage fidelity $F(t)$ is a random variable
depending on the disorder realization $\{x_j\}$.  We will mostly be interested in the
expectation $\ExpE{F(t)}$ over disorder realizations.
 Comparing the behavior of
$F(t)$ for the clean case ($\eta=0$) and  $\ExpE{F(t)}$ for
the disordered case ($\eta>0$) gives a precise way of studying the effect of disorder on the robustness of quantum memories. Here we are particularly interested in the asymptotic scaling
of the expected storage time~$\ExpE{T_{\textrm{storage}}}$ in the limit of  large system size.

\subsubsection*{Main results}
Our results are three-fold. First, we formulate a dynamical localization condition for the one-particle Hamiltonian which is sufficient to achieve an exponential scaling
of the expected storage time, that is,
\begin{align}
\ExpE{T_{\textrm{storage}}} = e^{\Omega(N)},  \label{eq:storagetimerandom}
\end{align}
in the limit of large system size $N$, see Theorem~\ref{thm:storage} in Section~\ref{subs:results} for details.
 We conjecture that our dynamical localization condition is satisfied in the limit of weak perturbations and strong disorder, that is, $\mu\ll \eta \ll 1$.
In other words, the homogeneous part of the perturbation must be small compared with the random part.
We prove that  the required localization condition is satisfied whenever
the entries of the orthogonal matrix describing the time evolution of the Majorana modes decay exponentially away from the diagonal with an~$N$-independent localization length~$\xi_1$, see
Lemma~\ref{lem:DLversussingle} in Section~\ref{sec:momentlocalization}.
The localization length~$\xi_1$ may diverge in the limit $\eta\to 0$,
as is the case in the standard 1D Anderson model, but it must be upper bounded
as~$\xi_1=O(\eta^{-\gamma})$ for some sufficiently small constant $\gamma>0$.

Second, we give supporting evidence that the desired scaling of the localization length can be achieved
by computing Lyapunov exponents of the one-particle eigenfunctions, see Section~\ref{sec:lyapunov}. This suggests a scaling $\xi_s \sim \log{(\eta^{-1})}$ in the limit~$\eta \to 0$ when the ratio~$\mu/\eta$ is kept constant.
We note that the logarithmic divergence of the localization length at the band center is a common feature of systems with disorder in the hopping amplitudes (so called off-diagonal disorder), see e.g.~\cite{eggarterried78,DelyonKunzSouillard83}.

Third, we  compute the storage time numerically
using the Monte Carlo technique developed by Terhal and DiVincenzo~\cite{TerhalDiVincenzo02}
for simulating quantum dynamics and measurements for non-interacting fermions (see Section~\ref{sec:numerics}).
The running time of our algorithm grows as $N^3/\delta^2$,
where $\delta$ is the precision up to which one needs to estimate the storage fidelity.
It allows us to compute the storage time for chains with  a few hundred sites (up to $N=256$)
in the regime of strong perturbations\footnote{As we discuss below, the storage fidelity is close to $1$
whenever $\epsilon<1/\sqrt{N}$ simply because the perturbed ground state has a large overlap with the
unperturbed one. To explore the asymptotic scaling of the storage time
for weak perturbations, say, $\epsilon\sim 10^{-2}$, one would need to simulate chains with at least
$N\sim 10^{4}$ sites.}, that is, $\mu \sim 1$ and $\eta =0$ (clean case), and  $\eta \sim \mu \sim 1$ (disordered case).
 The simulation shows that in the absence of disorder the storage time grows  as a logarithm of the system size:
\begin{align}
T_{\textrm{storage}} \sim O(\log N). \label{eq:logscaling}
\end{align}
This scaling has been recently predicted by Kay~\cite{Kay11}
based on mean-field arguments, see Section~\ref{sec:numerics} for details.
In the presence of disorder we observe
an approximately linear scaling~$\ExpE{T_{\textrm{storage}}}\sim N$.
This confirms the expected enhancement of the storage time, although the enhancement is much weaker
than our theory predicts, see Eq.~(\ref{eq:storagetimerandom}).
This discrepancy could be accounted for by the fact that the system size $N$ is comparable with the localization
length $\xi_s$ in the simulated regime, whereas   Eq.~(\ref{eq:storagetimerandom}) is expected to hold
only when $N\gg \xi_s$. It could also point to
an interesting possibility that a crossover from a polynomial to an exponential scaling of the storage
time  occurs as one interpolates between strong ($\mu \sim \eta \sim 1$) and
weak ($\mu \ll\eta \ll 1$) perturbations.
Finally, we give examples where
an artificially engineered deterministic disorder potential  leads to improved storage times compared to random disorder.

For a more formal statement of our results see Section~\ref{subs:results}.

\subsubsection*{Discussion}

The Majorana chain model~\cite{Kitaev00} has a number of properties which make it amenable to the study of localization. A major difference to e.g., the 2D  toric code is that the perturbed memory Hamiltonian~$H_0+V$
is an exactly solvable model, assuming that the perturbation $V$ is quadratic in fermionic operators. This means that the  full many-particle dynamics can be connected to a single-particle problem. The exact solvability of the model
provides powerful analytical tools for analyzing the effect of a perturbation, see e.g. Theorem~\ref{thm:QAC}
in Section~\ref{subs:results} and makes possible efficient numerical computation of the storage fidelity.
Finally, Anderson localization in the 1D geometry
requires only arbitrarily weak disorder which is likely to be present in the system due to natural imperfections (see Section~\ref{sec:lyapunov}).

Practical proposals for physically realizing the Majorana fermion chain model have recently received a lot of attention. Fu and Kane~\cite{FuKane08} showed that a 1D~wire with Majorana edge modes can be realized in a superconductor-topological insulator-superconductor junction. This exploits the proximity effect between an ordinary ($s$-wave) superconductor and the surface of a strong topological insulator. Several groups
proposed to realize and compute with Majorana fermions in 1D semiconducting wires or networks deposited on an $s$-wave superconductor~\cite{Sauetal10,Oregetal10,Aliceaetal11}. More recently, Jiang et al.~\cite{Jiangetal11} suggested a realization of Majorana fermions  using optically trapped fermionic atoms. Methods for accessing the encoded information were studied in~\cite{JiangKanePreskill11,BondersonLutchyn11,Houetal11,Hassleretal11}.

One may wonder whether the paradigmatic toy model of a disordered memory studied here bears any relevance to these practical proposals.  Especially with realizations based on superconductors, a major difference to our idealized model is that  most physical operations directly couple to the topological charge, i.e.,  the quasi-particles defined by the system Hamiltonian. This is fortunate if we consider a regime of  small system size (compared to the inverse perturbation strength) and short storage times: here the information may be thought of as evolving under a Hamiltonian without hopping terms, rendering error correction unnecessary (cf.~Corollaries~\ref{corol:boring1},\ref{corol:boring2}
in Section~\ref{subs:results} and Fig.~\ref{figosc:subfig1} in Section~\ref{sec:numerics}).  However, outside this regime, error-correction is essential to preserve information for longer times. Here the fact that operations couple to quasi-particle excitations is less welcome: it  implies that encoding and measuring syndrome information corresponding to the `clean' system as envisioned here may be difficult or impossible. Ultimately, this depends on the system under consideration. Here we take the viewpoint that some form of error-correction will most likely be required in any realization of a stable quantum memory. This justifies the consideration of idealized error-correction procedures, but also poses the challenge of identifying systems where such operations may be practically realizable.

\subsubsection*{Outline}
In Section~\ref{sec:qwires}, we introduce
Kitaev's model~\cite{Kitaev00} of quantum wires with unpaired Majorana modes,
describe the error-correction procedure and state our main results.
In Sections~\ref{sec:gap},\ref{sec:QAC} we
prove two technical theorems needed for our analysis of the storage fidelity.
Theorem~\ref{thm:gap} provides a tight bound on the strength
of a perturbation capable of destroying the unpaired Majorana modes.
Theorem~\ref{thm:QAC} is a stronger version of the quasi-adiabatic continuation
method~\cite{Osborne07,HW05,Hastings10}
in which the effective Hamiltonian that governs the quasi-adiabatic evolution has a strictly exponential decay of interactions.
These theorems are used in Section~\ref{sec:disorderenhanced} to show that the storage time scales exponentially with the system size in the Anderson localization regime (Theorem~\ref{thm:storage}).
In Section~\ref{sec:dynamiclocalization}, we discuss
 localization properties of the  one-particle Hamiltonian and compute Lyapunov exponents.
Our simulation algorithm and numerical results are discussed in
Section~\ref{sec:numerics}.

\section{Quantum wires with unpaired Majorana fermions\label{sec:qwires}}

\subsection{Definition of the model}
\label{subs:themodel}
Following~\cite{Kitaev00}, we consider  a system of spinless electrons that live  on a chain
with $N$ sites. The electrons can hop between adjacent sites. In order for  unpaired Majorana modes to emerge,
the chain must be  put in contact with a bulk superconductor.
In this setup the number of electrons in the chain is only conserved modulo two due to the tunneling of Cooper pairs between the bulk superconductor and the chain.
This model can be described by the following Hamiltonian proposed in~\cite{Kitaev00}:
\be
\label{KitaevModel}
\hat{H}=\sum_{j=1}^N -\mu_j(a_j^\dag a_j - 1/2) + \sum_{j=1}^{N-1}  -w(a_j^\dag a_{j+1}+a_{j+1}^\dag a_j)  + \Delta a_j a_{j+1} + \Delta^* a_{j+1}^\dag a_j^\dag,
\ee
where $\mu_j$ is the chemical potential at a site $j$, $w$ is the hopping amplitude,
 and $\Delta$ is the superconducting pairing induced by the proximity effect.
Note that $\Delta$ is typically complex, $\Delta=|\Delta|e^{i\theta}$.
We will refer to the Hamiltonian~(\ref{KitaevModel}) as a {\em Majorana chain model}.
It will be convenient to rewrite $\hat{H}$ in terms of $2N$ Majorana
fermionic modes  $c_1,\ldots,c_{2N}$ by
\be
a_j=\frac{e^{-i\theta/2}}2\left( c_{2j-1} + i c_{2j} \right)
\quad \mbox{and} \quad
a_j^\dag=\frac{e^{i\theta/2}}2\left( c_{2j-1} - i c_{2j} \right),  \quad j=1,\ldots,N.
\ee
The Majorana operators are self-adjoint, $c_p^\dag=c_p$, and obey commutation rules
\be
c_p c_q + c_q c_p = 2\delta_{p,q} \hat{I}\ , \quad \quad c_p^2 = \hat{I}\ ,
\ee
where $\hat{I}$ is the identity. We can now rewrite the Hamiltonian Eq.~(\ref{KitaevModel}) as
\be
\hat{H}=\hat{H}_0 + \hat{V}\ ,\label{eq:hamiltoniandefinition}
\ee
where
\be
\label{H0}
\hat{H}_0= \frac{iJ}2 \sum_{j=1}^{N-1} c_{2j} c_{2j+1}\ , \quad \quad J\equiv \Delta+w
\ee
will be chosen as our `clean' memory Hamiltonian and
\be
\label{V}
\hat{V}=\frac{i(|\Delta|-w)}2 \sum_{j=1}^{N-1} c_{2j-1} c_{2j+2} - \frac{i}2 \sum_{j=1}^N \mu_j c_{2j-1} c_{2j}\ .
\ee
We will be mostly interested in the regime $|\Delta| \approx w \gg |\mu_j|$, so that $\hat{H}_0$
corresponds to the exact  equalities $w=|\Delta|$, $\mu_j=0$, while
$\hat{V}$ represents small deviations from this idealized setting.
It will always be assumed that the parameters $\Delta$ and $w$ are known, while
the chemical potential $\mu_j$  plays the role of an unknown perturbation.
The encoding and decoding steps should not depend on the realization of $\mu_j$.
In our simulations the chemical potential will be chosen as
\be
\label{mu}
\mu_j=\mu+\eta x_j, \quad \quad x_j\in [-1,1],
\ee
where $x_1,\ldots,x_N$ is a realization of the quenched disorder and
$\eta>0$ controls the disorder strength.
We will consider three scenarios (in all cases we assume $\mu>0$): \\

\begin{tabular}{|l|l|}
\hline
& \\
{\bf No disorder} & $\eta=0$. \\
& \\
{\bf Random disorder}  &  \parbox{10cm}{$\eta>0$ and
 $\{x_j\}$ are i.i.d. random variables
with uniform probability density.} \\
& \\
{\bf Pseudo-random disorder} & \parbox{10cm}{$\eta>0$ and
the sequence
$\{x_j\}$ is optimized to minimize the localization length.} \\
& \\
\hline
\end{tabular}\\
\\

In the first case (no disorder) the system is translation-invariant
and excitations of $\hat{H}$ can be described by Bloch waves\footnote{The analogue of Bloch waves
for a chain with boundary are superpositions of right-moving and left-moving waves.
The relative amplitudes of these waves and the admissible values of the momentum can be easily computed using the scattering method.}.
This clean case will provide us with baseline estimates of the storage time
needed to assess the effect  of adding randomness.
The second case (random disorder)  attempts to describe a random external electric potential, for example, the one created by impurities.
In this regime coherent propagation of quasiparticles is suppressed due to the Anderson localization
and a more favorable scaling of the storage time is to be expected.
Once it is established that Anderson localization is beneficial for the memory stability, a natural
question is whether it can be brought about on purpose by applying a
suitably engineered external potential. This is the main motivation for considering the third case
(pseudo-random disorder).
Here we assume that the constant part of the potential $\mu$ represents an unknown perturbation,
while the sequence $\{x_j\}$ represents a controlled external potential introduced on purpose in order to localize wavefunctions of the quasiparticles.   For simplicity we set $w=|\Delta|$ in the simulations, while our analytical results summarized in Section~\ref{subs:results} hold for any~$w$ and~$\Delta$.

\subsection{Encoding a qubit into the ground state}
\label{subs:encoding}

The Hamiltonian  $\hat{H}_0$ is a sum of pairwise-commuting terms
 formed by disjoint pairs of modes $c_{2j}c_{2j+1}$ which can be diagonalized simultaneously.
Every pair of modes $c_{2j}$, $c_{2j+1}$ can be formally combined to
a complex fermionic mode with an excitation energy $J$.
Ground states of $\hat{H}_0$ are $-1$ eigenvectors of every
term $ic_{2j}c_{2j+1}$ which can thus be regarded as `stabilizers'
of a quantum code.
The boundary Majorana modes $c_1$ and $c_{2N}$ play a special role since they do not
appear in $\hat{H}_0$ and commute with every term in $\hat{H}_0$.
Combined together, $c_1$ and $c_2$ form a complex boundary mode
whose occupation number operator is
\be
\hat{n}_b=b^\dag b, \quad b=\frac12 (c_1 + i c_{2N})\ .\label{eq:boundarymodeoperator}
\ee
This boundary mode has zero excitation energy since $[b,\hat{H}_0]=0$.
Let $|g_\sigma \ra$ be the ground state of $\hat{H}_0$
in which the boundary mode has occupation number $\hat{n}_b=\sigma$,
$\sigma=0,1$.

Although the ground state of $\hat{H}_0$ has two-fold degeneracy, it cannot  be used directly to encode
a qubit.  The reason is that $|g_0\ra$ and $|g_1\ra$ belong to different superselection sectors
defined by the fermionic parity operator
\be
\label{parity}
\hat{P}=(-1)^{\sum_{j=1}^N a_j^\dag a_j} = \prod_{j=1}^N (-i)c_{2j-1} c_{2j}\ .
\ee
Rewriting $\hat{P}$ as a product of $c_{2j} c_{2j+1}$ and $c_1 c_{2N}$ one can
easily check that~$\hat{P}\, |g_\sigma\ra = (-1)^{\sigma} |g_\sigma\ra$.
Since any physical state has either an even or odd number of fermions,
we conclude that coherent superpositions of $|g_0\ra$ and $|g_1\ra$ are unphysical.
We will avoid this problem by adding a {\em reference system} $R$
that contains one or several fermionic modes with some special states
$|0_R\ra$ and $|1_R\ra$ having even and odd fermionic parity. We can then encode a one-qubit state $\alpha |0\ra+\beta|1\ra$ into
a state
\be
\label{encoding}
|g\ra= \alpha |g_0\ra \otimes |0_R\ra + \beta |g_1\ra \otimes |1_R\ra
\ee
in which the total number of fermions is even. For simplicity we will assume that the
reference system has trivial dynamics, although our results can be generalized
 in a straightforward way to the case when the reference system is another Majorana chain
with the Hamiltonian Eq.~(\ref{KitaevModel}), that is, $|0_R\ra=|g_0\ra$  and $|1_R\ra=|g_1\ra$.
In this case we can choose logical Pauli operators on the encoded
qubit as $\overline{Z}=(-i)c_1 c_{2N}$ and $\overline{X}=(-i)c_1 \tilde{c}_1$, where
$\tilde{c}_1$ is the first Majorana mode on the chain representing the reference system.

\subsection{Error correction}
\label{subs:EC}

Let us now discuss how one can perform the
decoding, that is, how to retrieve the original
encoded state $|g\ra$, see Eq.~(\ref{encoding}),  from the  time-evolved state $|g(t)\ra=e^{i\hat{H}t}\, |g\ra$.
Note that we cannot simply reverse the time evolution since some
terms in $\hat{H}$ are unknown.

The important insight made in~\cite{Kitaev00} is that
the ground subspace spanned by $|g_0\ra$ and $|g_1\ra$
can be regarded as a quantum error correcting code protecting encoded
information against  local perturbations that involve  {\em even fermionic operators},
i.e., operators commuting with $\hat{P}$.
Let us first characterize the set of errors that can be corrected by this code.
Define {\em stabilizers}
\be
\label{stabilizers}
\hat{S}_j=(-i)c_{2j} c_{2j+1}\ , \quad j=1,\ldots,N-1\ ,
\ee
such that $\hat{H}_0=-\sum_{j=1}^{N-1} \hat{S}_j$ and  {\em elementary errors}
\be
E_j=(-1)^{a_j^\dag a_j} = (-i) c_{2j-1} c_{2j}\ , \quad j=1,\ldots,N
\ee
that are analogous to single-qubit errors in the standard error correction theory.
The unitary evolution $e^{i\hat{H}t}$ can be represented as a linear combination
of errors of a form $E=c_{j_1} \cdots c_{j_{2k}}$, where each error $E$
can be uniquely written as a product of stabilizers and elementary errors (up to an overall phase).
Any such error either commutes or anti-commutes with any stabilizer $\hat{S}_j$.
We will say that $E$ is an error of weight $v$ if it involves exactly $v$ elementary errors.
One can easily check that an error $E$ commutes with all stabilizers only if $v=0$ or $v=N$.
In the first case $E$ is a product of stabilizers and thus it acts trivially on any ground state
$|g\ra$. In the second case $E$ is a product of stabilizers and the parity operator $\hat{P}$,
see Eq.~(\ref{parity}), which can be regarded as the logical-$Z$ operator since
$\hat{P}|g_\sigma\ra=(-1)^\sigma |g_\sigma\ra$. Note that $\hat{P}$ is equivalent to $(-i)c_1 c_{2N}$
modulo stabilizers.

Error correction begins with the {\em syndrome measurement}. It involves a non-destructive eigenvalue measurement
for every stabilizer $\hat{S}_j$ which yields an eigenvalue $(-1)^{s_j}$,
 $s_j\in \{0,1\}$. Let $s\in \{0,1\}^{N-1}$ be the {\em syndrome}, that is, the list
of all measurement outcomes. The syndrome measurement
maps the state $|g(t)\ra$ to $\hat{Q}_s |g(t)\ra$, where
\begin{align}
\hat{Q}_s = \prod_{j=1}^{N-1} \frac12 \left(\hat{I}+(-1)^{s_j} \, \hat{S}_j\right)\label{eq:syndromemeasurement}
\end{align}
is the projector onto the subspace with syndrome~$s$.
(Note that the stabilizers $\hat{S}_j$ depend on the superconducting phase $\theta$, namely,
$\hat{S}_j=a_j^\dag a_{j+1} + a_{j+1}^\dag a_j -e^{i\theta} a_j a_{j+1} -e^{-i\theta} a_{j+1}^\dag a_j^\dag$. Hence the syndrome measurement can be realized only if the superconducting phase $\theta$ is known.)

We can now return the system back to the ground state using the syndrome information~$s$
by applying a correction operator $C(s)$. It will be chosen as a product of elementary errors
$E_j$ which is (a) consistent with the observed syndrome, and (b) uses as few elementary errors
as possible. The first condition demands that
\be
\label{C(s)}
C(s) \hat{S}_j =(-1)^{s_j} \hat{S}_j C(s)
\ee
for all $j=1,\ldots,N-1$. It guarantees that the corrected state $C(s)\hat{Q}_s |g(t)\ra$ is indeed
in the ground subspace of $\hat{H}_0$ since it is a~$+1$ eigenvector of every stabilizer~$\hat{S}_j$.
The second condition captures the intuition that the Hamiltonian
$\hat{H}$ involves only errors of small weight ($v=1$ and $v=2$) and thus
typical errors generated by the unitary evolution $e^{i\hat{H}t}$ have weight much smaller than $N$
(which may or may not be true depending on the evolution time $t$).
If $E'$ and $E''$ are products of elementary errors consistent with the syndrome $s$ then
the product $E'E''$ commutes with all stabilizers and thus either $E'=E''$ or $E'=E''\hat{P}=E''\prod_{j=1}^N E_j$
(modulo stabilizers). It follows that Eq.~(\ref{C(s)}) has only two solutions which use
mutually complementary subsets of elementary errors. The smallest weight solution $C(s)$
thus has weight at most $N/2$. We note that the error correction can only fail
by introducing a logical-$Z$ error since it uses only even fermionic operators.
In other words,  $C(s)\hat{Q}_s |g(t)\ra \sim (a_s \hat{I} + b_s \hat{P} )|g\ra$ for
any observed syndrome $s$ with some amplitudes $a_s,b_s\in \CC$.
Also note that in contrast to stabilizers, the elementary errors $E_j=(-1)^{a_j^\dag a_j}$
and the correction operators $C(s)$ do not depend
on the parameters of the Hamiltonian.

The procedure described above is analogous to the
minimum weight decoding for the 1D repetition code, see for instance~\cite{Dennis2002}.
If a syndrome $s$ is viewed as a collection of excitations, a candidate error $E$
that might have caused $s$ can be viewed as a pairing between the excitations.
Every excitation in~$s$ is  paired with some other excitation or with the boundary
by a chain of elementary errors $E_j$.
Since in the 1D geometry any set of excitations can be paired in only two different ways,
any syndrome~$s$ has  exactly two candidate errors  $E',E''$ consistent with
$s$ (modulo stabilizers) such that $E'E''\sim \hat{P}$ (modulo stabilizers).
The correction operator $C(s)$ is chosen  as the shortest error chain.

One can formally describe how the error correction acts on states by a TPCP map
\be
\label{ECmap}
\Phi_{ec}(\rho) = \sum_s C(s) \hat{Q}_s \rho \, \hat{Q}_s C(s)^\dag,
\ee
where the sum is over all $2^{N-1}$ syndromes $s$.
This gives the following expression for the storage fidelity with an initial state $|g\ra$
and a storage time $t$:
\be
F_{|g\ra}(t)=\la g|\Phi_{ec}(|g(t)\ra\la g(t)|)|g\ra = \sum_s |\la g|C(s) \hat{Q}_s e^{i\hat{H}t} |g\ra|^2\ .\label{eq:storagefidelityexplicit}
\ee
Since all operators $\hat{H},\hat{Q}_s$, and $C(s)$ preserve fermionic parity,
the reference system plays no role in the expression for $F_{|g\ra}(t)$. In particular,
one can formally compute $F_{|g\ra}(t)$ using an `unphysical' initial state
$|g\ra=\alpha |g_0\ra+\beta|g_1\ra$ instead of the one defined in Eq.~(\ref{encoding}).

It is worth pointing out that, as discussed in the introduction, error correction is not the only possible strategy to protect the encoded qubit from decoherence. In the special case of
the memory Hamiltonian and
the perturbation described by Eqs.~(\ref{H0},\ref{V})
the refocusing technique can be used instead. For example, applying a refocusing operator $C=c_1 c_3 \ldots c_{2N-1}$
exactly in the middle of the time evolution would reverse the sign of the Hamiltonian making the overall evolution trivial:
\[
C^\dag e^{i\hat{H}t/2} C e^{i\hat{H}t/2}=e^{-i\hat{H}t/2} e^{i\hat{H}t/2} = \hat{I}\ .
\]
The advantage of the error correction approach is that it applies to perturbations much more general than the one defined
in Eqs.~(\ref{V}), see below.

\subsection{Summary of main results}
\label{subs:results}

It will be convenient to state our results for perturbations more general than the one
defined in Eq.~(\ref{V}). Specifically, we will assume that
$\hat{H}_0=(i/2)\sum_{j=1}^{N-1} c_{2j} c_{2j+1}$
is the Hamiltonian defined in Eq.~(\ref{H0}) with $J=1$
and
\be
\label{Vpq}
\hat{V}=\frac{i}4 \sum_{p,q=1}^{2N} V_{p,q} \, c_p c_q.
\ee
Here $V$ is some real anti-symmetric matrix of size $2N$.
Let us say that a perturbation has {\em strength} $\epsilon$
and {\em range} $r$ iff
\be
\|V\|\le \epsilon \quad \mbox{and} \quad V_{p,q}=0 \quad \mbox{unless $|p-q|\le r$}\ .
\ee
For example, the perturbation $\hat{V}$ describing the Majorana chain, see Eq.~(\ref{V}),
has strength $\epsilon = \max_j |\mu_j| + |w-|\Delta||$ and range $r=3$.
Note that any physical perturbation $\hat{V}$ has strength $O(1)$ even though
the norm of $\|\hat{V}\|$ is typically an extensive quantity, that is,
$\|\hat{V}\|\sim N$.

Recall that the unperturbed Hamiltonian $\hat{H}_0$ has a two-fold degenerate
ground space separated from excited states by an energy gap $\Delta E=J=1$.
The degeneracy stems from two unpaired Majorana modes localized
on the left and the right boundaries of the chain.
Our first result provides a tight bound on the perturbation strength capable
of destroying these zero-energy boundary modes
and closing the gap above the ground state.
 This result can be regarded as a stronger version of the gap stability proved
for more general Hamiltonians with topological order
in~\cite{BravyiHastingsMichalakis01,BravyiHastings01}. Unfortunately,
the techniques used in the present paper only apply to non-interacting fermions.
\begin{theorem}[\bf Gap stability under perturbations]\label{thm:gap}
Consider any perturbation $\hat{V}$ with strength $\epsilon<1$ and
range~$r$.
Let $E_0^\uparrow \le E_1^\uparrow \le E_2^\uparrow$ be the three lowest eigenvalues of $\hat{H}_0+\hat{V}$.
Then
\be
\label{gap_stability}
E_1^\uparrow- E_0^\uparrow  \le \left(\frac{2\epsilon}{1+\epsilon} \right)^{\frac{N}{2r}} \quad \mbox{and} \quad
E_2^\uparrow- E_1^\uparrow \ge 1-\epsilon - \left(\frac{2\epsilon}{1+\epsilon} \right)^{\frac{N}{2r}}
\ee
for all large enough $N$.
\end{theorem}
Hence any short-range perturbation with strength $\epsilon<1$ preserves the two-fold ground state degeneracy of $\hat{H}_0$
up to exponentially small corrections. The degenerate ground state is separated from excited states by
an energy gap roughly $1-\epsilon$.
To see that a perturbation with strength $\epsilon=1$ is capable of closing the gap
$E_2^\uparrow- E_1^\uparrow$,  we can choose
 $\hat{V}=(i\epsilon/2) \sum_{j=1}^N c_{2j-1} c_{2j} =\epsilon\sum_{j=1}^N (a_j^\dag a_j-1/2)$.
In this  homogeneous case the gap of $\hat{H}_0+\hat{V}$
was calculated in~\cite{Kitaev00}, namely, $E_2^\uparrow-E_1^\uparrow=1-\epsilon$
for all $0<\epsilon\le 1$. In that sense the bound of Theorem~\ref{thm:gap} is tight.
Let us emphasize  that the energy splitting $E_1^\uparrow- E_0^\uparrow$  for the ground state
is bounded by a pure exponential function of $N$. In contrast,
more general bounds on the energy splitting proved in~\cite{BravyiHastingsMichalakis01,BravyiHastings01}
decay slower than exponentially.
The proof of Theorem~\ref{thm:gap} also provides
an explicit construction of the unpaired Majorana modes corresponding
to the perturbed Hamiltonian $\hat{H}_0+\hat{V}$. These modes
are linear combinations of the Majorana operators~$c_j$
which are exponentially localized on the left and on the right boundaries of the chain

Our second result is a lower bound on the worst-case storage fidelity
\be
\label{Fgeneral}
F(t)=\min_{|g\ra}  \; \la g| \Phi_{ec}(e^{i(\hat{H}_0+\hat{V})t } |g\ra\la g|  e^{-i(\hat{H}_0+\hat{V})t})|g\ra\ ,
\ee
or more precisely, its expectation $\ExpE{F(t)}$ over disorder realizations. We give a dynamical localization condition for the one-particle problem which, when satisfied, implies that  $\ExpE{F(t)}$ remains close to~$1$ for times~$t$ exponentially large in the system size.

Define an orthogonal matrix $R(t)\in SO(2N)$ describing evolution of the Majorana
operators $c_p$ in the Heisenberg picture,
\be
c_p(t)\equiv e^{i(\hat{H}_0+\hat{V})t} c_p e^{-i(\hat{H}_0+\hat{V})t}=\sum_{q=1}^{2N} R_{p,q}(t) c_q\ .\label{eq:cpheisenberg}
\ee
In the case when $\hat{V}$ includes random disorder,
the theory of Anderson localization roughly predicts that the time-evolved
operator~$c_p(t)$ is well-localized near the original mode~$c_p$ for all $t\in \RR$ with high probability over the disorder distribution. This is equivalent to the so called {\em dynamical localization} condition: the matrix of
expectation values $\ExpE{|R_{p,q}(t)|}$ decays rapidly away from the main diagonal (see e.g., ~\cite{Stolz11}).
Our bound on the storage fidelity requires dynamical localization only at large scales, that is,
$|p-q|\sim N$. Unfortunately, we also need to impose  constraints on the decay of certain
multi-point correlation functions composed of matrix elements of $R\equiv R(t)$.
More precisely,  suppose ${p},{q}$ are $m$-tuples of integers in the interval $[1,2N]$
such that  $p_1<p_2<\ldots <p_m$ and $q_1<q_2<\ldots <q_m$. Define a distance between $p$ and $q$ as
\[
|{p}-{q}|_1=\sum_{a=1}^m |p_a -q_a|\ .
\]
We denote by $R[{p},{q}]$ the $m\times m$ submatrix of $R$
obtained by retaining only the rows indexed by ${p}$ and the columns indexed by ${q}$, i.e., $R[{p},{q}]_{a,b}=R_{p_a,q_b}$ for all $1\leq a,b\leq m$.
\begin{dfn}[\bf Multi-point dynamical localization]\label{def:DL}
The unitary evolution operator $e^{i(\hat{H}_0+\hat{V})t}$ exhibits multi-point dynamical localization   iff
there exist constants $C,\xi>0$ such that for all sufficiently large $N$, for all $m\ge 1$,
and for all ordered $m$-tuples ${p}$, ${q}$ such that $|{p}-{q}|_1 \ge N/8$ one has
\be
\label{DL}
\ExpE{\left|\det R[{p},{q}](t)\right|} \le C^m e^{-\textfrac{N}{\xi}}.
\ee
Here the expectation value is taken over disorder realizations.
\end{dfn}
\noindent
We will refer to the constant $\xi$ as the (dynamical) {\em localization length}.
In general, $\xi$ does not coincide with the localization length $\xi_s$ that controls
decay of eigenfunctions (the spectral localization). The relationship between $\xi$ and $\xi_s$
is discussed in more detail in Section~\ref{sec:dynamiclocalization}.
For $m=1$ one has $\det R[{p},{q}]=R_{p,q}$ and the multi-point localization condition reduces to the usual (two-point) dynamical localization, although we only need it for large scales, $|p-q|\sim N$.
In Section~\ref{sec:momentlocalization} we discuss how to replace the
expectation value of the determinant in  Eq.~\eqref{DL} by more standard multi-point correlation   functions.
Let us emphasize that in the present paper we {\em do not} prove Eq.~(\ref{DL}).
Instead, our goal is to establish a link between dynamical  localization and the storage fidelity.
It is provided by the following theorem.
\begin{theorem}[\bf Storage fidelity]\label{thm:storage}
Consider any random ensemble of perturbations $\hat{V}$ with strength $\epsilon<1/4$ and
range~$r$. Consider any time~$t$ such that
the unitary evolution $e^{i(\hat{H}_0+\hat{V})t}$ obeys
the multi-point dynamical localization condition
with some localization length $\xi$ and constant $C$.
Then the expected storage fidelity can be bounded as
\be
\label{lower_bound}
\ExpE{F(t)}\ge 1-t^2 \cdot e^{-\frac{N}{r}}  - O(1) \cdot e^{-\nu  N}
\ee
for all sufficiently large $N$, where
\[
\nu \equiv  \nu(\epsilon,C,\xi,r)= \frac1{\max{(\xi, 16r,  280)}} - \alpha r^3 \sqrt{\epsilon} -  \beta C \epsilon^{\frac1{16}},
\]
and $\alpha,\beta>0$ are some constants that depend on details of the proof.
\end{theorem}
\noindent
We conclude that for a fixed fidelity threshold $F_0<1$, the storage time $T_{\textrm{storage}}(F_0)$ grows exponentially with~$N$ if the ensemble of Hamiltonians satisfies multi-point dynamical localization with time-independent constants $C,\xi$ such that $\nu>0$. 
Computing the parameters $\xi, C$ and checking
whether $\nu>0$  obviously requires more detailed information about the perturbation
$\hat{V}$ and the disorder distribution.  In the special case of the Majorana chain model
we have $\epsilon\leq \mu+\eta$, where $\mu$ and $\eta$ control the strength of the homogeneous and the random
parts of the perturbation, see Eq.~(\ref{mu}).
In Section~\ref{sec:lyapunov}, we provide some evidence indicating that the regime $\nu>0$
can be achieved for $\mu\sim \eta \ll 1$ by computing the Lyapunov exponent~$\ell$ of the one-particle eigenfunctions. With the technique developed by Eggarter et al~\cite{eggarterried78}
for analyzing tight-binding chains with an off-diagonal disorder we show that  $\ell^{-1} \sim \log{(1/\eta)}$ in the limit $\eta \to 0$ when the ratio~$\mu/\eta$ is kept constant. Using the rough estimate~$\xi\sim\xi_s\sim\ell^{-1}$ and assuming~$C=O(1)$ in this limit, we would get
\[
\nu(\epsilon,C,\xi,r) \approx \frac1{\xi} - \beta C \epsilon^{\frac1{16}} = \frac{\Omega(1)}{\log{(1/\eta)}} - O(1)\cdot \eta^{\frac1{16}} >0 \]
for small enough $\eta$.
Another scenario where the condition $\nu>0$ could be satisfied is
when some randomness is present in the coefficients of the unperturbed Hamiltonian $\hat{H}_0$
so that the disorder strength and the perturbation strength are independent parameters.
In this scenario we expect bounds of the form~$\xi,C=O(1)$ in the limit $\epsilon\to 0$, so that
$\nu>0$ for sufficiently small $\epsilon$. Our proof of Theorem~\ref{thm:storage}
can be easily extended to this  setting.

Our last theorem is a technical tool needed to prove
Theorem~\ref{thm:storage}, although it might be interesting on its own right.
It provides a stronger version of the exact quasi-adiabatic continuation~\cite{Osborne07,HW05,Hastings10},
that is,  a unitary operator mapping the ground subspace of $\hat{H}_0$ to the one of $\hat{H}_0+\hat{V}$.
\begin{theorem}[\bf Quasi-adiabatic continuation]\label{thm:QAC}
Consider any perturbation $\hat{V}$
with strength $\epsilon<1/2$ and
range~$r$.
Let $|g_\sigma\ra$ and $|\tilde{g}_\sigma\ra$ be the ground states of
$\hat{H}_0$ and $\hat{H}_0+\hat{V}$ in the sector with fermionic parity~$\sigma\in \{0,1\}$. Then~$|{g}_\sigma\ra= \hat{U}\cdot |\tilde{g}_\sigma\ra$,
where $\hat{U}$ is a unitary operator
describing evolution under a time-dependent quadratic
Hamiltonian $\hat{D}(u)$, $u\in [0,1]$, with strength roughly~$\epsilon$.
More precisely,
\be
\label{qac1}
\hat{U}=\calT \cdot \exp{\left[ i \int_0^1 du \hat{D}(u) \right]}\ ,
\ee
where
\be
\label{qac2}
\hat{D}(u)=\frac{i}4 \sum_{p,q=1}^{2N} D_{p,q}(u) \, c_p c_q
\ee
for some real anti-symmetric matrix $D(u)$ such that
\be
\label{Dnorm}
\|D(u)\| \le \frac{3\epsilon}{1-2\epsilon}
\ee
for all $u\in [0,1]$ and all $N$. Furthermore, if $\epsilon<1/4$,
the Hamiltonian $\hat{D}(u)$ is exponentially decaying:
\be
\label{Ddecay}
|D_{p,q}(u)|\le c\epsilon r^2 (2/3)^{\frac{|p-q|}{4r}} \quad \mbox{for all $p,q$}\ .
\ee
The constant coefficient $c=O(1)$ depends on details of the proof.
\end{theorem}
Note that this theorem requires a stronger bound on $\epsilon$
although we believe that the exponential decay of  $D_{p,q}(u)$ holds for all $\epsilon<1$.
Our proof yields a rather large constant coefficient $c=2\times 10^6$ since we have not tried to optimize it.
Let us also note  that a weaker version of Theorem~\ref{thm:QAC}  in which
$|D_{p,q}(u)|$ has  stretched exponential decay
can be easily derived using the techniques of~\cite{Osborne07,HW05,Hastings10}.
Unfortunately, any bound in Eq.~(\ref{Ddecay}) decaying slower than exponentially
is not sufficient for our purposes since it results in an error distribution in which
uncorrectable high-weight errors are not sufficiently suppressed.

A simple corollary of Theorem~\ref{thm:QAC} is that
 the overlap $|\la \tilde{g}_\sigma|g_\sigma\ra|$ between
the unperturbed and perturbed ground states is close to $1$
as long as $N\ll \epsilon^{-2}$.
\begin{corol}[{\bf Perturbed versus unperturbed ground states}]
\label{corol:boring1}
Consider any perturbation $\hat{V}$
with strength $\epsilon<1/8$.
Let $|g_\sigma\ra$ and $|\tilde{g}_\sigma\ra$ be the ground states of~$\hat{H}_0$ and~$\hat{H}_0+\hat{V}$ with fermionic parity~$\sigma\in\{0,1\}$. Then
\be
\label{boring1}
|\spr{\tilde{g}_\sigma}{g_\sigma}|^2 \geq 1-4N \epsilon^2
\ee
for all $\sigma=0,1$ and for all~$N$.
\end{corol}
This also implies that error correction
is not really necessary
in the regime $N\ll \epsilon^{-2}$, since the initial encoded state $|g\ra$
is very close to the ground state of $\hat{H}_0+\hat{V}$. As a consequence, the unitary evolution $e^{i(\hat{H}_0+\hat{V})t}$ has
no effect on $|g\ra$ except for an exponentially small dephasing
between $|g_0\ra$ and $|g_1\ra$.
For simplicity, below we consider the maximally entangled encoded
state, $|g\ra=(|g_0\ra \otimes |0_R\ra + |g_1\ra \otimes |1_R\ra)/\sqrt{2}$,
see Section~\ref{subs:encoding}.
 The quantity~$F_{\ket{g}}(t)$ is the entanglement fidelity~\cite{Horodeckietal99} of the channel composed of storage, time-evolution and error-correction.
\begin{corol}[\bf Cosine law]
\label{corol:boring2}
Consider any perturbation $\hat{V}$
with strength $\epsilon<1/8$ and
let $\delta=E_1^\uparrow-E_0^\uparrow$
be the energy splitting of the ground state, see Theorem~\ref{thm:gap}. Then
\be
\label{boring2}
|F_{|g\ra}(t) -\cos^2{(\delta t/2)} |\le  8{\epsilon}\sqrt{N}\ .
\ee
This bound applies even if no error correction is performed, that is, if $\Phi_{ec}$ is
the identity map.
\end{corol}
\noindent We conclude that the interesting region of parameters in which
 error correction really matters is $N\gtrsim \epsilon^{-2}$.

Finally, let us briefly summarize the results of numerical simulations
performed for the Majorana chain model with $w=|\Delta|$
and a random site-dependent chemical potential $\mu_j=\mu + \eta x_j$,
where
$x_j\in [-1,1]$ are i.i.d. random variables, see Section~\ref{sec:numerics} for details. The corresponding perturbation strength
is $\epsilon\leq \mu+\eta$.
We computed the storage fidelity and the storage time in the following three regimes:
\begin{description}
\item[Weak perturbations ($\epsilon \ll 1$):]
In this regime we can only perform the simulation for $N\ll 1/\epsilon^2$.
Here our numerical results agree with Corollary~\ref{corol:boring2}: the storage time is solely
determined by the dephasing of the ground states and grows exponentially with the system size,~$T_{\textrm{storage}}\sim 1/\delta=\exp(\Omega(N))$,
where~$\delta$ is the ground state energy splitting, see~Theorem~\ref{thm:gap}.
Disorder does not play any role here.

\item[Strong perturbations ($\epsilon\sim 1$), no disorder:] Here we observe a  logarithmic scaling of the storage time,
$T_{\textrm{storage}} \sim \log{(N)}$,   for $4\le N\le 256$. This behavior can be understood by mapping the time evolution of the Majorana chain model to a quantum quench in the transverse-field Ising chain
and using known results on the decay of magnetization in the latter model, see Section~\ref{sec:numerics}.

\item[Strong perturbations with disorder:]  We observe an approximately linear scaling of the storage time,
$\ExpE{T_{\textrm{storage}}} \sim N$, for $4\le N\le 128$.
The localization length $\xi_s$ is comparable with the system size for this range of $N$.  The storage times are clearly enhanced compared to the clean case. At the same time, we expect that the exponential scaling of the storage
time cannot be achieved for strong perturbations, $\epsilon\sim 1$, because $\nu<0$ in Theorem~\ref{thm:storage}.
The storage fidelity exhibits strong fluctuations as a function of disorder; no self-averaging behavior is observed.
\end{description}
Surprisingly, we also found that the random i.i.d. disorder potential $\{x_j\}$ is not the optimal one as far
as the storage time is concerned. The best performance was numerically observed for
a disorder potential generated by iterations of the logistic map, namely, $x_j=1-2y_j$, where $y_{j+1}=ay_j(1-y_j)$ for $1\le j\le N-1$.
Such disorder potentials are uniquely determined by the initial condition $y_1\in [0,1]$ and the parameter $a$.
Note that $y_j\in [0,1]$ for all $j$ as long as $0\le a\le 4$. Sequences generated by iterations of the logistic map are known to exhibit chaotic behavior for~$3.57\lesssim a< 4$ and almost all initial conditions.
For example, choosing $a=3.9914,y_1=0.2845$ we were able to boost the storage time from~$\ExpE{T_{\textrm{storage}}}\approx 31$ to $T_{\textrm{storage}}\approx 176$ for $N=64$.

\section{Unpaired Majorana modes are stable under perturbations}
\label{sec:gap}
The proof of Theorem~\ref{thm:gap} is presented in Section~\ref{sec:thmgapproof}.  We begin with a summary of the relevant connections between the one-particle and the many-particle problem in Section~\ref{sec:quadraticfermionhamiltonians}.

\subsection{Normal form of quadratic fermionic Hamiltonians\label{sec:quadraticfermionhamiltonians}}
We will be using the following well-known fact throughout the paper.
\begin{fact}[{\bf Williamson eigenvalues}]
\label{fact:1}
Let $A$ be any real anti-symmetric matrix of size $2N$.
Then there exist non-negative real numbers
$0\le \lambda_1 \le \ldots \le \lambda_N$
called Williamson eigenvalues of $A$
and complex vectors $\psi_1,\ldots,\psi_N\in \CC^{2N}$ such that
\[
A \psi_j =i\lambda_j \psi_j \quad \mbox{and} \quad A\bar{\psi}_j = -i\lambda_j \bar{\psi}_j
\]
where $\bar{\psi}_j$ is the complex conjugate of $\psi_j$. In addition,
$\la \psi_j|\psi_k\ra=\delta_{j,k}$ and
$\la \bar{\psi}_j |\psi_k\ra=0$.
\end{fact}
Consider a Hamiltonian
\begin{align*}
\hat{H}&=\frac{i}{4}\sum_{p,q=1}^{2N} H_{p,q}c_pc_q\ ,
\end{align*}
where $H$ is an anti-symmetric real matrix. Then the Williamson eigenvalues $0\leq \lambda_1\leq\ldots\leq \lambda_N$ of~$H$ are the single-particle excitation energies of~$\hat{H}$. To see this explicitly, let $\{\psi_j\}_j$ be the corresponding complex (normalized) eigenvectors, and let $\{\ket{p}\}$ be the computational basis of~$\mathbb{C}^{2N}$. Fact~\ref{fact:1} implies that the operators  $\{\hat{a}_j\}_{j=1}^N$ defined by
\begin{align*}
\hat{a}_j&=\frac{1}{\sqrt{2}}\sum_{p=1}^{2N}\langle p|\psi_j\rangle c_p\ ,
\end{align*}
obey the fermionic CCRs, that is, we can regard them as annihilation operators of some fermionic modes. A simple calculation shows that
\begin{align*}
[\hat{H},\hat{a}_j]&=-\lambda_j\hat{a}_j\ ,
\end{align*}
hence the $\{\hat{a}_j\}$ are canonical (complex) modes of $\hat{H}$. In particular, if $E^\uparrow_0$ is the ground state energy of $\hat{H}$, then
\begin{align}
\hat{H}&=E^\uparrow_0+\sum_{j=1}^N \lambda_j \hat{a}_j^\dagger\hat{a}_j\ .\label{eq:Hamiltonianasdirac}
\end{align}
Alternatively, we can write the Hamiltonian in terms of Majorana modes $\{\tilde{c}_p\}_{p=1}^{2N}$ as
\begin{align}
\hat{H}&=\frac{i}{2}\sum_{j=1}^N \lambda_j \tilde{c}_{2j-1}\tilde{c}_{2j}\qquad\textrm{ where } \tilde{c}_{2j-1}=\hat{a}_j+\hat{a}_j^\dagger\textrm{ and }\tilde{c}_{2j}=(-i)(\hat{a}_j-\hat{a}_j^\dagger)\ .\label{eq:Hamiltonianasmajorana}
\end{align}
The ground state energy $E_0^\uparrow=-\frac{1}{2}\sum_{j=1}^{N}\lambda_j$ can be found, e.g., by comparing~\eqref{eq:Hamiltonianasdirac} with~\eqref{eq:Hamiltonianasmajorana}.

\subsection{Stability analysis based on the Brillouin-Wigner method}\label{sec:thmgapproof}
Define anti-symmetric real matrices $H_0$ and $V$ such that
\be
\hat{H}_0=\frac{i}4 \sum_{p,q=1}^{2N} (H_0)_{p,q} \, c_p c_q
\quad
\mbox{and}
\quad
\hat{V} =\frac{i}4 \sum_{p,q=1}^{2N} V_{p,q} \, c_p c_q. \label{eq:antisymmetricmatrices}
\ee
Let $0\le \lambda_1(u)\le \ldots \le \lambda_N(u)$ be the Williamson eigenvalues of $H_0+uV$,
where $u\in [0,1]$.
Equivalently, $\lambda_j(u)$ are single-particle excitations energies of
the Hamiltonian $\hat{H}_0+u\hat{V}$.
Our choice of $\hat{H}_0$ implies that
\be
\label{u=0}
\lambda_1(0)=0 \quad \mbox{and} \quad  \lambda_j(0)=1 \quad \mbox{for all  $j=2,\ldots,N$}\ .
\ee
Since $\|V\|\le \epsilon$,
 Weyl's inequality~\cite{Weyl12} yields
\be
\label{Weyl}
\lambda_1(u)\le \epsilon \quad \mbox{and} \quad
\lambda_j(u)\ge 1-\epsilon \quad \mbox{for all $j=2,\ldots,N$}\ .
\ee
In the rest of the proof we focus on the (normalized) eigenvector $\psi_1\equiv \psi_1(u)$ such that
\be
\label{psi1}
(H_0+u V) \psi_1=i\lambda_1(u) \psi_1, \quad \la \psi_1|\psi_1\ra=1, \quad u\in [0,1].
\ee
Define
\be
\label{u^*}
u^*=\max_{u\in [0,1]}\, u \quad \quad  \mbox{subject to}
\quad
\lambda_1(v) \le \frac{1-\epsilon}2\quad\textrm{ for all }v\in [0,u]\ .
\ee
Continuity of $\lambda_1(u)$ implies that either $u^*=1$ or $u^*\in (0,1)$ in which case
$\lambda_1(u^*)=(1-\epsilon)/2$. We will show that the second case actually leads to a contradiction for sufficiently large $N$
(this trick is borrowed from~\cite{BravyiHastings01}).
From now on we will only consider $u\in [0,u^*]$, so with Eq.~(\ref{Weyl}), we have
\be
\label{Weyl+trick}
\lambda_1(u)\le \frac12(1-\epsilon) \quad \mbox{and} \quad
\lambda_j(u)\ge 1-\epsilon \quad \mbox{for all $j=2,\ldots,N$}\ .
\ee
We will now explicitly compute $\psi_1$  using  Brillouin-Wigner perturbation theory.
Let $\{ |p\ra \}$, $p=1,\ldots,2N$, be the standard orthonormal basis of $\CC^{2N}$
in which $H_0$ and $V$ are real and anti-symmetric.
Define operators
\begin{align*}
Q=I-|1\ra\la 1|  - |2N\ra\la 2N| \quad \mbox{and} \quad G(u)=(i\lambda_1(u) - H_0)^{-1} Q.
\end{align*}
Note that $G(u)$ is well-defined because of Eqs.~(\ref{u=0},\ref{Weyl+trick}).
Representing $\psi_1 =Q \psi_1 + \alpha_L |1\ra  + \alpha_R|2N\ra$ for some
$\alpha_L,\alpha_R \in \CC$ and rewriting Eq.~(\ref{psi1}) as  $Q(i\lambda_1(u) - H_0) \psi_1= uQV \psi_1$
we obtain
\be
Q\psi_1=uG(u) V \psi_1\label{BWfirst}
\ee
and
\be
\label{BW}
\psi_1=(I-u G(u) V)^{-1} (\alpha_L |1\ra  + \alpha_R |2N\ra).
\ee
Taking into account that $\|V\|\le \epsilon$ and using Eqs.~(\ref{u=0},\ref{Weyl+trick})  one gets
\be
\label{BWbound}
\| u G(u) V \| \le \| G(u) V\| \le \frac{\epsilon}{|\lambda_1(u)-1|} \le \frac{2\epsilon}{1+\epsilon} \equiv \eta.
\ee
Note that $\eta<1$ as long as $\epsilon<1$. With~\eqref{BWfirst}, this implies that
$\|\alpha_L\ket{1}+\alpha_R\ket{2N}\|=\|(I-Q)\psi_1\|\geq (1-\eta)\|\psi_1\|> 0$, hence the expression on the rhs.~of~\eqref{BW} is non-zero and well-defined because of~\eqref{BWbound}.
We can now use the fact that $G(u)$ and $V$ have support only near the main diagonal
to show that $\psi_1$ is a sum of two exponentially decaying states
localized near the two boundaries of the chain. We will need the following simple fact.
\begin{prop}
\label{prop:res}
Let $X$ be any matrix such that $\|X\| <\eta$ for some $0<\eta<1$
and $X_{p,q}=0$ unless $|p-q|\le r$. Then
\be
|(I-X)^{-1}_{p,q}| \le (1-\eta)^{-1} \eta^{\frac{|p-q|}{r}}\quad \mbox{for all $p,q$}\ .
\ee
\end{prop}
\begin{proof}
Indeed, let $n_0$ be the largest integer such that $r n_0\le |p-q|$. Then
$(X^n)_{p,q}=0$ unless $n\ge n_0$. It follows that
\[
|(I-X)^{-1}_{p,q}| \le  \sum_{n=n_0}^\infty |(X^n)_{p,q}| \le  \sum_{n=n_0}^\infty \|X\|^n \le
\frac{\eta^{n_0}}{1-\eta}\ .
\]
\end{proof}
Represent $\psi_1$ as
\be
\label{left&right}
\psi_1=\beta_L \psi^L + \beta_R \psi^R, \quad |\beta_L|^2 + |\beta_R|^2 =1,
\ee
where $\psi^L$ and $\psi^R$ are normalized states defined as
\be
\label{psiLR}
\psi^L\sim (I-uG(u) V)^{-1}\, |1\ra, \quad \psi^R\sim (I-uG(u) V)^{-1}\, |2N\ra,
\quad \la \psi^L|\psi^L\ra=\la \psi^R|\psi^R\ra=1.
\ee
(All  states and coefficients above depend on $u$.)
Proposition~\ref{prop:res} and Eq.~(\ref{BWbound}) imply that $\psi^L$ and $\psi^R$ decay exponentially away from  the left and  right boundary, respectively:
\be
\label{zero-mode-decay}
|\la p| \psi^L\ra | \le (1-\eta)^{-1} \eta^{\frac{p-1}{r}} \quad \mbox{and} \quad
|\la p| \psi^R\ra | \le  (1-\eta)^{-1} \eta^{\frac{2N-p}{r}}\ .
\ee
Here we used the fact that the normalizing factors in~\eqref{psiLR} are upper bounded by~$1$. This follows using $\la 1|G(u)=0$
which implies $\|(I-uG(u) V)^{-1}\, |1\ra\|^2\geq |\la 1|(I-uG(u) V)^{-1}|1\ra|^2=1$ and similarly for $(I-uG(u) V)^{-1}\,\ket{2N}$.

We would like to show that the states $\psi^L$, $\psi^R$ can be made
{\em real} by choosing their overall phase (up to exponentially small corrections).
Since many equations below include equalities that hold with exponentially small
corrections, let us set up a special notation. Given  numbers
$x,y \in \CC$, we will say that $x$ exp-equals $y$
and write $x \expeq y$
 iff $|x-y|\le \eta^{N/\gamma r}$
for some constant $1<\gamma<2$ and for all sufficiently large $N$.
The same notation $X \expeq Y$ will be applied to operators $X,Y$ acting on $\CC^{2N}$
meaning that $\|X-Y\|\le \eta^{N/\gamma r}$
for some constant $1<\gamma<2$ and for all sufficiently large $N$.

Define a projector $P_0\equiv P_0(u)$ as
\be
\label{P_0(u)}
P_0=|\psi_1\ra\la \psi_1| + |\bar{\psi}_1\ra\la \bar{\psi}_1|.
\ee
The spectral bounds Eq.~(\ref{Weyl+trick}) imply that $P_0(u)$ has rank $2$ and
depends smoothly on $u$. In addition, Eq.~(\ref{zero-mode-decay}) implies
that matrix elements of $P_0$ decay exponentially as one deviates
from the top-left and the bottom-right corner.
Let $\Pi^L$ and $\Pi^R$ be diagonal projectors
that project onto subspaces
spanned by the first $N$ and the last $N$ basis vectors  respectively
(the left and the right halves of the chain). By definition, $\Pi^L+\Pi^R=I$.
Define also $P_0^L=\Pi^L P_0 \Pi^L$ and $P_0^R=\Pi^R P_0 \Pi^R$
(both operators depend on $u$).
Using Eq.~(\ref{zero-mode-decay}) one easily gets
\be
\Pi^L P_0 \Pi^R \expeq 0.
\ee
This implies that
\be
P_0 \expeq P_0^L + P_0^R.\label{eq:pzeroblock}
\ee
Taking into account that $\|X^2-Y^2\| \le 2\|X-Y\|$ for any operators $X$, $Y$ whose
norm in bounded by~$1$, we conclude that
\be
(P_0^L)^2 \expeq P_0^L  \quad \mbox{and} \quad
(P_0^R)^2 \expeq P_0^R,
\ee
that is,  $P_0^L$ and $P_0^R$ are exponentially close to projectors.
Furthermore, $P_0^L(0)$ and $P_0^R(0)$ are rank-$1$ projectors onto
site~$1$ and~$2N$, respectively. Continuity then implies that
$P_0^L(u)$ and $P_0^R(u)$ are exponentially close to rank-$1$ projectors
for all $u\in [0,u^*]$. On the other hand,
using Eqs.~(\ref{left&right},\ref{P_0(u)})
and the exponential decay of $\psi^L$, $\psi^R$, see Eq.~(\ref{zero-mode-decay}) one gets
\be
P_0^L \expeq |\beta_L|^2 ( |\psi^L\ra\la \psi^L |  + |\bar{\psi}^L\ra\la \bar{\psi}^L |)
\quad
\mbox{and}
\quad
P_0^R \expeq |\beta_R|^2 ( |\psi^R\ra\la \psi^R |  + |\bar{\psi}^R\ra\la \bar{\psi}^R |). \label{eq:pzerolrdef}
\ee
This is possible only if
\be
\label{bm_aux}
|\la \bar{\psi}^L|\psi^L\ra|\expeq |\la \bar{\psi}^R|\psi^R\ra| \expeq 1,
\quad \mbox{and} \quad
|\beta_L|\expeq |\beta_R| \expeq \frac1{\sqrt{2}}\ .
\ee
Let us choose the overall phase of $\psi^L$ and $\psi^R$ to make them real
(up to exponentially small corrections), that is,
$\la\bar{\psi}^L|\psi^L\ra \expeq 1$ and $\la \bar{\psi}^R|\psi^R\ra \expeq 1$.

Computing $i\lambda_1(u)$ as an expectation value $\la \psi_1|H_0+uV|\psi_1\ra$
and taking into account that $\psi^L,\psi^R$ are (approximately) real, while $H_0+uV$
is a real anti-symmetric matrix, one arrives at
\be
\label{final-bound}
i\lambda_1 \expeq \bar{\beta}_L  \beta_R \la \psi^L|H_0+uV|\psi^R\ra + \beta_L \bar{\beta}_R
 \la \psi^R|H_0+uV|\psi^L\ra \expeq 0,
\ee
since the overlap between $\psi^L$ and $\psi^R$ is exponentially small, see Eq.~(\ref{zero-mode-decay}).
Recall that this bound holds for all $u\in [0,u^*]$.
Suppose now that $u^*<1$. Then
\be
\lambda_1(u^*)=\frac{1-\epsilon}2
\ee
see Eq.~(\ref{u^*}). But this would contradict to
Eq.~(\ref{final-bound}) for sufficiently large $N$.
Hence $u^*=1$, that is, Eq.~(\ref{final-bound}) holds for $u=1$.

The three lowest eigenvalues of $\hat{H}_0+\hat{V}$ can be expressed in terms of the
Williamson eigenvalues of $H+V$  as
\be
E_1^\uparrow - E_0^\uparrow = \lambda_1(1)\expeq 0  \quad \mbox{and} \quad
E_2^\uparrow - E_0^\uparrow = \lambda_2(1).
\ee
Weyl's inequality implies that $\lambda_2(1)\ge 1-\epsilon$, see Eq.~(\ref{Weyl}).
Hence
\[
E_2^\uparrow -E_1^\uparrow=\lambda_2(1)-\lambda_1(1) \ge 1-\epsilon-\eta^{N/2r}\ .
\]
This completes the proof of Theorem~\ref{thm:gap}.

Finally, let us remark that the real normalized vectors $\psi^L$ and $\psi^R$ determine
zero-energy  unpaired Majorana modes $c^L$ and $c^R$ localized on the left and the right boundaries by
\[
c^L\expeq \sum_{p=1}^{2N} \la p|\psi^L\ra \, c_p \quad
\mbox{and}
\quad
c^R\expeq \sum_{p=1}^{2N} \la p|\psi^R\ra \, c_p\ ,
\]
that is, $c^L c^R + c^R c^L \expeq 0$, $[c^L,\hat{H}_0 +\hat{V}]\expeq 0$,
and $[c^R,\hat{H}_0 +\hat{V}]\expeq 0$ (cf.~\eqref{eq:Hamiltonianasmajorana}).

\section{Quasi-adiabatic continuation}
\label{sec:QAC}

In this section we prove Theorem~\ref{thm:QAC}
and Corollaries~\ref{corol:boring1},\ref{corol:boring2}.
For any $u\in [0,1]$ let
$0\le \lambda_1(u)\le \ldots \le \lambda_N(u)$
be the Williamson eigenvalues of $H_0+u V$, see Fact~\ref{fact:1}.
The corresponding normalized eigenvectors satisfy $(H_0+u V)\psi_j=i\lambda_j(u) \psi_j$.
Theorem~\ref{thm:gap} implies that $\lambda_1(u)$ is exponentially small while
$\lambda_j(u)\gtrsim 1-\epsilon$ for $j\ge 2$ with exponentially small corrections.
Define projectors
\[
P_+=\sum_{j=2}^N |\psi_j \ra\la \psi_j|,
\quad
P_-=\bar{P}_+=\sum_{j=2}^N |\bar{\psi}_j \ra\la \bar{\psi}_j|
 \quad \mbox{and} \quad P_0=|\psi_1\ra\la \psi_1| + |\bar{\psi}_1\ra\la \bar{\psi}_1|.
\]
Here all states and operators are $u$-dependent.
Fact~\ref{fact:1} implies that these projectors form an orthogonal decomposition of the identity,
$P_0+P_+ + P_- =I$.
\begin{prop}
\label{prop:map}
Let~$|g_\sigma(u)\ra$ be the ground state of $\hat{H}_0+u\hat{V}$
in the sector with fermionic parity~$\sigma\in\{0,1\}$.
Suppose an orthogonal matrix $R\in SO(2N)$ satisfies
$R P_+(0) R^T = P_+(u)$
for some $u\in [0,1]$. Consider any unitary operator $\hat{U}$ such that
\[
\hat{U} c_p \hat{U}^\dag =\sum_{q=1}^{2N} R_{p,q} c_q.
\]
Then $\hat{U}\, |g_\sigma(u)\ra = |g_\sigma(0)\ra$ for $\sigma=0,1$
up to an overall phase.
\end{prop}
\begin{proof}
Let $\{\hat{a}_j=\hat{a}_j(u)\}_{j=1}^N$ be the
complex fermionic annihilation operators associated with the eigenmodes of $\hat{H}=\hat{H}_0+u\hat{V}$, see~\eqref{eq:Hamiltonianasdirac}. It follows that $|g_\sigma(u)\ra$ can be defined as the state
annihilated by all $\hat{a}_j(u)$ with $j=2,\ldots,N$
and having fermionic parity $\sigma$.
The condition $R P_+(0) R^T = P_+(u)$ implies that $R^T\psi_j(u)$ is a linear
combination of $\psi_k(0)$ with $k=2,\ldots,N$ for any $j\ge 2$.
Equivalently, $\hat{U} \hat{a}_j(u) \hat{U}^\dag$ is a linear combination
of $\hat{a}_k(0)$ with $k=2,\ldots,N$.
Thus $\hat{U}^\dag|g_\sigma(0)\ra$ is annihilated by any operator $\hat{a}_j(u)$
with $j=2,\ldots,N$. Finally, the condition $\det{R}=1$ implies that $\hat{U}$
preserves fermionic parity, that is, $\hat{U}^\dag|g_\sigma(0)\ra =|g_\sigma(u)\ra$,
up to an overall phase factor.
\end{proof}
We  will construct a rotation $R=R(u)$ as above by setting up a differential equation on
$P_+(u)$. This can be done using standard first-order degenerate perturbation theory. Namely,
suppose $P_+(u+\delta u)=e^{D\delta u} P_+(u) e^{-D\delta u}$ for
some anti-hermitian generator $D\equiv D(u)$ and infinitesimally small $\delta u$.
Then $e^{-D\delta u}(H_0+uV + \delta u V) e^{D\delta u}$ must commute
with $P_+(u)$ in the first order in $\delta u$. Hence
\be
\label{PT}
\frac{\partial P_+(u)}{\partial u} =[D(u),P_+(u)]
\ee
whenever $D(u)$ is an anti-hermitian operator satisfying
\be
\label{SW}
P_+^\perp ([H_0+u V, D(u)] + V)P_+=0.
\ee
In addition, we would like to find a real solution $D(u)$,
since we need to generate an orthogonal rotation.
Let us choose $D(u)$ as
\be
\label{D(u)}
D(u)=\calE(\calO(V)),
\ee
where $\calO$ and $\calE$  are  $u$-dependent super-operators
discarding the block-diagonal while keeping the `block-off-diagonal' part and inserting the energy denominator.
More precisely,
\begin{align*}
\cO(X)&=X-\sum_{\alpha\in\{0,+,-\}}P_\alpha XP_\alpha.
\end{align*}
To define $\calE$ let $\{|\phi_\lambda\ra\}$ be the orthonormal basis
of eigenvectors of $H_0+u V$
such that $(H_0+u V) |\phi_\lambda\ra= \lambda |\phi_\lambda\ra$.
Note that all eigenvalues $\lambda$ are imaginary.
We define $\calE(X)$ by its non-zero matrix elements in this basis,
\be
\label{E(X)}
\la \phi_\lambda |\calE(X) |\phi_{\lambda'}\ra = -\frac{ \la \phi_\lambda |X |\phi_{\lambda'}\ra}{\lambda-\lambda'}
\quad \mbox{for $\lambda\ne \lambda'$}\ .
\ee
It is clear that $D(u)$ defined by Eq.~(\ref{D(u)}) satisfies Eq.~(\ref{SW}).
It remains to check that $D(u)$ is real and anti-symmetric in the standard basis.
\begin{prop}
\label{prop:aux}
Let $X$ be any operator on $\CC^{2N}$. Then
\be
\label{Oprop}
\calO(X)^\dag=\calO(X^\dag), \quad \overline{\calO(X)}=\calO(\overline{X}),
\ee
and
\be
\label{Eprop}
\calE(X)^\dag=\calE(X^\dag), \quad \overline{\calE(X)}=\calE(\overline{X}).
\ee
\end{prop}
\begin{proof}
Identities Eq.~(\ref{Oprop})
follow from the fact that $P_\alpha$ are Hermitian and from the complex
conjugation rules $P_+=\overline{P}_-$, $\overline{P}_0=P_0$.
To prove Eq.~(\ref{Eprop}) we note that $|\overline{\phi}_\lambda\ra=|\phi_{-\lambda}\ra$.
Taking the complex conjugate of Eq.~(\ref{E(X)}), replacing $\lambda$ and $\lambda'$ by
$-\lambda$ and $-\lambda'$ respectively, and taking into account that
$\overline{\lambda}=-\lambda$, $\overline{\lambda'}=-\lambda'$, one gets
$\overline{\calE(X)}=\calE(\overline{X})$. The identity $\calE(X)^\dag=\calE(X^\dag)$
is obvious since $\{ |\phi_\lambda\ra\}$ is an orthonormal basis
and because the eigenvalues $\{\lambda\}$ are imaginary.
\end{proof}
Since $V$ is anti-hermitian and real, Proposition~\ref{prop:aux} implies that $D(u)$
 is also anti-hermitian and real.  Hence $D(u)$ is a real anti-symmetric matrix.
 Integrating Eq.~(\ref{PT}) over $u\in [0,1]$
we conclude that $P_+(1)=R P_+(0) R^T$, where
$R\in SO(2N)$ is the time-ordered exponential describing evolution under~$D(u)$,
\be
\label{Rdef}
R=\calT\cdot \exp{\left( \int_0^1 du D(u) \right)}\ .
\ee
 Define an adiabatic evolution Hamiltonian as
\begin{align*}
\hat{D}(u)=\frac{i}4 \sum_{p,q=1}^{2N} D_{p,q}(u)\, c_p c_q\ .
\end{align*}
Simple algebra shows that $[i\hat{D}(u),c_p]=\sum_{q=1}^{2N} D_{p,q}(u) c_q$.
Hence we can choose a unitary $\hat{U}$ satisfying the conditions of
Proposition~\ref{prop:map} as the time-ordered exponential describing the unitary evolution under~$\hat{D}(u)$, namely,
\begin{align*}
\hat{U}=\calT \cdot \exp{\left[ i\int_{0}^1 du \hat{D}(u)\right]}\ .
\end{align*}
Let us now bound the norm of $D(u)$. We note that
\[
P_+ D(u) P_0 =-i\int_0^\infty dt e^{i(H_0+u V)t} P_+ V P_0 e^{-i(H_0+u V)t}\ ,
\]
with a similar representation for other blocks $P_\alpha D(u) P_\beta$ with $\alpha\ne \beta$. This yields
\[
\| P_+ D(u) P_0\| \le \|V\| \int_0^\infty e^{-\lambda_2(u)t + \lambda_1(u) t} \le \frac{\epsilon}{\lambda_2(u)-\lambda_1(u)}
\le \frac{\epsilon}{1-2\epsilon}\ ,
\]
since $\lambda_2(u)\ge 1-\epsilon$ and $\lambda_1(u)\le \epsilon$.
(In fact, since $\lambda_1(u)$ is exponentially small in $N$,  we could get a stronger bound
with $1-\epsilon$ in the denominator which, however, would require $N$ to be
sufficiently large.)
We note that $P_- D(u) P_0$ is the complex conjugate of $P_+ D(u) P_0$ and thus
$\|P_-D(u) P_0\| \le \epsilon/(1-2\epsilon)$.
Similar arguments yield $\| P_+ D(u) P_-\| \le \epsilon/2(1-\epsilon)$.
Since $P_\alpha D(u) P_\alpha=0$ for all $\alpha$, we conclude  that
\[
\|D(u)\| \le \| P_+ D(u) P_0\| + \| P_- D(u) P_0\| + \| P_+ D(u) P_-\| \le \frac{3\epsilon}{1-\epsilon}\ .
\]
This proves Eq.~(\ref{Dnorm}).

It remains to show that off-diagonal matrix elements of $D(u)$ decay sufficiently fast.
We will represent $D(u)$ using contour integrals.
Let $C_{\pm}$ and $C_0$ be circles of radius
$3\epsilon/2$ centered at $z=\pm i$ and $z=0$ respectively.
Note that $C_{\pm}$ encircles all eigenvalues $\pm i\lambda_j(u)$ with $j=2,\ldots,N$,
while $C_0$ encircles $\pm i\lambda_1(u)$.
It follows that
\be
\label{cont}
P_{\alpha}(u)=\frac1{2\pi i} \int_{z\in C_\alpha} dz (zI-H_0-u V)^{-1}, \quad  \quad \alpha \in \{+,-,0\}\ .
\ee
Define $G_0(z)=(zI-H_0)^{-1}$. We note that $\|G_0(z)\|\le 2/(3\epsilon)$ for all $z\in C_\alpha$
since the spectrum of $H_0$ consists of eigenvalues $0,\pm i$.
Hence
\[
(zI-H_0-uV)^{-1}=(I-uG_0(z)V)^{-1} G_0(z),
\quad
\mbox{where} \quad
\|uG_0(z)V \| \le \frac23 \quad \mbox{for all $z\in C_\alpha$}\ .
\]
Furthermore, $(G_0(z) V)_{p,q}=0$ unless $|p-q|\ge r+1$
and $G_0(z)_{p,q}=0$ unless $|p-q|\le 1$.
Proposition~\ref{prop:res} from Section~\ref{sec:gap} now implies that
\[
|(zI-H_0-uV)^{-1}_{j,k} |\le 3 (2/3)^{\frac{|j-k|-1}{r+1}} \| G_0(z)\| \le \frac3{\epsilon} (2/3)^{\frac{|j-k|}{2r}}
\]
for all $z\in C_\alpha$. Block-off-diagonality of $D(u)$ implies that $D(u)=\sum_{\alpha\ne \beta} P_\alpha D(u) P_\beta$.
Using the contour integrals Eq.~\eqref{cont}
and taking into account Eq.~\eqref{E(X)}  we arrive at
\begin{align*}
P_\alpha D(u) P_\beta =
\frac1{(2\pi i)^2} \int_{z\in C_\alpha}   \int_{z'\in C_\beta } \; \; \frac{dz \, dz'}{z'-z}\;
(zI-H_0-u V)^{-1} V (z'I-H_0-u V)^{-1}\ .
\end{align*}
Taking into account that $|z-z'|\ge 1-3\epsilon \ge 1/4$ we get
\begin{align*}
|(P_\alpha D(u) P_\beta)_{j,k}|\le 4 \left( \frac{3\epsilon}2 \right)^2 \left(\frac{3}{\epsilon}\right)^2
\sum_{p,q=1}^{2N} |V_{p,q}| (2/3)^{\frac{|j-p| + |q-k|}{2r}}\ .
\end{align*}
The bound
$|j-p|+|q-k|\ge |j-k|-|p-q|\ge |j-k|-r$ yields
\[
\sum_{p,q=1}^{2N} |V_{p,q}| (2/3)^{\frac{|j-p| + |q-k|}{2r}}
\le
2\epsilon (2/3)^{|j-k|/4r} C^2, \quad C=2\sum_{p=0}^\infty (2/3)^{p/4r}\ .
\]
One can easily check that $(2/3)^{1/4r}\le 1-1/(16r)$ for all $r\ge 1$.
This implies $C\le 32 r$.
Combining the above bounds together
we arrive at
\[
|(P_\alpha D(u) P_\beta)_{j,k}|\le  81\cdot 2^{12} \cdot \epsilon r^2  (2/3)^{|j-k|/4r}\ .
\]
Combining the six terms $P_\alpha D(u) P_\beta$ together yields
the desired bound  Eq.~\eqref{Ddecay}, completing the proof of Theorem~\ref{thm:QAC}.

We end this section by proving Corollaries~\ref{corol:boring1} and~\ref{corol:boring2}.
Using Theorem~\ref{thm:QAC} we can express the perturbed ground state in each sector as $\ket{\tilde{g}_\sigma}=\hat{U}(1)^\dag\ket{g_\sigma}$ where $\hat{U}(u)$ satisfies
\begin{align*}
\frac{\partial \hat{U}(u)}{\partial u} &=i\hat{D}(u)\hat{U}(u)\
\end{align*}
and where $\hat{D}(u)$ is a Hermitian effective Hamiltonian
with strength roughly~$\epsilon$, see Eq.~\eqref{Dnorm}.
Let
\begin{align}
f(u)=\bra{g_\sigma}\hat{U}(u)\ket{g_\sigma}\ \label{eq:ftointegrate}
\end{align}
such that $|f(1)|=|\la \tilde{g}_\sigma|g_\sigma\ra|$
and let $Q=\hat{I}-|g_\sigma\ra\la g_\sigma|$.
Taking the derivative over~$u$ gives
\[
\frac{\partial f(u)}{\partial u}=i\bra{g_\sigma} \hat{D}(u)\ket{g_\sigma}f(u)+i\bra{g_\sigma}\hat{D}(u)Q\hat{U}(u)\ket{g_\sigma}\ .
\]
The first term in the derivative only contributes to the phase of~$f(u)$ and can be ignored.
The second term can be bounded using the Cauchy-Schwarz inequality which yields
\be
\label{derivative1}
-\frac{\partial |f(u)|}{\partial u} \le \sqrt{ \bra{g_\sigma}\hat{D}(u)Q\hat{D}(u)\ket{g_\sigma} \la g_\sigma| \hat{U}^\dag Q \hat{U} |g_\sigma\ra}=
\sqrt{1-|f(u)|^2} \sqrt{h(u)}\ ,
\ee
where $h(u)\equiv  \bra{g_\sigma}\hat{D}(u)Q\hat{D}(u)\ket{g_\sigma}$.
 We can bound $h(u)$ by observing that  the states $Qc_pc_q\ket{g_\sigma}$ and $Qc_{p'}c_{q'}\ket{g_\sigma}$  are orthogonal unless $(p,q)$ is a permutation of $(p',q')$. This yields
\begin{align*}
h(u) &\leq \frac{1}{16}\sum_{p,q=1}^{2N}\left(|D_{p,q}(u)|^2+|D_{p,q}(u)D_{q,p}(u)|\right)
=\frac18 \trace{D^\dag(u) D(u)}\ .
\end{align*}
The bound~(\ref{Dnorm}) and the assumption $\epsilon<1/8$ imply $\|D(u)\| \le 4\epsilon$
and thus $\trace{D^\dag(u) D(u)}\le 2N\|D(u)\|^2 \le 32 N \epsilon^2$.
We conclude that $h(u)\le 4N\epsilon^2$. Substituting this into Eq.~(\ref{derivative1})
yields
\begin{align}
\frac{\partial}{\partial u} \sqrt{1-|f(u)|^2} \le 2\epsilon \sqrt{N}\ .\label{eq:derivativebound}
\end{align}
Integrating~\eqref{eq:derivativebound} over $u\in [0,1]$ and taking into account that $f(0)=1$ gives
$\sqrt{1-|f(1)|^2}\le 2\epsilon \sqrt{N}$, that is, $|f(1)|^2 \ge 1-4N\epsilon^2$
which proves  Corollary~\ref{corol:boring1}.

It remains to prove Corollary~\ref{corol:boring2}.
By definition, the storage fidelity is
\[
F_{|g\ra}(t) = \la g|\Phi_{ec}(|g(t)\ra\la g(t)|)|g\ra,
\quad \mbox{where} \quad |g(t)\ra = e^{i(\hat{H}_0+\hat{V})t}\, |g\ra.
\]
Choose the perturbed ground states such that $\la \tilde{g}_\sigma |g_\sigma\ra\ge 0$
and let $|\tilde{g}\ra =(|\tilde{g}_0\ra\otimes |0_R\ra  +|\tilde{g}_1\ra \otimes |1_R\ra)/\sqrt{2}$.
Simple algebra shows that
\be
\label{g1}
\| \, |g\ra\la g| - |\tilde{g}\ra\la \tilde{g} | \, \|_1 \le  2\sqrt{1-\min_\sigma |\la \tilde{g}_\sigma |g_\sigma\ra|^2}
\le 4\epsilon \sqrt{N}\ .
\ee
Consider the time-evolved perturbed ground state
\[
|\tilde{g}(t)\ra\equiv e^{i(\hat{H}_0+\hat{V})t}\, |\tilde{g}\ra
\sim
(|\tilde{g}_0\ra\otimes |0_R\ra  +e^{i\delta t}\,  |\tilde{g}_1\ra \otimes |1_R\ra)\sqrt{2}\ .
\]
Define also an auxiliary state
\[
|h(t)\ra =(|{g}_0\ra\otimes |0_R\ra  +e^{i\delta t}\,  |{g}_1\ra \otimes |1_R\ra)\sqrt{2}\ .
\]
Hence $|h(t)\ra$ is a ground state of $\hat{H}_0$
whose time evolution amounts to the dephasing $e^{i\delta t}$ between $|g_0\ra$ and $|g_1\ra$.
Using Eq.~\eqref{g1} we get
\begin{align*}
\| \, |g(t)\ra \la g(t)| - |h(t)\ra \la h(t)| \, \|_1 & \le
\| \, |g(t)\ra \la g(t)| -  |\tilde{g}(t)\ra\la \tilde{g}(t)|  \, \|_1 + \| \,  |\tilde{g}(t)\ra\la \tilde{g}(t)| -
|h(t)\ra \la h(t)| \, \|_1 \nn \\
&= 2 \| \, |g\ra\la g| - |\tilde{g}\ra\la \tilde{g} | \, \|_1 \le 8\epsilon \sqrt{N}\ .
\end{align*}
Replacing $|g(t)\ra\la g(t)|$ by $|h(t)\ra\la h(t)|$ in the expression for $F_{|g\ra}(t)$
and taking into account that $|h(t)\ra\la h(t)|$ is invariant under $\Phi_{ec}$,
we arrive at
\[
|F_{|g\ra}(t) -|\la g |h(t)\ra|^2|  \le  8{\epsilon}\sqrt{N}\ .
\]
This proves Eq.~(\ref{boring2}) since $|\la g |h(t)\ra|^2=\cos^2{(\delta t/2)}$.

\section{Lower bound on the storage fidelity\label{sec:disorderenhanced}}
In this section, we prove Theorem~\ref{thm:storage}.
We first express the time-evolved state $|g(t)\ra=e^{i(\hat{H}_0+\hat{V})t}|g\rangle$ as the result of applying a certain unitary to a  related ground state~$|h(t)\ra$.
This unitary is composed of the quasi-adiabatic continuation operator~$\hat{U}$ of Theorem~\ref{thm:QAC} and a time-evolved version thereof. By expanding the time-ordered exponential defining~$\hat{U}$, we rewrite this unitary in terms of correctable and uncorrectable errors. Using the exponential decay~\eqref{Ddecay} and the  localization assumption~\eqref{DL}, we then show that the weight of uncorrectable errors is exponentially small in the system size. This implies the desired result.

\subsection{Time-evolution in the Anderson localization regime\label{subs:timeevolution} }
As we argued in Section~\ref{subs:EC}, we can formally compute the storage fidelity by
considering encoded states $|g\ra=\alpha_0 |g_0\ra + \alpha_1 |g_1\ra$ even though these states
are unphysical.
Let $\hat{U}$ be the quasi-adiabatic continuation operator satisfying the conditions of Theorem~\ref{thm:QAC},
see Eqs.~(\ref{qac1},\ref{qac2}),
that is, $|{g}_\sigma\ra=\hat{U}|\tilde{g}_\sigma\ra$,
where $|\tilde{g}_\sigma\ra$ is the  ground state of $\hat{H}_0+\hat{V}$
in the sector with fermionic parity $\sigma$.
We can rewrite the time-evolved state as
\bea
|g(t)\ra&=&e^{i(\hat{H}_0+\hat{V})t} \hat{U} e^{-i(\hat{H}_0+\hat{V})t} e^{i(\hat{H}_0+\hat{V})t}  \hat{U}^\dag |g\ra \nn \\
&=& e^{i(\hat{H}_0+\hat{V})t} \hat{U} e^{-i(\hat{H}_0+\hat{V})t} (\alpha_0 |\tilde{g}_0\ra + e^{\pm i\delta t} \alpha_1 |\tilde{g}_1\ra) \nn \\
&=& e^{i(\hat{H}_0+\hat{V})t} \hat{U} e^{-i(\hat{H}_0+\hat{V})t} \hat{U}^\dag |h(t)\ra,
\eea
where $\delta=E_1^\uparrow - E_0^\uparrow$ is the exponentially small
energy splitting, see Theorem~\ref{thm:gap}, and
\[
|h(t)\ra\equiv  \alpha_0 |{g}_0\ra + e^{\pm i\delta t} \alpha_1 |{g}_1\ra
\]
is a `dephased version' of the initial state $|g\ra$.
(Here and below we ignore the overall phase of quantum states).
The operator $\hat{U}$ describes a unitary evolution under
a time-dependent Hamiltonian
\begin{align*}
\hat{D}(u)&=\frac{i}{4}\sum_{p,q=1}^N D_{p,q}(u)c_pc_q\ ,\qquad 0\leq u\leq 1\ ,
\end{align*}
with exponentially decaying interactions,
\begin{align}
|D_{p,q}(u)|\leq c\epsilon r^2 \lambda^{|p-q|}\qquad\textrm{ where }\lambda=(2/3)^{\textfrac{1}{4r}}\ , \label{Ddecayp}
\end{align}
see Theorem~\ref{thm:QAC}. Defining
\begin{align*}
\hat{U}(t)&=e^{i(\hat{H}_0+\hat{V})t}\hat{U}e^{-i(\hat{H}_0+\hat{V})t}
\end{align*}
we have
\be
\label{Wop}
\ket{g(t)}=\hat{U}(t)\hat{U}^\dagger\, \ket{h(t)}\ .
\ee
We will use~\eqref{Ddecayp} to argue $\hat{U}$ (respectively $\hat{U}^\dagger$) has small weight on uncorrectable errors. We will subsequently use the localization condition~\eqref{DL} to argue that the same is true for~$\hat{U}(t)\hat{U}^\dagger$.

\subsection{Uncorrectable errors and their qubit weight}
\label{subs:JW}
It will be convenient to describe the error correction
algorithm of Section~\ref{subs:EC}
by mapping the fermionic modes to qubits via the  Jordan-Wigner transformation
\be
\label{JW}
c_{2j-1}=Z_1 \cdots Z_{j-1} X_{j} \quad \mbox{and}\quad
c_{2j} = Z_1 \cdots Z_{j-1} Y_{j}\ ,
\ee
where $j=1,\ldots,N$ and $X_j$ ($Z_j$) stands for the Pauli $\sigma^x$  ($\sigma^z$)
operator acting on a qubit $j$.  This establishes a one-to-one correspondence between Majorana monomials
and $N$-qubit Pauli operators.
For any Majorana monomial $E=c_{p_1} \cdots c_{p_k}$ we define
the {\em fermionic weight} $\fw(E)$ as the number of Majorana modes in $E$, that is,
$\fw(E)=k$. We also define the {\em qubit weight} $\qw(E)$ as the number of qubits acted on by the Pauli  version of $E$.
Applying the Jordan-Wigner transformation Eq.~(\ref{JW}) to the stabilizers
$\hat{S}_j=(-i) c_{2j} c_{2j+1}$ and elementary errors $E_j=(-1)^{a^\dag_j a_j}=(-i)c_{2j-1}c_{2j}$ one gets
\[
\hat{S}_j= X_j X_{j+1}, \quad j=1,\ldots,N-1 \quad \quad  \mbox{and} \quad \quad
E_j=Z_j, \quad j=1,\ldots,N\ .
\]
The stabilizer code specified by $\hat{S}_1,\ldots,\hat{S}_{N-1}$ is the repetition code
on $N$ qubits. It corrects Pauli errors of $Z$-type on any subset of  at most $N/2-1$ qubits.
Although Pauli errors of $X$-type cannot be corrected, fermionic superselection rules imply
that only an even number of $X$-type errors can occur since any $X$-type error changes the
parity of the total number of fermions. Hence admissible errors of $X$-type
are always products of stabilizers $\hat{S}_j$ and can be ignored.
A general Pauli error $E$ can be uniquely represented as $E=E^X E^Z$,
where $E^X$ and $E^Z$ are products of $X$-type and $Z$-type errors respectively.
Let us say that $E$ is an {\em uncorrectable error} iff $E^Z$ acts on at least $N/2$ qubits.
It follows that any uncorrectable error has qubit weight at least $N/2$.

\subsection{Errors caused by  the quasi-adiabatic evolution operator\label{sec:errorboundsec}}
Let us expand the unitary evolution operator $\hat{U}$ of Theorem~\ref{thm:QAC}
in the basis of Majorana monomials:
\be
\label{opbasis}
\hat{U}=\sum_{E} \omega_E \, E, \quad \quad \omega_E=\frac1{2^N} \trace{(E^\dag \hat{U})}\ .
\ee
Here the sum runs over all $4^N$ Majorana monomials $E$.
The goal of this section is to bound the amplitudes $\omega_E$ of uncorrectable errors $E$.
\begin{lemma}
\label{lemma:omega(P)}
Let $E$ be any Majorana monomial with fermionic weight $\fw$
and qubit weight $\qw$.   Then
\be
\label{PauliWeight}
|\omega_E| \le
\epsilon^{\fw/4} \lambda^{\qw} e^{yN}, \quad \mbox{where} \quad y\equiv \frac{3\sqrt{\epsilon} cr^2}{1-\sqrt{\lambda}}\ ,
\ee
and where the constants $c$ and $1/2\leq \lambda\leq 1$ are defined by~\eqref{Ddecayp}.
\end{lemma}
\begin{proof}
Indeed, expanding the time-ordered exponential defining~$\hat{U}$ we get
\[
|\omega_E|\le \sum_{n=0}^\infty \frac{1}{4^n n!}
\sum_{\vec{p}} \sum_{\vec{q}} \, \max_{\vec{u}}\;  \Gamma^{(n)}_{\vec{p},\vec{q},\vec{u}} \, \, ,
\]
where
\[
\Gamma^{(n)}_{\vec{p},\vec{q},\vec{u}}=
|D_{p_1,q_1}(u_1) \cdots D_{p_n,q_n}(u_n)| \cdot
\frac1{2^N} |\trace{(E^\dag c_{p_1} c_{q_1} \cdots c_{p_n} c_{q_n})}|.
\]
In the above equations the sums over
$\vec{p}$ and $\vec{q}$ run over all  $n$-tuples of integers in the interval $[1,2N]$,
while the maximum is taken over all $n$-tuples $\vec{u}$ of real numbers in the interval~$[0,1]$.
We claim that $\Gamma^{(n)}_{\vec{p},\vec{q},\vec{u}}=0$ unless
\be
\label{nrange}
2n\ge \fw \quad \mbox{and} \quad \sum_{a=1}^n (3+|p_a-q_a|) \ge 2\qw\ .
\ee
The first condition is obvious: the fermionic weight of $c_{p_1} c_{q_1} \cdots c_{p_n} c_{q_n}$
is at most $2n$, while $E^\dag$ has fermionic weight $\fw$.
To derive the second condition we compute the trace by applying the Jordan-Wigner
transformation to all operators and noting that  $c_p c_q$ has a qubit weight at most $(3+|p-q|)/2$.
Combining the second condition in Eq.~(\ref{nrange}) with the bound  Eq.~(\ref{Ddecayp}) results in
\[
\Gamma^{(n)}_{\vec{p},\vec{q},\vec{u}} \le (c\epsilon r^2)^n \lambda^{\sum_{a=1}^n |p_a-q_a|}
\le (c\epsilon r^2\lambda^{-3})^n \lambda^{2\qw}\ .
\]
This implies
\[
\Gamma^{(n)}_{\vec{p},\vec{q},\vec{u}} \le (c\epsilon r^2\lambda^{-3/2})^n \lambda^{\qw} \lambda^{\sum_{a=1}^n |p_a-q_a|/2}\ .
\]
Taking into account that
\[
\sum_{\vec{p}} \sum_{\vec{q}} \, \lambda^{\sum_{a=1}^n |p_a-q_a|/2} \le
(2N)^n \left( \sum_{p\in \ZZ} \lambda^{|p|/2} \right)^n \le \frac{(4N)^n}{(1-\lambda^{1/2})^n}
\]
we arrive at
\[
|\omega_E|\le \lambda^{\qw} \sum_{n=\fw/2}^\infty \frac{x^n}{n!}, \quad x\equiv \frac{c\epsilon r^2 N\lambda^{-3/2}}{1-\lambda^{1/2}}\ .
\]
Since we assumed $1/2\le \lambda<1$, we have $\lambda^{-3/2}\le 3$.
Using $\epsilon^n\le \epsilon^{\fw/4} \epsilon^{n/2}$ and extending the sum to all non-negative $n$ we easily get Eq.~(\ref{PauliWeight}).
\end{proof}

We next consider the coefficients~$\nu_E$ in the expansion
\begin{align*}
\hat{U}(t) &=\sum_E \nu_E E\ ,\qquad\nu_E=\frac{1}{2^N}\trace(E^\dagger \hat{U}(t))\ .
\end{align*}
\begin{lemma}\label{lem:timeevolvedunitaryexpansion}
Let $E$ be any Majorana monomial with fermionic weight~$\fw$ and qubit weight~$\qw$. Assume that the localization condition~\eqref{DL} holds for some constants $C,\xi\geq 0$, where $C\geq 1$. Then
\begin{align*}
\ExpE{|\nu_E|}\leq
(1+\epsilon^{1/8})^{2N}e^{yN}\cdot
\begin{cases}
\epsilon^{N/(16\cdot 24)}\epsilon^{\fw/16}\qquad &\textrm{ if } \qw\geq N/4\ \textrm{ and }\ \fw> N/24\\
\max\{\lambda^{1/8},e^{-1/\xi}\}^NC^{\fw}\epsilon^{\fw/8}\qquad &\textrm{ if } \qw\geq N/4\ \textrm{ and }\ \fw\leq N/24\\
\epsilon^{\fw/8} &\textrm{ if }\qw <N/4\
\end{cases}
\end{align*}
where $y$ is defined in Eq.~\eqref{PauliWeight}.
\end{lemma}
The above lemma is the only step in the proof that uses the dynamical localization condition of Definition~\ref{def:DL}.
In fact, the lemma can be proved using a slightly weaker version of Eq.~\eqref{DL}. Namely, we only need a bound
\[
\frac1{2^N}\ExpE{|\trace{( E_1(t) E)}|}\le C^m e^{-\frac{N}{\xi}}
\]
for any Majorana monomials $E,E_1$ with fermionic weight $m$ and
qubit weights $\qw(E)\ge N/4$, $\qw(E_1)<N/8$. This bound can be easily derived from Eq.~\eqref{DL} as explained below.
\begin{proof}
We have
\begin{align}
\nu_E=\sum_{E_1}\omega_{E_1}\frac{1}{2^N}\trace(E_1(t)E^\dagger)\label{eq:nuedef}
\end{align}
 where
 $\hat{A}(t)=e^{i(\hat{H}_0+\hat{V})t}\hat{A}e^{-i(\hat{H}_0+\hat{V})t}$ is the time-evolved operator~$\hat{A}$. If $E_1$ has fermionic weight $m$, i.e.,
\begin{align}
E_1&=c_{p_1}\cdots c_{p_{m}}\equiv c_{\vec{p}}\quad \textrm{for some }\quad p_1<p_2<\cdots <p_{m}\ , \label{eq:fermionicoperatoreone}
\end{align}
 then
\begin{align*}
E_1(t)&=c_{p_1}(t)\cdots c_{p_m}(t)=\frac{1}{m!}\sum_{\pi\in S_m}(-1)^\pi c_{p_{\pi(1)}}(t)\cdots c_{p_{\pi(m)}}(t)\
\end{align*}
since the time-evolved operators $c_p(t)$ anticommute.
With Eq.~\eqref{eq:cpheisenberg}, we get
\begin{align}
E_1(t)&=\frac{1}{m!}\sum_{\vec{q}}\tilde{R}_{\vec{p},\vec{q}} c_{\vec{q}}\ ,\qquad \tilde{R}_{\vec{p},\vec{q}}=\sum_{\pi\in S_m}(-1)^\pi \prod_{a=1}^m R_{p_{\pi(a)},q_a}=\det R[\vec{p},\vec{q}|\label{eq:submatrixdetdef}
\end{align}
where the sum is over all $m$-tuples $\vec{q}=(q_1,\ldots,q_m)$ of integers in the interval~$[1,2N]$ and $R_{p,q}=R_{p,q}(t)$.  The latter expression shows that $\tilde{R}_{\vec{p},\vec{q}}$ is antisymmetric in the components of $\vec{q}$. In particular, we can restrict the sum  to $m$-tuples with distinct entries, $q_a\neq q_b$ for $a\neq b$. We conclude that $E_1(t)$ is a linear combination of monomials with the same fermionic weight~$m$ as~$E_1$. Using this fact, we can restrict the sum~\eqref{eq:nuedef} to monomials~$E_1$ with the same fermionic weight~$\fw$ as~$E$. Applying Lemma~\ref{lemma:omega(P)} then  gives
\begin{align*}
\ExpE{|\nu_E|} &\leq \sum_{E_1:\fw(E_1)=\fw} \epsilon^{\fw/4}\lambda^{\qw(E_1)}e^{yN}\frac{1}{2^N}\ExpE{|\trace(E_1(t)E^\dagger)|}\\
&\leq \binom{2N}{\fw}\epsilon^{\fw/4}e^{yN}\max_{E_1:\fw(E_1)=\fw}\ExpE{\lambda^{\qw(E_1)}\frac{1}{2^N}|\trace(E_1(t)E^\dagger)|}\\
&\leq \epsilon^{\fw/8}(1+\epsilon^{1/8})^{2N}e^{yN}\max_{E_1:\fw(E_1)=\fw}\ExpE{\lambda^{\qw(E_1)}\frac{1}{2^N}|\trace(E_1(t)E^\dagger)|}\ .
\end{align*}
The third inequality of the Lemma follows immediately. Similarly, the first inequality follows because $\epsilon^{\fw/8}\leq \epsilon^{\fw/16}\epsilon^{N/(16\cdot 24)}$ for $\fw>N/24$. It remains to show that for all $E_1$ with $\fw(E_1)=\fw\leq N/24$, we have
\begin{align}
\ExpE{\lambda^{\qw(E_1)}\frac{1}{2^N}|\trace(E_1(t)E^\dagger)|}\leq \max\{\lambda^{1/8},e^{-1/\xi}\}^N C^\fw\qquad\textrm{if }\qw\geq N/4\ .\label{eq:toshoweq}
\end{align}
As this statement is trivial if $\qw(E_1)\geq N/8$,  it suffices to consider the case where $\qw(E_1)< N/8$.
If $E_1=c_{\vec{p}}$ for $\vec{p}=(p_1<\cdots <p_\fw)$, we have
\begin{align*}
E_1(t)&=\frac{1}{\fw!}\sum_{\vec{q}} \det R[\vec{p},\vec{q}] c_{\vec{q}}
\end{align*}
according to~\eqref{eq:submatrixdetdef} and hence
\begin{align}
\max_{E_1:\substack{\fw(E_1)=\fw,\\ \qw(E_1)< N/8}}\ExpE{\lambda^{\qw(E_1)}\frac{1}{2^N}|\trace(E_1(t)E^\dagger)|}&\leq \max_{\vec{p}: \qw(c_{\vec{p}})<N/8}\ExpE{|\det R[\vec{p},\vec{q}]|}\label{eq:e1bound}
\end{align}
assuming that~$E=c_{\vec{q}}$. Consider two Majorana monomials $c_{\vec{p}}$ and $c_{\vec{q}}$ with fermionic weight~$\fw$.  We claim that
\begin{align}
\sum_{a=1}^\fw |p_a-q_a|\geq N/8\qquad\textrm{ if } \qw(c_{\vec{q}})\geq N/4,\ \qw(c_{\vec{p}})<N/8\ \textrm{ and }\fw\leq N/24\ .\label{eq:distqubitweight}
\end{align}
Indeed, since $c_{\vec{q}}=\left(\prod_{a=1}^\fw c_{p_a}\right)\prod_{a=1}^\fw (c_{p_a}c_{q_a})$ (up to a sign), we  have
\begin{align*}
\qw(c_{\vec{q}})\leq \qw(c_{\vec{p}})+\sum_{a=1}^\fw \qw(c_{p_a}c_{q_a})\ .
\end{align*}
Using $\qw(c_pc_q)\leq (3+|p-q|)/2 $, we conclude that
\begin{align}
\sum_{a=1}^\fw |p_a-q_a|\geq 2(\qw(c_{\vec{q}})-\qw(c_{\vec{p}}))-3\fw\ ,
\end{align}
and the claim~\eqref{eq:distqubitweight} follows.
Combining~\eqref{eq:e1bound} and~\eqref{eq:distqubitweight} with our localization assumption~\eqref{DL} gives~\eqref{eq:toshoweq}, as desired.
\end{proof}

In summary, Lemma~\ref{lemma:omega(P)} and
Lemma~\ref{lem:timeevolvedunitaryexpansion}
show that
\begin{align*}
|\omega_E| &\leq \beta_1^{\fw}\Delta_1^N\lambda^\qw\ ,\qquad
\ExpE{|\nu_E|}\leq \beta_2^{\fw}\Delta_2^N\cdot \begin{cases}
\tilde{\lambda}^N\qquad&\textrm{ if }\qw\geq N/4\\
1 &\textrm{otherwise}
\end{cases}
\end{align*}
where
\begin{align*}
\begin{matrix}
\beta_1 &=&\epsilon^{1/4}\qquad &\lambda&=&(2/3)^{1/(4r)}\qquad &\Delta_1&=&e^y\\
\beta_2&=&C\epsilon^{1/16}\qquad &\tilde{\lambda}&=&\max\{\lambda^{1/8},e^{-1/\xi},\epsilon^{1/(16\cdot 24)}\} \qquad &\Delta_2&=&(1+\epsilon^{1/8})^2e^y\ .
\end{matrix}
\end{align*}
Consider two Majorana monomials  $E_1,E_2$ with fermionic weights $\fw_1$ and $\fw_2$ such that $\qw(E_1E_2)\geq N/2$. We then have
\begin{align}
\ExpE{|\omega_{E_1}\nu_{E_2}|}\leq \beta_1^{\fw_1}\beta_2^{\fw_2}\cdot (\Delta_1\Delta_2)^N\cdot \max\{\lambda,\tilde{\lambda}\}^N \label{eq:omeganubound}
\end{align}
because $\qw(E_1)+\qw(E_2)\geq \qw(E_1E_2)$.

\subsection{Lower bound on the storage fidelity}
We can now easily prove Theorem~\ref{thm:storage}.

\begin{proof}
Represent the time-evolved state $|g(t)\ra$ in Eq.~(\ref{Wop}) as
\[
|g(t)\ra = \hat{U}(t)\hat{U}^\dagger \, |h(t)\ra =  \psi_{\good}+ \psi_{\bad},
\]
where $\psi_{\good}$ and $\psi_{\bad}$  represent the overall contribution of correctable (qubit weight $<N/2$) and uncorrectable
(qubit weight $\ge N/2$) errors, respectively:
\[
\psi_{\good}=\sum_{E\, : \, \qw(E)<N/2} \; \theta_E\,  E\, |h(t)\ra
\quad
\mbox{and}
\quad
\psi_{\bad}=\sum_{E\, : \, \qw(E)\ge N/2} \; \theta_E\,  E\, |h(t)\ra.
\]
Here
\begin{align*}\theta_E=\frac{1}{2^N}\trace(E^\dagger\hat{U}(t)\hat{U}^\dagger)=\sum_{\substack{E_1,E_2:\\ E_2E_1=E}}\nu_{E_2}\omega_{E_1}\
\end{align*}
is the coefficient corresponding to the monomial~$E$ in the expansion of the operator~$\hat{U}(t)\hat{U}^\dagger$.
We obtain the following bound on the norm of $\psi_{\bad}$:
\begin{align*}
\| \psi_{bad}\| \le \sum_{E\, : \, \qw(E)\ge N/2} \; \; |\theta_E|&\leq \sum_{\substack{E_1,E_2:\\ \qw(E_2E_1)\ge N/2}} |\nu_{E_2}\omega_{E_1}|\ .
\end{align*}Taking the expectation over disorder realizations, we get
with~\eqref{eq:omeganubound},
\begin{align*}
\ExpE{\|\psi_{bad}\|}&\leq \sum_{\fw_1,\fw_2}\binom{2N}{\fw_1}\binom{2N}{\fw_2}\beta^{\fw_1}\beta_2^{\fw_2} (\Delta_1\Delta_2)^N\eta^N\\
&=\left((1+\beta_1)^{2}(1+\beta_2)^{2}(\Delta_1\Delta_2)\right)^N\max\{\lambda,\tilde{\lambda}\}^N\ .
\end{align*}
 Defining
\begin{align}
f(\epsilon,C,\xi,r)&\equiv e^{2y}(1+\epsilon^{1/8})^2(1+\epsilon^{1/4})^2(1+C\epsilon^{1/16})^2\max\{(2/3)^{1/(32r)},e^{-1/\xi},\epsilon^{1/(16\cdot 24)}\}\ ,\label{eq:fbaddef}
\end{align}
we conclude that $\ExpE{\| \psi_{bad}\|}\leq f(\epsilon,C,\xi,r)^N$. We will assume that $f(\epsilon,C,\xi,r)<1$. Let $\psi_{\good}'=\| \psi_{\good} \|^{-1}\cdot  \psi_{\good}$,  such that $\psi_{\good}'$ is a normalized state.
Simple algebra shows that
\[
\ExpE{\| \, |\psi_{\good}'\ra - |g(t)\ra \, \|}  \le 2\ExpE{\|\psi_{\bad}\|} \le 2f(\epsilon,C,\xi,r)^N\ .
\]
Furthermore, using the definition of the error correction map, see Eq.~(\ref{ECmap}), we get
\[
\Phi_{ec}(|\psi_{\good}'\ra\la \psi_{\good}'|)=|h(t)\ra\la h(t)|
\]
since $C(s) \hat{Q}_s E \, |h(t)\ra \sim |h(t)\ra$ for any correctable error $E$.
This allows us to bound the expected storage fidelity as
\begin{align*}
\ExpE{F(t)}\ge |\la  h(t)|g\ra|^2 -2\sqrt{2} f(\epsilon,C,\xi,r)^N=
\bigl||\alpha_0|^2+|\alpha_1|^2 e^{\pm i\delta t}\bigr|^2 - 2\sqrt{2} f(\epsilon,C,\xi,r)^N\ ,
\end{align*}
where $\delta=E_1^\uparrow - E_0^\uparrow$ is the exponentially small energy splitting
of the ground state, see Theorem~\ref{thm:gap}. Since~$\ket{g}$ is a normalized state, the first term on the rhs.~is lower bounded by $\cos^2(\delta t/2)$.

Collecting all the constants involved in the definition of $f(\epsilon,C,\xi,r)$
and taking into account that $\epsilon<1/4$ one easily gets the bound
\[
f(\epsilon,C,\xi,r)\le \exp{\left[ \alpha r^3 \sqrt{\epsilon} + \beta C \epsilon^{\frac1{16}} \right]}
\cdot
\max{\left[
e^{-\frac1{16r}}, e^{-\frac1{\xi}},e^{-\frac1{280}}\right]},
\]
where $\alpha,\beta=O(1)$ are some constants.
Furthermore, Theorem~\ref{thm:gap} implies
\[
\delta \le \left(\frac{2\epsilon}{1+\epsilon}\right)^{\frac{N}{2r}} \le 2^{-\frac{N}{2r}}.
\]
Therefore $\cos^2(\delta t/2)\ge 1- t^2 2^{-\frac{N}{r}}$.
This completes the proof of Theorem~\ref{thm:storage}.
\end{proof}

\section{Localization and the one-particle problem\label{sec:dynamiclocalization}}
In this section we explain how to map the one-particle Hamiltonian $H_0+V$ corresponding to the Majorana chain model
to the one-dimensional Anderson model with disorder in the hopping amplitudes.
We give evidence for the exponential localization of eigenvectors in the latter model by computing
Lyapunov exponents for the
one-particle eigenfunctions. Finally we show how to replace the expectation value of the determinant in the
multi-point dynamical localization condition  Eq.~(\ref{DL}) by more standard  multi-point correlation functions,
see Section~\ref{sec:momentlocalization}.

\subsection{Relation to  one-particle Anderson models\label{sec:oneparticleproblem}}

Consider a fixed configuration of the chemical potentials $\mu_1,\ldots,\mu_N$ and define
a triple of one-particle Hamiltonians acting on $\CC^{2N}$, $\CC^{2N}$, and $\CC^N$ respectively:
\begin{align}
H_0+V&=\left[ \ba{cccccccc}
 0 &       -\mu_1&  & & & &  \\
 \mu_1   &  0   & 1     &       & & & \\
          & -1   & 0     & -\mu_2 & & &  \\
          &      & \mu_2& 0      & \ddots \hphantom{\mu_4} & &  \\
          &&& \ddots  & \ddots\hphantom{\mu_4^2}&1 &\\
          &&&&-1 & 0  & -\mu_{N}  \\
& & & & &  \mu_{N} & 0 \\
\ea \right]\ .\label{eq:hslocalization}
\end{align}
\begin{align}
H_s&=\left[ \ba{cccccccc}
 0 &       \mu_1&  & & & &  \\
 \mu_1   &  0   & -1     &       & & & \\
          & -1   & 0     & \mu_2 & & &  \\
          &      & \mu_2& 0      & \ddots \hphantom{\mu_4} & &  \\
          &&& \ddots  & \ddots\hphantom{\mu_4^2}&-1 &\\
          &&&&-1 & 0  & \mu_{N}  \\
& & & & &  \mu_{N} & 0 \\
\ea \right]\ .\label{eq:hslocalizationhermitian}
\end{align}
\begin{align}
H_+&=\left[ \ba{cccccccc}
 \mu_1^2 & \mu_2 &  & & & &  \\
 \mu_2   &  \mu_2^2 & \mu_3 & & & & \\
 & \mu_3 & \mu_3^2 & \mu_4 & & &  \\
& & \mu_4 & \mu_4^2 & \ddots \hphantom{\mu_5} & &  \\
&&& \ddots  & \ddots\hphantom{\mu_4^2}&  \mu_{N-1} &\\
&&&& \mu_{N-1} & \mu_{N-1}^2  & \mu_{N}  \\
& & & & &  \mu_{N} & -1+\mu_N^2  \\
\ea \right]\ .\label{eq:oneparticleandersonmodel}
\end{align}
Note that $H_0+V$ is the one-particle Hamiltonian describing the Majorana chain model
with $J=1$ and  $|\Delta|=w$, see Eqs.~(\ref{H0},\ref{V}).
The Hamiltonian Eq.~\eqref{eq:hslocalizationhermitian} corresponds to an Anderson model with off-diagonal disorder.
It describes a particle that can hop between adjacent sites,
where every second hopping amplitude is random and the remaining hopping amplitudes are homogeneous.
Hamiltonian~\eqref{eq:oneparticleandersonmodel} involves  both off-diagonal as well as (correlated) on-site disorder. Because~\eqref{eq:oneparticleandersonmodel} avoids the bipartite structure inherent in~\eqref{eq:hslocalizationhermitian}, this form will be more useful for our purposes. Furthermore, if one neglects the diagonal matrix elements
in~\eqref{eq:oneparticleandersonmodel}, one obtains an  Anderson model with uncorrelated random
hopping amplitudes which has been very well studied. In particular, dynamical localization of this model
has been established\footnote{Strictly speaking, the dynamical localization bound of Ref.~\cite{DelyonKunzSouillard83}
applies to any energy interval that does not contain the point $E=0$.}
in~\cite{DelyonKunzSouillard83}.

Let us now describe relationships between the three Hamiltonians defined above.
We begin by noting that  $H_s$ and $i(H_0+V)$ can be mapped to each other by
a diagonal unitary operator:
\begin{align}
H_s&\equiv i\, U(H_0+V)U^\dagger, \quad
\mbox{where} \quad U=\sum_{p=1}^{2N} i^{p-1} \, |p\ra\la p|.
\label{eq:hsexpression}
\end{align}
Hence these two Hamiltonians are equivalent as far as their localization properties are concerned.

To describe the map from $H_0+V$ to $H_+$
 let us identify the $2N$~Majorana operators $c_1,\ldots,c_{2N}$ with standard basis vectors of a Hilbert space~$\mathcal{H}\cong\mathbb{C}^N\otimes\mathbb{C}^2$ such that
\begin{align}
c_{2p-1} \leftrightarrow \ket{p}\otimes\ket{0}\qquad\textrm{ and }\qquad c_{2p}\leftrightarrow \ket{p}\otimes\ket{1}\ ,\qquad p=1,\ldots,N\ .\label{eq:cpencoding}
\end{align}
We then have
\begin{align*}
H_0&=T\otimes \ket{1}\bra{0}-T^\dagger\otimes\ket{0}\bra{1}\qquad T=\sum_{k=1}^{N-1}\ket{k}\bra{k+1}\\
V&=S\otimes\left(\ket{1}\bra{0}-\ket{0}\bra{1}\right)\qquad S=\sum_{k=1}^N\mu_k \proj{k}\ .
\end{align*}
A straightforward calculation  gives
\begin{align*}
-(H_0+V)^2&=(T^\dagger+S)(T+S)\otimes\proj{0}+(T+S)(T^\dagger+S)\otimes\proj{1}\ .
\end{align*}
In terms of these operators, $H_+$ is given by
\begin{align}
H_+&=(T+S)(T^\dagger+S)-I\ .\label{eq:oneparticleandersonmodelp}
\end{align}
The spectrum of $H_+$ is therefore formed by $N$ eigenvalues $E_j=\lambda_j^2-1$,
$j=1,\ldots,N$. Localization properties of eigenvectors of $H_0+V$ can directly  be translated to
those of $H_+$ and vice versa.
Indeed, projecting Williamson eigenvectors~$\psi_j$ onto the even sublattice we obtain a complete family of eigenvectors for~$H_+$. Conversely, let
$\phi\in \CC^N$ be a (real) eigenvector of $H_+$ with an eigenvalue $E\ge -1$.
Using the identity $E+1=\|(T^\dagger+S)\ket{\phi}\|^2$ it is easy to check that
\begin{align}
\ket{\psi}&=\frac{i}{\|(T^\dagger+S)\ket{\phi}\|}(T^\dagger+S)\ket{\phi}\otimes\ket{0}+\ket{\phi}\otimes\ket{1}\ .\label{eq:oneparticlemultiparticle}
\end{align}
and $\bar{\psi}$ are orthogonal eigenvectors of~$H_0+V$ with eigenvalues~$\pm i\sqrt{E+1}$ (cf.~Fact~\ref{fact:1}). Furthermore, since~$T^\dagger-S$ is local and the encoding~\eqref{eq:cpencoding} respects locality (modulo a factor of $2$), localization properties of $\ket{\phi}$ translate directly into those of $\ket{\psi}$ as claimed.

\subsection{Lyapunov exponents \label{sec:lyapunov}}
Since there is a local correspondence between eigenvectors of $H_0+V$ and
eigenvectors of $H_+$, see Eq.~(\ref{eq:oneparticlemultiparticle}), it suffices
to establish localization properties for the latter.
Consider an eigenvector
\begin{align}
H_+ \psi =E\psi\ ,\qquad \psi=\sum_{n=1}^N (-1)^n\,  \psi_n \ket{n}. \label{eq:Beigenvector}
\end{align}
Using the explicit form of $H_+$, see Eq.~(\ref{eq:oneparticleandersonmodel}), one gets
\bea
(\mu_1^2-E)\psi_1&=&\mu_2 \psi_2, \nn \\
(\mu_n^2-E)\psi_n &=& \mu_{n+1} \psi_{n+1} + \mu_n \psi_{n-1}, \quad \mbox{$n=2,\ldots,N-1$}, \nn \\
(\mu_N^2-E-1)\psi_N  &=& \mu_N \psi_{N-1}\ . \nn
\eea
This can be rewritten using the transfer matrix representation as
\[
\left[ \ba{c} \psi_{n+1} \\ \psi_n \\ \ea \right]
=M_n \left[ \ba{c} \psi_{n} \\ \psi_{n-1} \\ \ea \right],
\quad \quad
\mbox{where} \quad
M_n=\left[ \ba{cc}
\frac{\mu_n^2 - E}{\mu_{n+1}} & -\frac{\mu_n}{\mu_{n+1}} \\
1 & 0 \\ \ea \right].
\]
It follows that
\[
\left[ \ba{c} \psi_{n+1} \\ \psi_n \\ \ea \right] = M_n  \cdots M_2 M_1 \left[ \ba{c} \psi_{1} \\ 0 \\ \ea \right].
\]
It will be convenient to introduce an auxiliary transfer matrix
\[
T_n=\left[ \ba{cc}
\mu_{n+1} & 0 \\
0 & 1 \\ \ea \right]\cdot M_n \cdot \left[ \ba{cc}
\mu_{n}^{-1} & 0 \\
0 & 1 \\ \ea \right]=
\left[ \ba{cc}
\frac{\mu_n^2 - E}{\mu_{n}} & -\mu_n \\
\mu_n^{-1} & 0 \\ \ea \right] \in SL(2,\RR).
\]
Then
\be
\label{lyap1}
\frac{\psi_{n+1}}{\psi_1} = \left(\frac{\mu_1}{\mu_{n+1}}\right) (T_n \cdots T_2 T_1)_{1,1}.
\ee
It is important that the matrix $T_n$ depends only on $\mu_n$. Since $\{\mu_n\}$ are i.i.d. random
variables, we can regard the product $T_n\cdots T_2 T_1$ as a random walk on $SL(2,\RR)$.
For later use let us denote
\[
T(\zeta)=\left[ \ba{cc}
\frac{\zeta^2 - E}{\zeta} & -\zeta \\
\zeta^{-1} & 0 \\ \ea \right].
\]
Let $\mathbb{P}^1$ be the real projective plain, i.e., the set of two-dimensional non-zero real vectors
where pairs of vectors lying on the same line are identified. Any matrix $T\in SL(2,\mathbb{R})$ defines
a one-to-one map $T\, : \, \mathbb{P}^1 \to \mathbb{P}^1$.
We will use the following statement (see e.g.,~\cite{damanik11}), which is a variant of F\"urstenberg's theorem:
\begin{theorem}[\cite{Fuerstenberg63}]\label{thm:fuerstenberg}
Let $\rho(\mu_j)$ be the probability density of $\mu_j$.
Suppose $T(\zeta)\in SL(2,\mathbb{R})$ satisfies
\begin{align}
\int \log \|T(\zeta)\| d\rho(\zeta)<\infty\ .\label{eq:integrablelog}
\end{align}
Let $G_\rho$ be the smallest closed subgroup of $SL(2,\mathbb{R})$ containing
all matrices $T(\zeta)$ with $\rho(\zeta)>0$.   Assume that
\begin{enumerate}[(i)]
\item
$G_\rho$ is not compact and
\item
there is no subset $L=\{v_1,v_2\}\subseteq \mathbb{P}^1$  with one or two elements such that $T(v_1)\in L$
 and $T(v_2)\in L$ for every $T\in G_\rho$.
\end{enumerate}
Then there exists $\ell>0$ that does not depend on disorder such that
\begin{align}
\lim_{n\to \infty} \frac1n \log{\|T_n \cdots T_2 T_1\|}=\ell\label{eq:normlimit}
\end{align}
with probability $1$ over disorder realizations.
\end{theorem}
To verify that this theorem can be applied, observe that the integrability condition~\eqref{eq:integrablelog} is satisfied since we assume that $\rho$ has compact support $[\mu-\eta,\mu+\eta]$ and
$\|T(\zeta)\| \le 2\max{(\zeta,\zeta^{-1},(\zeta^2-E)\zeta^{-1})}$, so the integral in Eq.~(\ref{eq:integrablelog})
converges absolutely.
Let $x,y$ be  in the support of $\rho$ and define
\begin{align}
T\equiv T(x)^{-1}T(y)&=\left[\begin{matrix}
a & 0\\
a-a^{-1} & a^{-1}
\end{matrix}\right]\qquad\textrm{ where }a\equiv x/y\ .\label{eq:Txydef}
\end{align}
Because
\begin{align}
T^{p}&=\left[\begin{matrix}
a^p & 0\\
a^p-a^{-p} & a^{-p}
\end{matrix}\right]\in G_\rho\qquad\textrm{ for every }p\ge 1\ ,\label{eq:tpexponent}
\end{align}
and we can assume without loss of generality that $a>1$, we conclude that  $G_\rho$ is not compact.
Similarly, we obtain from~\eqref{eq:tpexponent}
\begin{align*}
\lim_{p\rightarrow\infty} T^{p}(v) = v_1, \qquad\textrm{ where } v_1 = \left(\begin{matrix} 1\\ 1\end{matrix}\right)\
\end{align*}
for any $v\in \mathbb{P}^1$ assuming $a>1$ (that is, $x>y$). Since all powers of~$T$ belong to $G_\rho$, we conclude that if a subset $L\subset \mathbb{P}^1$ satisfying condition~(ii) exists then
$v_1\in L$.
On the other hand,  interchanging the role of $x$ and $y$ in~\eqref{eq:Txydef} gives  $a<1$ and leads to
\begin{align*}
\lim_{p\rightarrow\infty} T^{p}(v) =v_2 \qquad\textrm{ where }v_2 = \left(\begin{matrix} 0\\ 1\end{matrix}\right)
\end{align*}
for any $v\in \mathbb{P}^1$. Hence if a subset $L\subset \mathbb{P}^1$ satisfying condition (ii) exists then
$L=\{v_1,v_2\}$.  Take $x\neq 0$ in the support of~$\rho$. Then $T(x)v_2=-x\left(\begin{matrix}
1\\
0\end{matrix}\right)\notin L$, hence there is no  subset~$L$ satisfying condition~(ii).

A theorem by Oseledec~\cite{Oseledec68} and Ruelle~\cite{Ruelle79} states that if a sequence of matrices $\{T_n\in SL(2,\mathbb{R})\}_{n\geq 1}$ satisfies
\begin{align*}
\lim_{n\rightarrow \infty}\frac{1}{n}\|T_n\|=0
\end{align*}
and~\eqref{eq:normlimit} for some constant $\ell>0$, then there is a basis $\{v_+,v_-\}$ of $\mathbb{R}^2$ such that
\begin{align}
\lim_{n\rightarrow \infty}\frac{1}{n}\log \|T_n\cdots T_1v_{\pm}\|=\pm \ell\ .\label{eq:oseledec}
\end{align}
The number  $\ell$ will be called the Lyapunov exponent of the eigenvector $\psi$.  From Eq.~\eqref{lyap1} and~\eqref{eq:oseledec} we infer that $\ell$ controls the rate of the exponential
growth (decay) of the chosen eigenvector, that is,
\be
\ell= \lim_{n\to \infty} \frac1n \log |\psi_n/\psi_1|.
\ee
Let $\ell(E)$ be the Lyapunov exponent considered as a function of the energy~$E$.
The main goal of this section is to compute~$\ell(E)$ in the limit of weak perturbations,
that is, $\epsilon=\mu+\eta\ll 1$.
We will argue  that
\begin{align}
\min_E\ell(E)\sim \frac{1}{\log{(1/\epsilon)}}\label{eq:lyapunovexponentscaling}
\end{align}
in the limit $\epsilon \to 0$, if  taking the limit one keeps the ratio $\mu/\eta$
between the strengths of the homogeneous and disordered part of the perturbation  fixed. Assuming
that the localization length~$\xi$ scales roughly as the inverse of the Lyapunov exponent, this suggests that
the localization condition  required for Theorem~\ref{thm:storage}  holds  for sufficiently weak perturbations (cf.~Section~\ref{sec:momentlocalization}).

Introduce  variables
\[
z_n=\frac{\psi_{n-1}}{\psi_n} \quad \mbox{for} \quad  n\ge 2 \quad \mbox{and} \quad z_1=0.
\]
Rewriting the eigenvalue equations for $\psi_n$ in terms of $z_n$ we obtain
\begin{align}
z_1 &= 0, \nn \\
&& \nn \\
z_{n+1} &= \frac{\mu_{n+1}}{\mu_n^2 - E - \mu_n z_n}, \quad n=1,\ldots,N-1, \label{eq:eigenvalueeq} \\
&& \nn \\
z_N &= \frac{\mu_N^2 - E -1}{\mu_N}\ . \nn
\end{align}
The Lyapunov exponent is
\begin{align}
\ell(E)\approx \frac{1}{N}\sum_{j=2}^N\log|z_j|\qquad\textrm{ for large }N\ .\label{eq:lyapunovdefinitionz}
\end{align}
We will now use Eqs.~(\ref{eq:eigenvalueeq},\ref{eq:lyapunovdefinitionz}) to argue that the minimal Lyapunov exponent obeys the scaling~\eqref{eq:lyapunovexponentscaling}.

Let us first discuss some numerical  evidence for~\eqref{eq:lyapunovexponentscaling}: we compute $\ell(E)$  numerically for different energies~$E$ using the recursion relation~\eqref{eq:eigenvalueeq}.
We choose the chemical potentials $\{\mu_j\}_j$ identically and independently distributed according to~\eqref{mu}, keeping the ratio~$\mu/\eta$ fixed while varying the average potential~$\mu$ (this corresponds to the perturbation strength~$\epsilon$). Throughout, we set~$N=10^{7}$. Figure~\ref{figlyap} shows that $\ell(E)$ is minimal around~$E\approx\mu^2$; furthermore, this minimum becomes more pronounced as the overall perturbation strength  $\mu$ is decreased.
\begin{figure}[ht]
\centering
\begin{center}
\subfigure[Lyapunov exponent~$\ell(E)$ as a function of~$E$, for $\mu=2^{-3}$ and $\mu/\eta\in\{1.5,2,2.5\}$. The inset shows the behavior around the minimum.]{
\epsfig{file=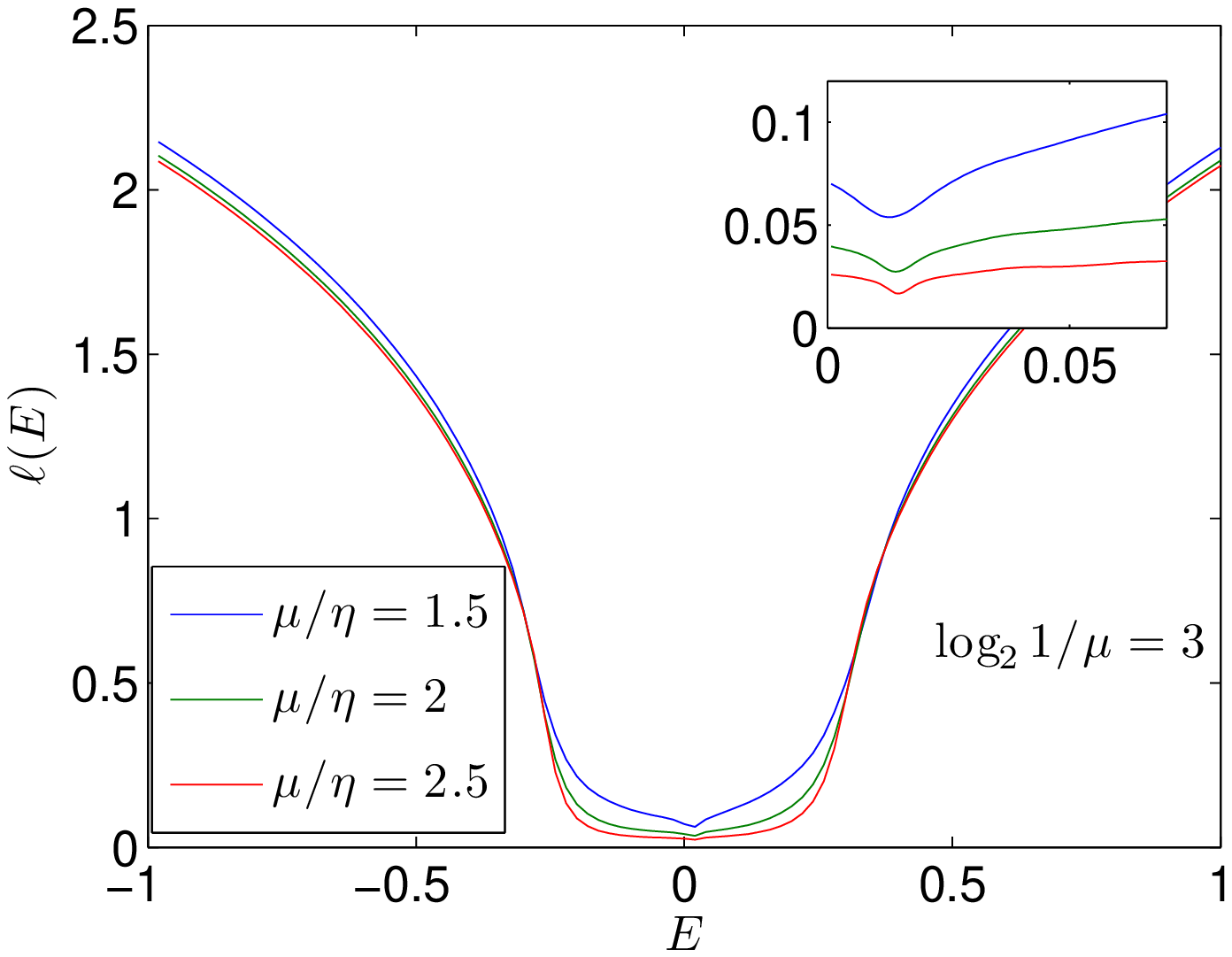,width=8.0cm}
  \label{figlyap:subfig1}
 }\ \ \subfigure[Lyapunov exponent~$\ell(E)$ for $\mu=2^{-6}$.]{
\epsfig{file=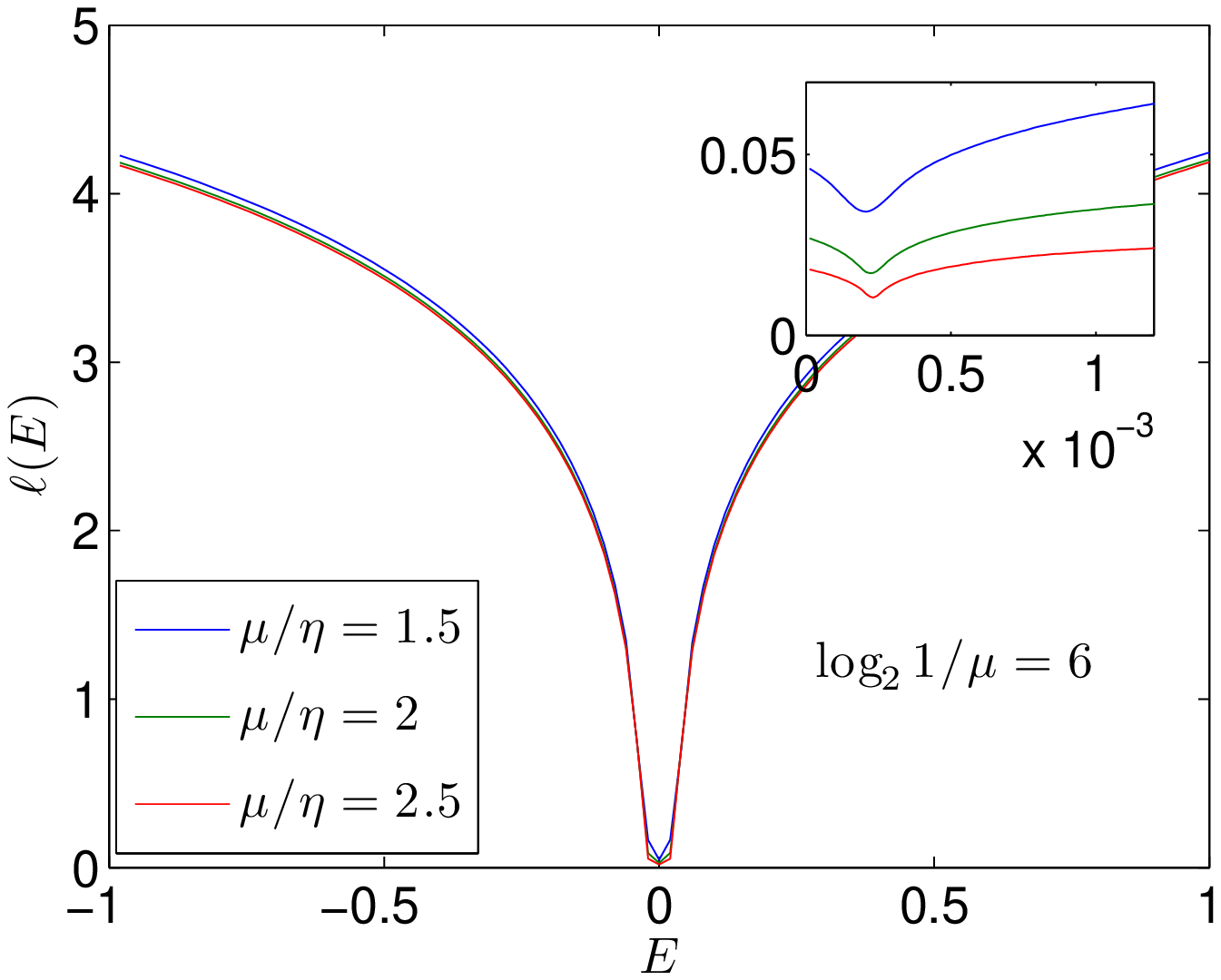,width=7.75cm}
   \label{figlyap:subfig2}}
\end{center}
\label{figlyap}
\caption{The Lyapunov exponent~$\ell(E)$ is positive throughout the interval~$[-1,1]$, and achieves its minimum around~$E\approx \mu^2$.}
\end{figure}
Figure~\ref{figlyapscaling} shows the value of the minimal Lyapunov exponent~$\ell_{\min}=\min_E \ell(E)$  as a function of the perturbation strength. We indeed observe a linear scaling $\ell_{\min}\sim\frac{1}{\log\textfrac{1}{\mu}}$ as~$\mu\rightarrow 0$, confirming~\eqref{eq:lyapunovexponentscaling}.
\begin{figure}[ht]
\centering
\begin{center}
\epsfig{file=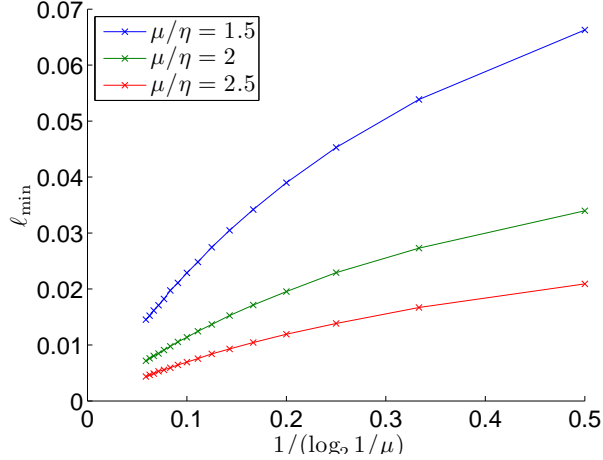,width=8.0cm}
\end{center}
\caption{The minimal Lyapunov exponent $\ell_{\min}=\min_E \ell(E)$ as a function of $\left(\log\textfrac{1}{\mu}\right)^{-1}$, where $\mu$ is the average chemical potential. In the limit $\mu\rightarrow 0$, this becomes linear, following~\eqref{eq:lyapunovexponentscaling}. \label{figlyapscaling}}
\end{figure}

Figure~\ref{figlyap} suggests studying the behavior of $\ell(E)$ for small values of~$E$ (i.e., $E\approx\mu^2$). In the following, we analyze the behavior of the exponent $\ell(E=0)$ at zero energy as a function of the perturbation strength. Unfortunately, these arguments do not seem to extend to the whole range of energies of interest.

\subsubsection*{Lyapunov exponent at $E=0$}
Here we give some heuristic arguments showing that with overwhelming probability over disorder realizations,  the Lyapunov exponent at energy $E=0$ scales as
\begin{align*}
\ell(E=0)\sim \frac{1}{\log\textfrac{1}{\epsilon}}\qquad\textrm{ as }\epsilon\rightarrow 0\ .
\end{align*}
 This is based on the techniques developed by Eggarter et al~\cite{eggarterried78} which relate the Anderson model to certain random walks.

For simplicity we will assume that
\begin{align}
\epsilon/2 \le \mu_n \le \epsilon \quad \quad \mbox{for all $n$}\label{eq:muassumption}
\end{align}
where $0<\epsilon\ll 1$ controls the perturbation strength. Specializing~\eqref{eq:eigenvalueeq} to $E=0$ gives the recursion relation
\be
\label{recursion}
z_{n+1}=\frac{\mu_{n+1}}{\mu_n^2 - \mu_n z_n}, \quad \quad z_2=\mu_2 \mu_1^{-2}\ .
\ee
A few starting elements in this sequence are
\be
\label{starting}
z_2=\mu_2 \mu_1^{-2}, \quad z_3=-\mu_3 \mu_2^{-2} \mu_1^{2} \cdot \frac1{1-\mu_1^2}, \quad
z_4=\mu_4 \mu_3^{-2} \mu_2^2 \mu_1^{-2} \cdot \frac{1-\mu_1^2}{1-\mu_2^2(1-\mu_1^{-2})}\ .
\ee

We will prove that for all sufficiently small $n$ the sequence $z_n$ has alternating signs, namely,
$z_n>0$ for even  $n$ and $z_n<0$ for odd $n$. Furthermore, we will show that the behavior of $x_n\equiv \ln{(z_n)}$
for {\em even} $n$ can be modeled by an unbiased random walk on a line with a   potential barrier
on the interval $(a,\infty)$, where~$a\sim |\log\epsilon|$. The potential is negligible outside of the interval
$(a,\infty)$ and grows roughly as $|z-a|^2$ for $z\in (a,\infty)$.
The walk starts from $x_2\approx  |\log\epsilon|$ near the left boundary of the potential barrier.
As soon as the random walk reaches values $x_n\approx -|\log\epsilon|$,
it jumps back to its starting point reaching $x_{n+1}\approx |\log\epsilon|$ in a single step.
This jump also breaks the alternating sign pattern of $z_n$ such that both
$z_n$ and $z_{n+1}$ are positive. The typical number of steps it takes for the walk to diffuse from
$|\log\epsilon|$ to $-|\log\epsilon|$ is  $n^*\sim |\log\epsilon|^2$. The random walk picture allows
one to partition the sequence $\{z_n\}$  into {\em cycles} of length roughly
$|\log\epsilon|^2$. Within each cycle the sequence $\{z_n\}$ follows the alternating sign pattern.
A cycle ends whenever two adjacent $z_n$'s have the same sign.
We will prove that each cycle increases $\log{|\psi_n|}$ by  roughly $|\log\epsilon|$.
Hence the average increase of $\log{|\psi_n|}$ per step is roughly
$|\log\epsilon|^{-1}$ providing the estimate of the Lyapunov exponent.

Iterating Eq.~(\ref{recursion}) twice one gets
\be
\label{walk1}
z_{n+2}=\left( \frac{\mu_{n+2} \mu_n}{\mu_{n+1}^2}\right) \cdot z_n \cdot \beta_n(z_n),
\ee
where
\be
\label{walk2}
\beta_n(z_n)=\frac{1-\mu_n/z_n}{1+\mu_n z_n - \mu_n^2}\ .
\ee
It is important that
\be
\label{beta}
\beta_n(z)\approx 1 \quad \mbox{whenever} \quad \epsilon \ll |z| \ll \epsilon^{-1}\ .
\ee
Furthermore,
\be
\label{beta1}
0\le \beta_n(z)\le 1 \quad \mbox{for all $z\ge \mu_n$}\ .
\ee
Below we consider only even values of $n$.
Introducing variables $x_n\equiv \ln{|z_n|}$ we obtain
\[
x_{n+2} = x_n + e_{n+2} + e_n - 2e_{n+1} +\varphi_n(x_n),
\]
where
\[
e_n= \ln|\mu_n/\epsilon| \quad \mbox{and} \quad \varphi_n(x)=\log\beta_n(e^x).
\]
By definition we have  $|e_n|=O(1)$.
Simple algebra shows that
\[
x_n=x_2 -e_n -e_2+2\sum_{j=2}^{n} (-1)^j \, e_{j} + \varphi_{n-2}(x_{n-2}) +  \varphi_{n-4}(x_{n-4}) + \ldots + \varphi_2(x_2)
\]
for any even $n$.
In order to interpret these equations in terms of a random walk, define auxiliary variables
$y_n$ by
\be
\label{ywalk}
y_{n+1} = y_n + \xi_n + \varphi_{2n}(y_n), \quad \mbox{where} \quad
y_1=x_2 \quad \mbox{and} \quad \xi_n =2(e_{2n+2}-e_{2n+1}).
\ee
One can easily check that
\[
|x_{2n} - y_n|\le O(1) \quad \mbox{for all $n$}\ .
\]
In addition, $\xi_n$ are i.i.d. random variables with $\EE(\xi_n)=0$ and variance $O(1)$
that can be easily computed given the distribution of $\mu_n$.
Hence $y_1,y_2,\ldots,y_n$ is an unbiased random walk on a line
where each step has length $O(1)$.
The term $\varphi_{2n}(y_n)$ can be regarded as an external force, or, position-dependent bias.
This force turns on only then $y_n$ becomes close to $\pm |\log\epsilon|$, see Eq.~(\ref{beta}). In both cases the force pushes the walker towards negative values of $y_n$, see Eq.~(\ref{beta1}).
In addition, the magnitude of the force grows approximately linearly with $y_n$
in the interval $y_n\in (a,\infty)$ with $a\approx |\log\epsilon|$. Indeed, in this interval
one has
\be
\label{beta2}
\varphi_{2n}(y)\approx -\log(1+\mu_n e^y) \approx -|y-|\log\mu_n||.
\ee
Such a force is roughly equivalent to
an energy barrier that reflects the walk
each time it enters the interval $(a,\infty)$.
From Eq.~(\ref{starting}) we infer that the starting point of the walk
obeys
\be
\label{ybound1}
y_1\approx |\log\epsilon|.
\ee
Consider now what happens when $y_n$
reaches values of order $- |\log\epsilon|$. Equivalently, $z_n$ reaches values of order~$\epsilon$.
Assuming that $z_n$ is positive and satisfies $z_n\le \mu_n$ we get
\be
z_{n+1}\ge \frac{\mu_{n+1}}{\mu_n^2} \sim \epsilon^{-1}\ .
\ee
Equivalently, $y=\log z_n$ jumps from $- |\log\epsilon|$ to $+|\log\epsilon|$
in a single step. Note also that both $z_n$ and $z_{n+1}$ are positive which
breaks the alternating sign pattern. We can now start the next walk described by Eq.~(\ref{ywalk}),
now using $z_n$ with odd $n$.
The number of steps needed for each of such walks to diffuse from $+|\log\epsilon|$
to  $- |\log\epsilon|$ is
\[
n^*\sim |\log\epsilon|^2.
\]
Since the walk spends most of its time in the region with negligible external force,
we can compute the increase of $\log|\psi_n|$ per cycle by setting $\beta_n(z_n)=1$.
Then we have
\[
\frac{|\psi_1|}{|\psi_{n}|} =\exp{\left[ \sum_{j=2}^{n} x_j \right]} \equiv \exp{[X_{n}]}\ .
\]
Simple algebra shows that for any even $n$ one has
\[
X_n=\sum_{j=2}^n x_j = \frac12 x_n + \frac12 e_n + e_1= \frac12 x_n +O(1).
\]
Hence we have
\[
X_{n}\approx (1/2) y_{n/2} \approx (-1/2)|\log\epsilon|.
\]
We conclude that
the logarithm of $|\psi_n|$ increases by roughly
\[
-X_{n} \approx (1/2)|\log\epsilon|.
\]
per cycle, where cycle consists of
$n^*\sim |\log\epsilon|^2$ steps. This means that the Lyapunov exponent is roughly
\[
\ell(E=0) \sim \frac1{n^*} \log{\frac{|\psi_{n^*}|}{|\psi_{1}|}} \sim
-\frac{X_{n^*}}{n^*} =
 \frac{\log\epsilon}{ |\log\epsilon|^2} \sim \frac1{|\log\epsilon|}\ .
\]

\subsection{Localization of multipoint correlation functions}
\label{sec:momentlocalization}
Here we show how to derive the multi-point dynamical localization condition~\eqref{DL}
 that involves the determinant
$\det R[p,q]$ from the analogous localization condition that involves more standard multi-point correlation
functions.   Recall that $|p-q|_1=\sum_{a=1}^m |p_a-q_a|$, where $p$ and $q$  are $m$-tuples of integers.
\begin{lemma}\label{lem:DLversussingle}
Let  $p,q$ be $m$-tuples of distinct integers in the interval $[1,2N]$.
Suppose there exists constants $\cwave,\xiwave$ such that
\be
\label{DL2}
\ExpE{\prod_{a=1}^m \left|R_{p_a,q_a}\right|} \le \cwave^m \exp{\left(-\frac{|p-q|_1}{\xiwave}\right)}
\ee
for all $N$ and for all $m$-tuples $p,q$ such that $|p-q|_1\ge N/8$.
Then the multi-point dynamical localization condition~\eqref{DL}
holds for some constants $C=O(\cwave\xiwave^2)$ and $\xi=O(\xiwave)$.
\end{lemma}
{\em Remark:} The arguments used in the proof of the lemma can also be applied to any fixed
disorder realization. In particular, if {\em all} matrix elements of $R\equiv R(t)$
obey a bound $|R_{p,q}|\le   \cwave e^{-|p-q|/\xiwave}$ for some disorder realization, then
Eq.~(\ref{DL2}) holds without the expectation value. The lemma then
implies that the multi-point dynamical localization holds for the chosen disorder realization
with some constants  $C=O(\cwave\xiwave^2)$, $\xi=O(\xiwave)$.
Repeating all the steps used in the proof of Theorem~\ref{thm:storage} for a fixed
disorder realization one would get the same bound Eq.~(\ref{lower_bound}) on the storage fidelity.
\begin{proof}
For any $\beta>0$ define a partition function of pairings at inverse temperature~$\beta$ as
\begin{align*}
Z_\beta(p,q) &=\sum_{\sigma\in S_{m}} e^{-\beta E_\sigma(p,q)}\qquad\textrm{ where }\qquad E_{\sigma}(p,q)=\sum_{a=1}^m |p_a-q_{\sigma(a)}|,
\end{align*}
where  $S_m$ is the group of permutations of $m$ objects. The bound Eq.~(\ref{DL2}) yields
\begin{align}
\ExpE{ |\det R[p,q]|}\le \sum_{\sigma\in S_m} \ExpE{ \prod_{a=1}^m \left|R_{p_a,q_{\sigma(a)}}\right|}\leq \cwave^m \sum_{\sigma\in S_m} e^{-|p-\sigma(q)|_1/\xiwave}=\cwave^m Z_{1/\xiwave}(p,q)\ .\label{eq:Rtopartitionfunction}
\end{align}
The statement of the lemma now follows from the following fact.
\begin{prop}\label{lem:partitionfunctionbound}
Let $0<\beta\le 1$ be a constant independent of $m$ and $N$.
Let $p$, $q$ be $m$-tuples of integers in the interval $[1,N]$ such that
$p_1<p_2<\cdots<p_m$ and $q_1<q_2<\cdots<q_m$.
Suppose that $|p-q|_1=\Omega(N)$. Then
\be
\label{Zbound}
Z_\beta (p,q)\le O(1/\beta)^{2m} e^{-N\Omega(\beta)}\ .
\ee
\end{prop}
\begin{proof}
We claim that it suffices to prove Eq.~(\ref{Zbound}) for the special case when
all integers in $p$ are {\em even} and all integers in $q$ are {\em odd}. Indeed,
$Z_\beta(p,q)=Z_{\beta/2}(p',q')$, where $p_a'=2p_a$ and $q_a'=2q_a$
(the latter partition function is defined for a system size $N'=2N$).
Hence, without loss of generality,  all integers in $p$ and $q$
are even. Furthermore, shifting each integer in $q$ by $\pm 1$ we change the energy $E_{\sigma}(p,q)$
at most by $m$. This changes the partition function at most by a factor $e^{\beta m}$ which can be
absorbed into the factor $O(1/\beta)^m$.

In the following, we prove Eq.~(\ref{Zbound}) assuming that all $p_a$ are even and all $q_a$ are odd.
Let us fix some permutation $\sigma \in S_m$.
We will say that an interval $[n,n']$ is a {\em pairing}
iff $n=p_a$, $n'=q_{\sigma(a)}$ or $n'=p_a$, $n=q_{\sigma(a)}$ for
some $1\le a\le m$. We can consider $\sigma$ as a collection of $m$ pairings.
For any $i\in [0,N]$ let $x_i$ be the number of pairings that contain the interval $[i,i+1]$.
Obviously, $0\le x_i\le m$ and $x_0=x_N=0$.
\begin{prop}
Let $E\equiv E_\sigma(p,q)$.
The sequence $x_0,x_1,\ldots,x_{N}$ has the following properties:
\bea
x_{i}-x_{i-1} &\in & \{-1,0,+1\}, \label{xstep} \\
\sum_{i=0}^{N} x_i &=& E, \label{xenergy} \\
\sum_{i=1}^{N} |x_{i}- x_{i-1}| &=& 2m. \label{xsum}
\eea
\label{prop:1}
\end{prop}
\begin{proof}
Let us prove Eq.~(\ref{xstep}).  Suppose $i$ is even, so that $i\notin q$.  If $i\notin p$ then
no pairing can terminate at $i$. Hence a pairing contains the interval $[i,i+1]$ iff it
contains the interval $[i-1,i]$, that is, $x_i=x_{i-1}$. Suppose now that $i\in p$,
that is there is one pairing that terminates at $i$. If this pairing lies on the left of $i$, then
$x_i=x_{i-1}-1$. If this pairing lies on the right of $i$, then $x_i=x_{i-1}+1$. The case of odd~$i$
is analogous. Property~\eqref{xenergy} is obvious since $E_{\sigma}(p,q)$
is nothing but the total length of all pairings in~$\sigma$. Property~\eqref{xsum} follows
from the fact that $x_{i}\ne x_{i-1}$ only if some pairing terminates at $i$ and the total number
of points where some pairing terminates is $2m$.
\end{proof}
We can think of $x=\{x_i\}$ as a step function with $2m+1$ steps.
Let $h_j$ and $w_j$ be the height and the width of the $j$-th step.
From Proposition~\ref{prop:1} we infer that
\be
\label{hwrule}
0\le h_j \le m, \quad h_0=h_{2m}=0, \quad h_i=h_{i-1}\pm 1, \quad \sum_{i=0}^{2m} w_i h_i =E,
\quad \sum_{i=0}^{2m} w_i =N.
\ee
Note that the locations of the steps are fully determined by $p$ and $q$
since $x_{i-1}\ne x_i$ iff $i\in p$ or $i\in q$.
Furthermore, the number of permutations $\sigma \in S_m$ that give rise
to the same sequence $x=\{x_i\}$ is at most
\be
\omega(x)=\prod_{i=1}^{2m} \max{(h_i,h_{i-1})}\ .
\ee
Indeed, each time we switch from one step of $x$  to another we either terminate or create some pairing
and we have a freedom to choose which pairing to terminate or create without changing the sequence~$x$.
The number of possible choices when we switch from the $i$-th step to the $(i+1)$-th step is at most $\max{(h_i,h_{i-1})}$.
Hence we can upper bound  the partition function as
\[
Z_\beta(p,q) \le  \sum_{n\ge E_0} e^{-\beta n} \sum_{\substack{0\le h_0,\ldots,h_{2m} \le m \\ \\
h_0=h_{2m}=0, \\ \\
h_i=h_{i-1}\pm 1\\ \\
\sum_i w_i h_i =n\\}}
\; \; \;
\prod_{i=1}^{2m} \max{(h_i,h_{i-1})}\ .
\]
where $E_0=|p-q|_1=\Omega(N)$ is the ground state energy
(one can easily check that for ordered $p$ and $q$ the minimum of $E_\sigma(p,q)$
is achieved for the trivial permutation $\sigma$).
Taking into account that $w_i\ge 1$ for all $i$ and using the inequality
$\max{(h_i,h_{i-1})} \le h_i+1$
we get
\[
\prod_{i=1}^{2m} \max{(h_i,h_{i-1})} \le \left(\frac1{2m}\sum_{i=1}^{2m}  \max{(h_i,h_{i-1})} \right)^{2m}
\le
\left(\frac1{2m}\sum_{i=1}^{2m}  w_i+w_ih_i \right)^{2m}\ .
\]
Using the constraints $\sum_{i=1}^{2m} w_i=N$, $\sum_{i=1}^{2m} w_i h_i=n$,
and the inequality $(2m)^{2m}\ge (2m)!$ we arrive at
\[
\prod_{i=1}^{2m} \max{(h_i,h_{i-1})} \le \frac{(N+n)^{2m}}{(2m)!}\ .
\]
Let $\lambda>0$ be a parameter to be chosen later. Then
\[
\prod_{i=1}^{2m} \max{(h_i,h_{i-1})} \le \lambda^{-2m} \sum_{k=0}^\infty \frac{\lambda^k (N+n)^{k}}{k!}=
\lambda^{-2m} e^{\lambda(N+n)}\ .
\]
Since the number of sequences $h_0,\ldots,h_{2m}$ satisfying conditions Eq.~(\ref{hwrule}) is
at most $2^{2m}$, we obtain
\[
Z_\beta(p,q) \le  (2/\lambda)^{2m} e^{\lambda N} \sum_{n\ge E_0} e^{-(\beta-\lambda)n}
= \frac{(2/\lambda)^{2m}}{1-e^{-(\beta-\lambda)}}\, e^{-(\beta-\lambda)E_0 + \lambda N}\ .
\]
Since we are promised that $E_0\ge \alpha N$ for some constant $\alpha=O(1)$, we can choose
$\lambda=\alpha \beta/(2+\alpha)$ which yields
\[
Z(p,q)\le  \frac{(2/\lambda)^{2m}}{1-e^{-2\beta(2+\alpha)^{-1}}}\,e^{-\alpha \beta(2+\alpha)^{-1} N}.
\]
This proves Eq.~(\ref{Zbound}).
\end{proof}
\end{proof}

\section{Numerical simulations\label{sec:numerics}}
In Section~\ref{sec:measurementsim}, we explain how to efficiently sample from the distribution of syndromes. In Section~\ref{sec:montecarlo}, we show how to compute fidelities efficiently. In combination, this gives a Monte Carlo method for estimating storage fidelities. In Section~\ref{sec:numericalresults}, we discuss numerical results obtained using this method.

\subsection{Simulation of the syndrome measurement\label{sec:measurementsim}}

Let $t$ be some fixed time and let $\pi(s)$ be the probability of measuring a syndrome
$s=(s_1,\ldots,s_{N-1})$ on the time-evolved state $|g(t)\ra$. Our first goal is to describe an efficient
algorithm that allows one to sample $s$ from the distribution $\pi(s)$.
An important fact is that both
time evolution and the stabilizer measurements belong to a class of operations known as {\em fermionic linear
optics} for which efficient simulation algorithms have been described by Knill~\cite{Knill01}
as well as Terhal and DiVincenzo~\cite{TerhalDiVincenzo02}.
For example, applying the algorithm of~\cite{TerhalDiVincenzo02} to our settings, we can reduce the problem of sampling $s$
from $\pi(s)$ to a series of simpler tasks: sample a bit $s_j$ from the conditional distribution
of $s_j$ given $s_1,\ldots,s_{j-1}$, where $j=1,\ldots,N-1$.
Using the techniques of~\cite{TerhalDiVincenzo02}, the conditional probability of, say, $s_j=0$ can be computed as a ratio of two determinants representing probabilities of  outcomes $s_1,\ldots,s_{j-1},0$ and $s_1,\ldots,s_{j-1}$.
Once this conditional probability is known, the bit $s_j$ can be set by tossing a coin with an appropriate bias.
Setting the bits of $s$ one by one starting from $s_1$, the computational cost of generating one full syndrome sample $s$ in this fashion is~$O(N^4)$ since this involves the computation of~$O(N)$ determinants of matrices of size $O(N)$,
see~\cite{TerhalDiVincenzo02} for details.

Here we propose a simplified version of this algorithm in which the computational cost of generating
one full syndrome sample is only $O(N^3)$. Our algorithm might also be more stable computationally
since it avoids computing the ratio of probabilities for exponentially unlikely events.
\begin{lemma}
\label{lemma:sampling}
Let $\rho$ be a fermionic Gaussian state of $2N$ Majorana modes and
$\pi(s)=\trace{(\rho\,  \hat{Q}_s)}$ be the probability of measuring a syndrome $s\in \{0,1\}^{N-1}$ on $\rho$, see Section~\ref{subs:EC}.
There exists an efficient algorithm that
takes as input the covariance matrix of $\rho$ and returns a sample $s$ drawn from the distribution $\pi(s)$.
The algorithm requires roughly $N^3$ arithmetic operations on real numbers
and the generation of $N-1$ random bits.
\end{lemma}
\begin{proof}
Define a random sequence of fermionic Gaussian states $\rho_0,\rho_1,\ldots,\rho_{N-1}$
such that $\rho_0=\rho$ and $\rho_j$ is obtained from $\rho_{j-1}$
by a non-destructive measurement of a syndrome bit $s_{j}$. More formally,
let $\hat{\Pi}_j(s_j) =(1/2) (I - i(-1)^{s_j} c_{2j} c_{2j+1})$ be the projector onto the eigen-subspace
with a fixed syndrome bit $s_j$ and $p_j(s_j)=\trace(\rho_{j-1} \hat{\Pi}_j(s_j))$ be the corresponding
probability.  Then we define
\be
\label{rho_j}
\rho_j = \frac1{p_j(s_j)}\,  \hat{\Pi}_j(s_j)  \rho_{j-1} \hat{\Pi}_j(s_j)
\ee
where $s_j$ is a random bit with a probability distribution $(p_j(0),p_j(1))$.
Let $M^{(j)}$ be the covariance matrix of $\rho_j$,
that is, $M^{(j)}_{p,q}=(-i/2)\trace(\rho_j [c_p,c_q])$.
Note that the matrix $M^{(0)}$ is the input of the algorithm.
By definition, we have
\be
\label{newpj}
2p_j(s_j) =1+(-1)^{s_j} \, M^{(j-1)}_{2j,2j+1}\ .
\ee
Choose any $2j+1<p<q\le 2N$. Using Eq.~(\ref{rho_j}) we arrive at
\[
M^{(j)}_{p,q} =  \frac1{p_j(s_j)} \trace{(-i)c_p c_q (I -i(-1)^{s_j} c_{2j} c_{2j+1}) \rho_{j-1}}\ .
\]
Since $\rho_{j-1}$ is a Gaussian state, we can employ Wick's theorem to compute the
trace $\trace(c_{2j} c_{2j+1} c_p c_q \rho_{j-1})$. After simple algebra one gets
\be
\label{newMj}
M^{(j)}_{p,q} = M^{(j-1)}_{p,q} - \frac{(-1)^{s_j}}{2p_j(s_j)} \, M^{(j-1)}_{2j,p}  M^{(j-1)}_{2j+1,q}
+  \frac{(-1)^{s_j}}{2p_j(s_j)} \, M^{(j-1)}_{2j,q}  M^{(j-1)}_{2j+1,p}\ .
\ee
Here it suffices to only compute matrix elements $M^{(j)}_{p,q}$  with
\[
p=1, \; \; 2j+2\le q\le 2N \quad \mbox{and} \quad 2j+2\le p<q\le 2N.
\]
Combining Eqs.~(\ref{newpj},\ref{newMj}) we obtain inductive rules
for computing conditional probabilities $p_j(s_j)$.
The overall computation requires roughly $N^3/6$ steps,
where each step amounts to computing the righthand side of Eq.~(\ref{newMj})
for some fixed $p,q,j$. This requires roughly $N^3$ additions, multiplications, and divisions
on complex numbers.
\end{proof}

Let us discuss how Lemma~\ref{lemma:sampling} can be used in our setting.
Let $|g(t)\ra$ be the final state of the memory before error correction.
Taking into account superselection rules we get
\be
\label{pi(s)}
\pi(s) \equiv \la g(t)|\hat{Q}_s |g(t)\ra = \sum_{\sigma =0,1} \, |\alpha_\sigma |^2 \la g_\sigma (t)|\hat{Q}_s|g_\sigma(t)\ra.
\ee
Recall that $|g_\sigma\ra$ is the ground state of $\hat{H}_0$ with fermionic parity $\sigma$,
$\alpha_\sigma$ are amplitudes of the initial encoded state,
and $|g_\sigma(t)\ra =e^{i(\hat{H}_0+\hat{V})t}\, |g_\sigma\ra$.
It follows that $\pi(s)$ is a probabilistic mixture of two distributions $\pi_0(s)$ and $\pi_1(s)$,
where
$\pi_\sigma(s)= \la g_\sigma (t)|\hat{Q}_s|g_\sigma(t)\ra$. The time-evolved
states $|g_\sigma(t)\ra$ are Gaussian and we can efficiently compute their
covariance matrices, namely, $M^\sigma(t)=R(t) M^\sigma(0) R(t)^T$,
where $R(t)=e^{(H_0+V)t}\in SO(2N)$. Here $H_0$ and $V$
are one-particle Hamiltonians corresponding to $\hat{H}_0$ and $\hat{V}$,
see~\eqref{eq:antisymmetricmatrices},
and $M^\sigma(0)\equiv M^\sigma$ are anti-symmetric matrices with non-zero entries (above the diagonal)
\begin{align}
M^\sigma_{1,2N}&=(-1)^\sigma\quad\textrm{ and }\quad M^\sigma_{2j,2j+1}=1\quad\textrm{ for }j=1,\ldots N-1\ \label{eq:covariancematrixt0}
\end{align}
The matrix exponential $R(t)$ can be computed in time $O(N^3)$ by
finding Williamson eigenvectors of  $H_0+V$,
see Section~\ref{sec:quadraticfermionhamiltonians}.
Applying Lemma~\ref{lemma:sampling}, we can sample
$s$ from $\pi_\sigma(s)$ and thus we can sample $s$ from $\pi(s)$.

It is worth mentioning that an efficient simulation of the syndrome measurement alone
can be used to assess the storage time of the memory by estimating the probability
of `bad syndromes' as defined in~\cite{CLBT09} (in our context, a syndrome $s$
is bad iff it can only be caused by an error of weight roughly~$N/2$).
In particular, if the overall probability of bad syndromes is exponentially small (in $N$)
for all $t\in \RR$, then the storage time of the memory grows at least exponentially
with  $N$, see~\cite{CLBT09} for details.

\subsection{Monte Carlo computation of the storage fidelity\label{sec:montecarlo}}
Represent the storage fidelity as an expectation value
\begin{align}
F_{\ket{g}}(t)&=\EE(f_s)\ ,\label{eq:montecarloexpression}
\end{align}
where $s$ is drawn from the distribution $\pi(s)=\la g(t)|\hat{Q}_s|g(t)\ra$ and $f_s\in [0,1]$ is
the fidelity between the initial encoded state and a normalized error corrected final state
for given a syndrome $s$. Here the expectation value is taken only over the distribution of $s$
(fixed disorder configuration).
More formally,
\be
\label{f_s}
f_s = \frac{|\bra{g}C(s)\hat{Q}_s\ket{g(t)}|^2}{\pi(s)}=\frac1{\pi(s)} \left|
\sum_{\sigma=0,1} \alpha_\sigma \la g_\sigma|C(s) e^{i(\hat{H}_0+\hat{V})t} |g_\sigma\ra \right|^2.
\ee
Here we omitted $\hat{Q}_s$ since $\hat{Q}_s C(s)|g_\sigma\ra= C(s)|g_\sigma\ra$.
The expectation value  $\EE(f_s)$ can be estimated with precision $\delta$ by
the standard
Monte Carlo method. It requires  $O(1/\delta^2)$ independent samples of $s$
and computation of $f_s$ for each of the samples.
The remaining step is to show that
for any given syndrome~$s$ we can compute~$f_s$ in time $O(N^3)$.

Recall that the Pfaffian of a complex anti-symmetric matrix $A$ of size $2m$ is defined
as
\[
\pf{(A)}=\frac1{2^m m!} \sum_{\tau\in S_{2m}} \mathrm{sgn}(\tau) \, A_{\tau(1),\tau(2)} \cdots A_{\tau(2m-1),\tau(2m)}\ .
\]
The Pfaffian can be easily computed up to an overall sign using the identity $\pf{(A)}^2 =\det{(A)}$.
It is also well known that $\pf{(A)}$ can be computed directly using the analogue of Gaussian elimination
in time $O(m^3)$.
We will exploit the following version of Wick's theorem.
\begin{prop}[\bf Wick's theorem]
\label{prop:wick}
Consider an ordered list of  operators
$L_1,\ldots,L_{2m}$ such that every operator $L_j$ is a linear combination of the
Majorana operators $c_1,\ldots,c_{2N}$ with complex coefficients.
Let $\psi$ be any fermionic Gaussian state.
Define an antisymmetric $2m\times 2m$ complex matrix $A$ such that
$A_{j,k}= \la \psi| L_j L_k |\psi \ra$ for $j<k$.
Then
\be
\label{wick}
\la \psi|L_{1} L_2  \cdots L_{2m}|\psi\ra =\pf{(A)}\ .
\ee
\end{prop}
\begin{proof}
Indeed, since $|\psi\ra$ is a Gaussian state, we can regard it is as the fermionic vacuum
for properly defined complex fermionic modes $a_p$, $a^\dag_p$, $p=1,\ldots,N$.
Rewriting $L_j$ as a linear combination of~$a_p$, $a^\dag_p$, we can now employ Wick's theorem in its standard form.
\end{proof}
Let  $\tilde{c}_1,\ldots,\tilde{c}_{2N}$ be the canonical modes of the
perturbed Hamiltonian $\hat{H}=\hat{H}_0+\hat{V}$, see~\eqref{eq:Hamiltonianasmajorana}.
As was mentioned above, finding the canonical modes amounts to finding Williamson
eigenvectors of the one-particle Hamiltonian $H_0+V$. This can be done in time $O(N^3)$.
Then
\begin{align}
e^{i\hat{H}t}&=\prod_{j=1}^N \left(\cos (\lambda_jt/2)-\tilde{c}_{2j-1}\tilde{c}_{2j}\sin(\lambda_jt/2)\right)=\prod_{j=1}^N \tilde{c}_{2j-1}\left(\cos (\lambda_jt/2)\tilde{c}_{2j-1}-\sin(\lambda_jt/2)\tilde{c}_{2j}\right)\nonumber\\
&\equiv {L}'_1\cdots {L}'_{2N}\ \qquad\textrm{ where }\quad
{L}'_{2j-1}=\tilde{c}_{2j-1},\  {L}'_{2j}=\cos (\lambda_jt/2)\tilde{c}_{2j-1}-\sin(\lambda_jt/2)\tilde{c}_{2j}\ .\label{eq:evolutionfactorization}
\end{align}
Suppose the minimum-weight error consistent with the syndrome~$s$ is a product of elementary errors at locations $j_1,\ldots,j_q$. Then
\begin{align}
C(s)&=E_1\cdots E_{j_q}= \prod_{m=1}^q(-i)c_{2j_m-1}c_{2j_m}\nonumber\\
&\equiv {L}''_1\cdots\hat{L}''_{2q}\quad\textrm{where }\quad {L}''_{2\mu-1}=(-i)c_{2j_\mu-1}\ , {L}''_{2\mu}=c_{2j_\mu}\ .\label{eq:erroropfactorization}
\end{align}
With~\eqref{eq:evolutionfactorization} and \eqref{eq:erroropfactorization} we conclude that there are operators ${L}_1,\ldots,{L}_{2m}$, $m=N+q$ such that
\begin{align*}
\bra{g_\sigma}C(s)\ket{g_\sigma(t)}&=\bra{g_\sigma}{L}_1\cdots {L}_{2m}\ket{g_\sigma}\ ,
\end{align*}
where each operator is a known linear combination of the Majorana operators $c_j$.
Expectation values $A_{j,k}=\la g_\sigma|L_j L_k|g_\sigma\ra$ can be easily computed
using the covariance matrices $M^\sigma$, see Eq.~(\ref{eq:covariancematrixt0}).
Hence we can compute $\bra{g_\sigma}C(s)\ket{g_\sigma(t)}$
using Proposition~\ref{prop:wick} in time $O(N^3)$.

The missing ingredient to compute $f_s$ is the probability $\pi(s)$, see Eq.~(\ref{f_s}).
Taking into account Eq.~(\ref{pi(s)}), it suffices to compute the overlaps
$\la g_\sigma(t)|\hat{Q}_s|g_\sigma(t)\ra$.
Since $|g_\sigma(t)\ra$ has fermionic parity $\sigma$, we conclude that
$\hat{Q}_s|g_\sigma(t)\ra$ is an eigenvector of $(-i)c_1 c_{2N}$
with an eigenvalue $(-1)^{\sigma+\sigma(s)}$, where
$\sigma(s)$ is number of non-zero syndrome bits modulo two.
Hence  $\la g_\sigma(t)|\hat{Q}_s|g_\sigma(t)\ra= |\la g_\sigma(t)|\psi_s\ra|^2$,
where $|\psi_s\ra$ is a Gaussian state with a covariance matrix
$M_s^\sigma$ such that the only non-zero elements of $M_s^\sigma$ above the diagonal are
\[
(M_s^\sigma)_{2j,2j+1} =(-1)^{s_j} \quad \mbox{and} \quad (M_s^\sigma)_{1,2N}=(-1)^{\sigma+\sigma(s)}\ .
\]
Using the standard formula for the overlap between two Gaussian states one arrives at
\[
\la g_\sigma(t)|\hat{Q}_s|g_\sigma(t)\ra= 2^{-N} \sqrt{\det{(M^\sigma(t)+M_s^\sigma)}}\ .
\]
To summarize, we can compute $f_s$ using Eq.~(\ref{f_s}) in time $O(N^3)$.

\subsection{Numerical results\label{sec:numericalresults}}
We have numerically simulated chains of  various sizes~$N$ in the range $4\leq N\leq 256$ with and without disorder.  Throughout, we consider the Hamiltonian~\eqref{KitaevModel} with  $w=|\Delta|$ and $J=1$. We compute the storage fidelity~$F_{\ket{g}}(t)$ for the maximally entangled state~$\ket{g}=\frac{1}{\sqrt{2}}\left(\ket{g_0}\otimes\ket{0_R}+\ket{g_1}\otimes\ket{1_R}\right)$ and corresponding storage times~$T_{\textrm{storage}}(F_0)$ for different thresholds~$F_0$ close to~$1$.

For system sizes up to $N\leq 12$, we compute~$F_{\ket{g}}(t)$ exactly by summing over syndromes, while for larger system sizes, we sample syndromes $10^5$~times for each time~$t$ to estimate the fidelity. Due to the $N^3$-scaling of this algorithm, this takes on the order of a day of a single desktop computer for~$N=256$.

\subsubsection*{Clean case}
We first consider the clean case with perturbation given by a uniform chemical potential~$\mu\in\{0.1,0.13,0.16\}$ and~$N=6$.  As the perturbation strength is $\|V\|=\epsilon=\mu$,
 we have $N\ll 1/\epsilon^2$ and we expect the storage fidelity to be dictated by
 dephasing due to the exponentially small energy splitting~$\delta$, see  Corollary~\ref{corol:boring2} in Section~\ref{subs:results}. Indeed, we recover the behavior $F_{\ket{g}}(t)\approx \cos^2(\delta t/2)$  as shown in Fig.~\ref{figosc:subfig1}. Fig.~\ref{figosc:subfig2} illustrates how this simple oscillatory behavior disappears for  larger~$\epsilon$.
\begin{figure}[H]
\centering
\begin{center}
\subfigure[In the regime $\epsilon^2N\ll 1$, the behavior of the storage fidelity~$F_{\ket{g}}(t)$ is dictated by dephasing due to the exponentially small splitting of the ground state degeneracy. This figure shows the error-corrected fidelity superimposed with the function~$\cos^2(\delta t/2)$, where~$\delta$ is the gap of~$H$. The system size is~$N=6$. ]{
\epsfig{file=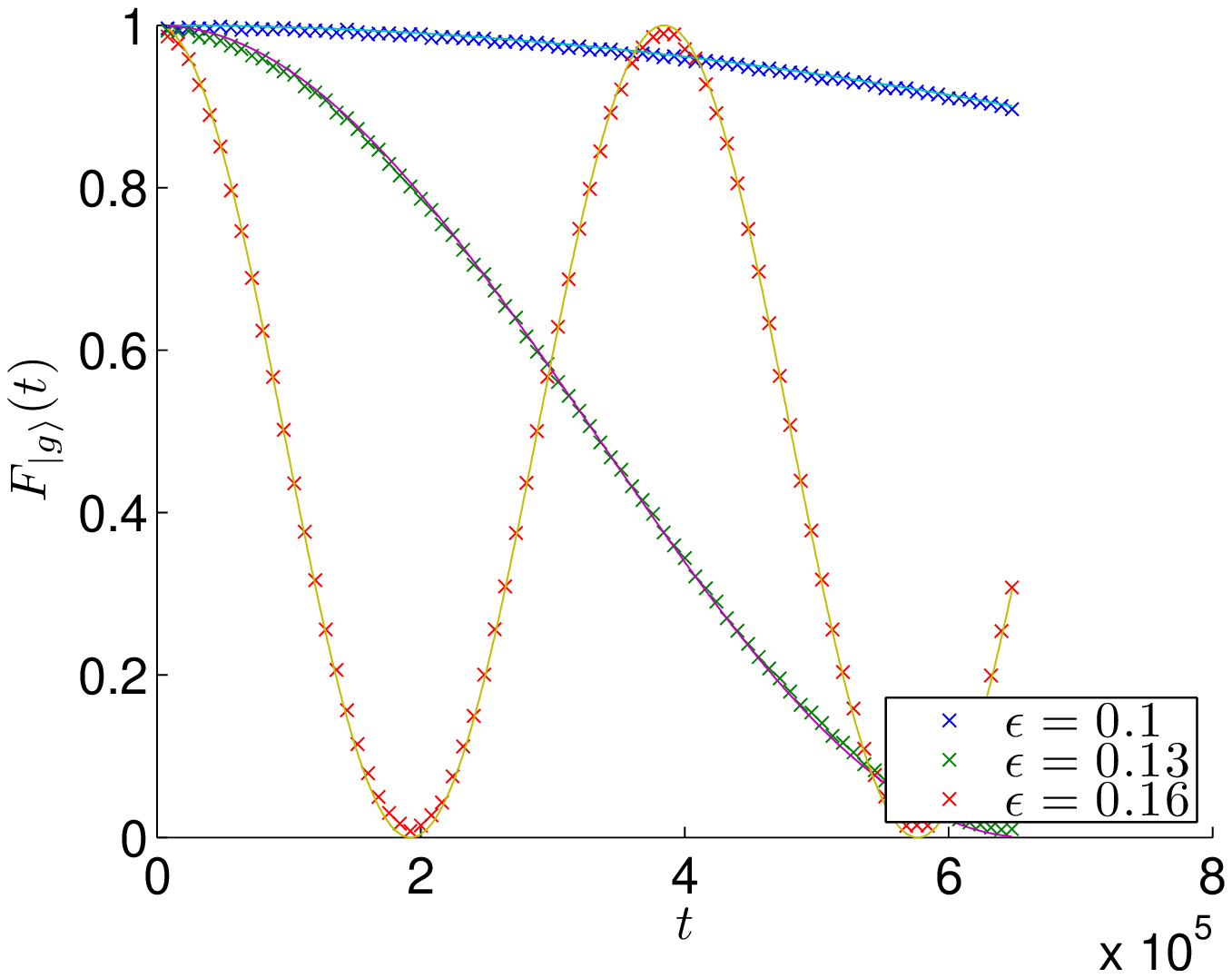,width=8.0cm}
   \label{figosc:subfig1}
 }\ \ \subfigure[When $\epsilon^2N>1$, the storage fidelity~$F_{\ket{g}}(t)$ has a more complex behavior. Here the system size is~$N=12$. In this and the following figures, we use lines as a guide to the eye.]{
\epsfig{file=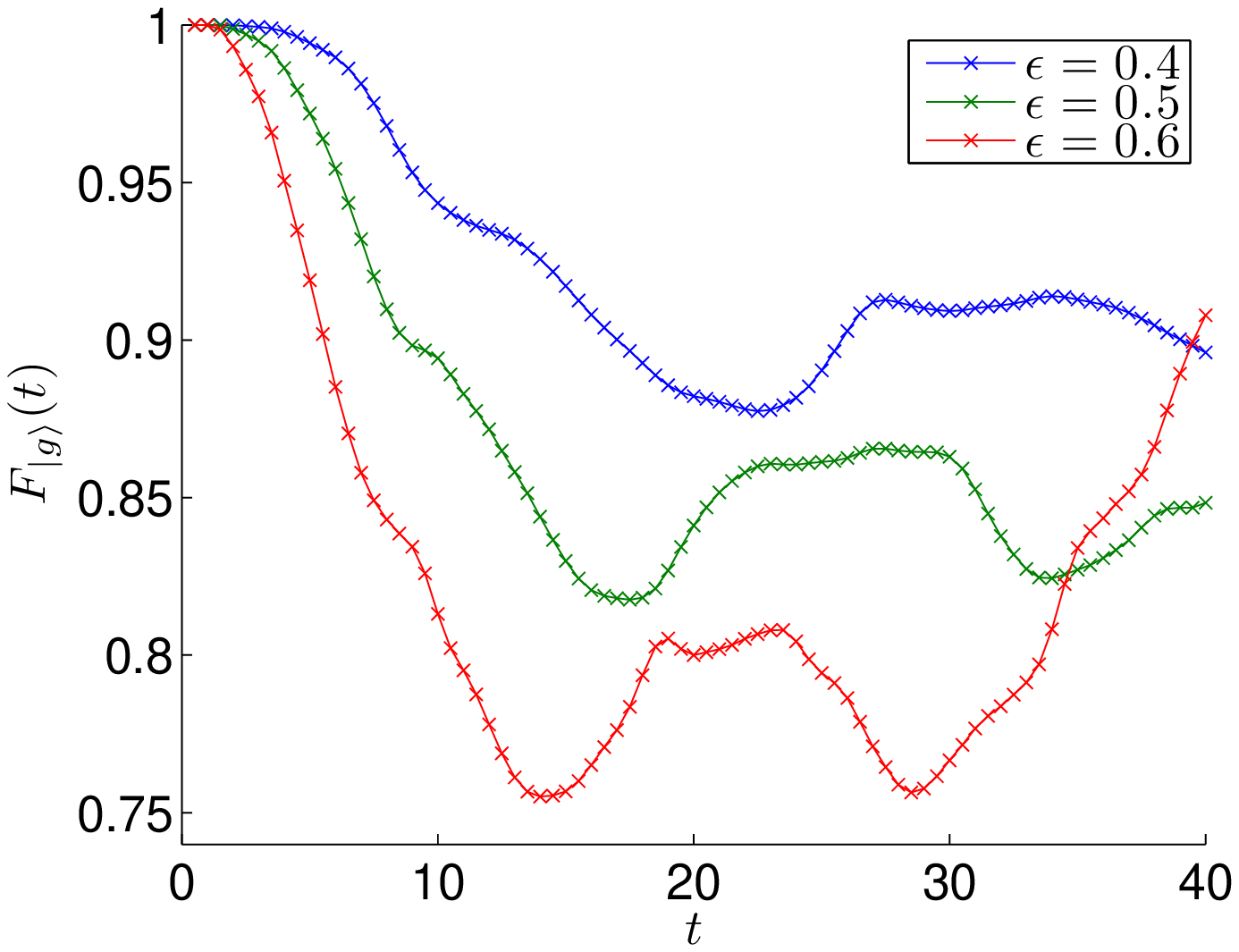,width=8.0cm}
   \label{figosc:subfig2}}
\end{center}
\label{figosc}
\caption{(No disorder) The storage fidelity~$F_{\ket{g}}(t)$ for different  perturbation strengths~$\epsilon$ in the regimes $\epsilon^2 N\ll 1$ and $\epsilon^2N>1$.}
\end{figure}
Fig.~\ref{figstortime:subfig1} shows how  the storage fidelity~$F_{\ket{g}}(t)$ depends on the system size~$N$ for a fixed perturbation strength~($\epsilon=0.7$). Figure~\ref{figstortime:subfig2} is derived from the data of~\ref{figstortime:subfig1} and shows the storage time~$T_{\textrm{storage}}(F_0)$ in the clean case  as a function of~$\log_2 N$ (for different fidelity thresholds~$F_0$). We observe a logarithmic scaling (cf.~\eqref{eq:logscaling}) in the clean case with a sufficiently strong perturbation.
\begin{figure}[H]
\centering
\begin{center}
\subfigure[The storage fidelity~$F_{\ket{g}}(t)$ as a function of time, for different system sizes but fixed perturbation~$\epsilon=0.7$.]{
\epsfig{file=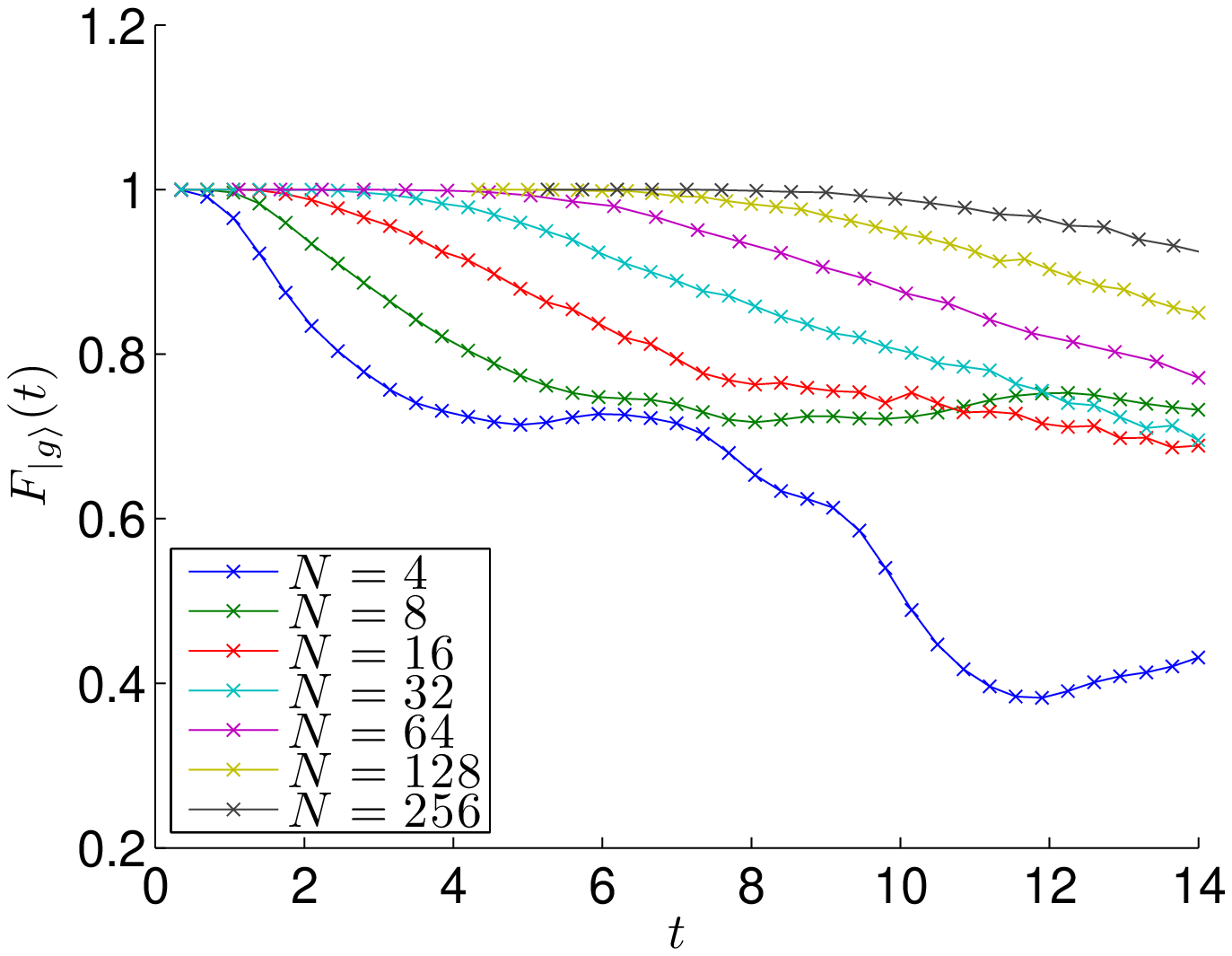,width=8.0cm}
   \label{figstortime:subfig1}
 }\ \ \subfigure[Storage times~$T_{\textrm{storage}}(F_0)$  for different fidelity thresholds~$F_0$ as a function of~$\log_2 N$.  Also shown are linear fits to the values at~$N=16,32,64,128$. For smaller system sizes, the storage time deviates from its asymptotic logarithmic scaling, presumably due to boundary effects.]{
\epsfig{file=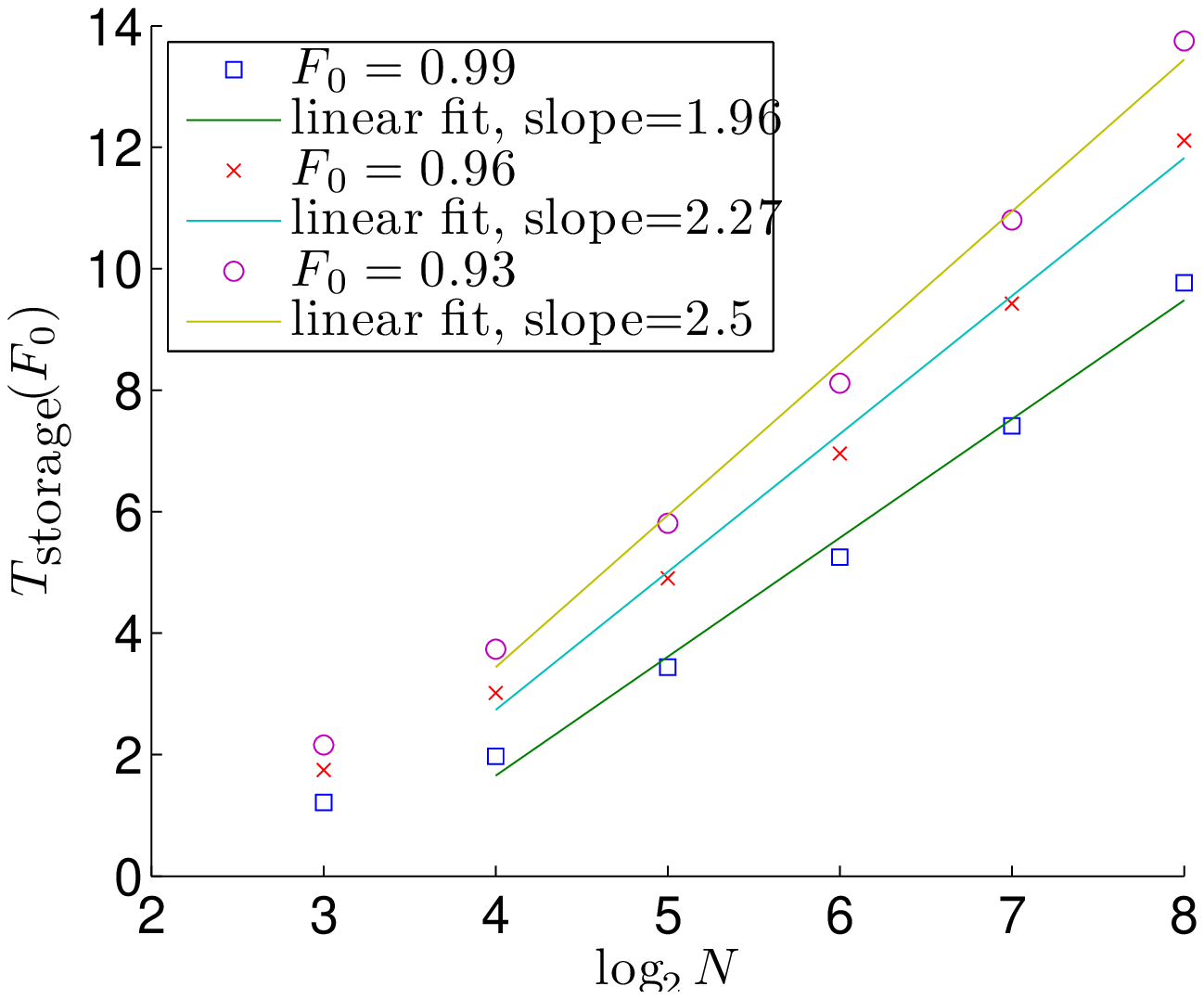,width=8.0cm}
   \label{figstortime:subfig2}}
\end{center}
\label{figstortime}
\caption{(No disorder) Fidelity and storage time vary with system size for  a fixed perturbation strength $\epsilon=0.7$.}
\end{figure}
\begin{figure}[H]
\begin{center}
\epsfig{file=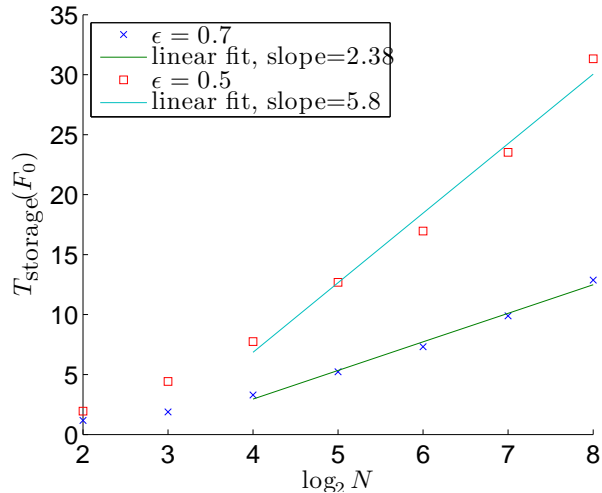,width=8cm}
\end{center}
\caption{(No disorder) Storage times vary with system size, for different perturbation strengths~$\epsilon\in\{0.5,0.7\}$ and fidelity threshold~$F_0=0.95$.}
\end{figure}
This confirms predictions made recently by Kay~\cite{Kay11} based on a mean-field analysis of error correction.
Following~\cite{Kay11} we can describe the time evolution of the encoded state $|g\ra=(|g_0\ra + |g_1\ra)/\sqrt{2}$
as a quantum quench in the transverse field Ising model. Indeed, applying the Jordan-Wigner transformation,
see Section~\ref{subs:JW}, the perturbed Hamiltonian becomes
\[
\hat{H}_0+\hat{V}=-\frac12 \sum_{j=1}^{N-1} X_j X_{j+1} + \frac12 \sum_{j=1}^N \mu_j Z_j,
\]
while the initial encoded state becomes the tensor product of $|+\ra$ states,
\[
|g\ra=|+\ra^{\otimes N}.
\]
Each elementary error $E_j=Z_j$ flips the corresponding qubit mapping $|+\ra$ to $|-\ra$.
The average fraction of flipped qubits in the final state is $(1-m^x(t))/2$, where
$m^x(t)$ is the magnetization,
\be
\label{rhoxt}
m^x(t)=\frac1N \sum_{j=1}^N \la g(t)|X_j|g(t)\ra, \quad  \quad |g(t)\ra= e^{i(\hat{H}_0+\hat{V})t}\,|g\ra.
\ee
The time evolution of $m^x(t)$
in the homogeneous case ($\mu_j=\mu$ for all $j$)  has  recently been
computed by Calabrese et al~\cite{Calabrese11} for quantum quenches within the ferromagnetic
phase, that is, $\mu<1$. In this regime $m^x(t)$ was shown to decay exponentially with time,
$m^x(t)\sim e^{-ct}$, where the coefficient $c>0$ can be easily computed from the quasiparticle spectrum and
Bogoliubov angles, see~\cite{Calabrese11} for details. Within the mean-field approximation
correlations between errors on different qubits are neglected. Then
the number of errors in the final state has expected value $N(1-m^x(t))/2$ and standard deviation $O(\sqrt{N})$.
Hence the error correction is likely to  succeed whenever
\[
N(1-m^x(t))/2 + O(\sqrt{N})\ll N/2,
\]
that is,
\[
m^x(t)\sim e^{-ct} \gg \frac1{\sqrt{N}}.
\]
This gives the desired logarithmic scaling of the storage time.

\subsubsection*{Disordered case\label{sec:numericsdisorderedcase}}
\begin{figure}
\subfigure[Storage fidelity $F_{\ket{g}}(t)$ for the clean case with $\mu=\epsilon=0.5$ and for $5$~disorder realizations with $\mu=0.5$ and randomness strength~$\eta=0.25$. The system size is~$N=12$.\label{figrandreal}]{
\epsfig{file=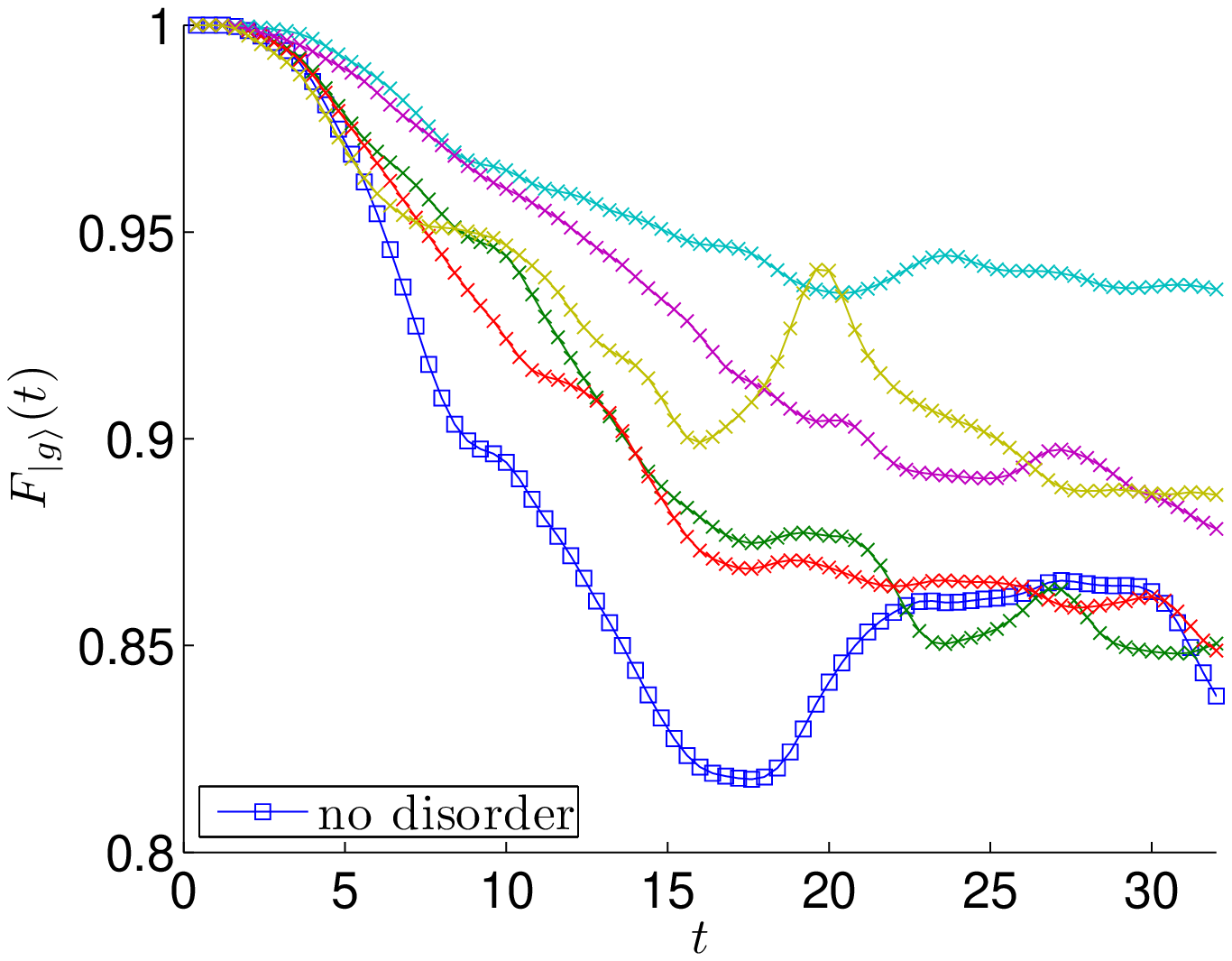,width=8.0cm}
}\ \ \
\subfigure[Plot of $\log_2 T_{\textrm{storage}}(F_0)$ versus $\log_2 N$ for the fidelity threshold $F_0=0.96$ with the same parameters as in Fig.~\ref{figrandreal}. For each system size, $10$~different disorder realizations are considered.  A straight line is fitted to the average $\log_2\bar{T}$ over disorder realizations, showing a linear relationship between~$T_{\textrm{storage}}(F_0)$ and~$N$.]{
\epsfig{file=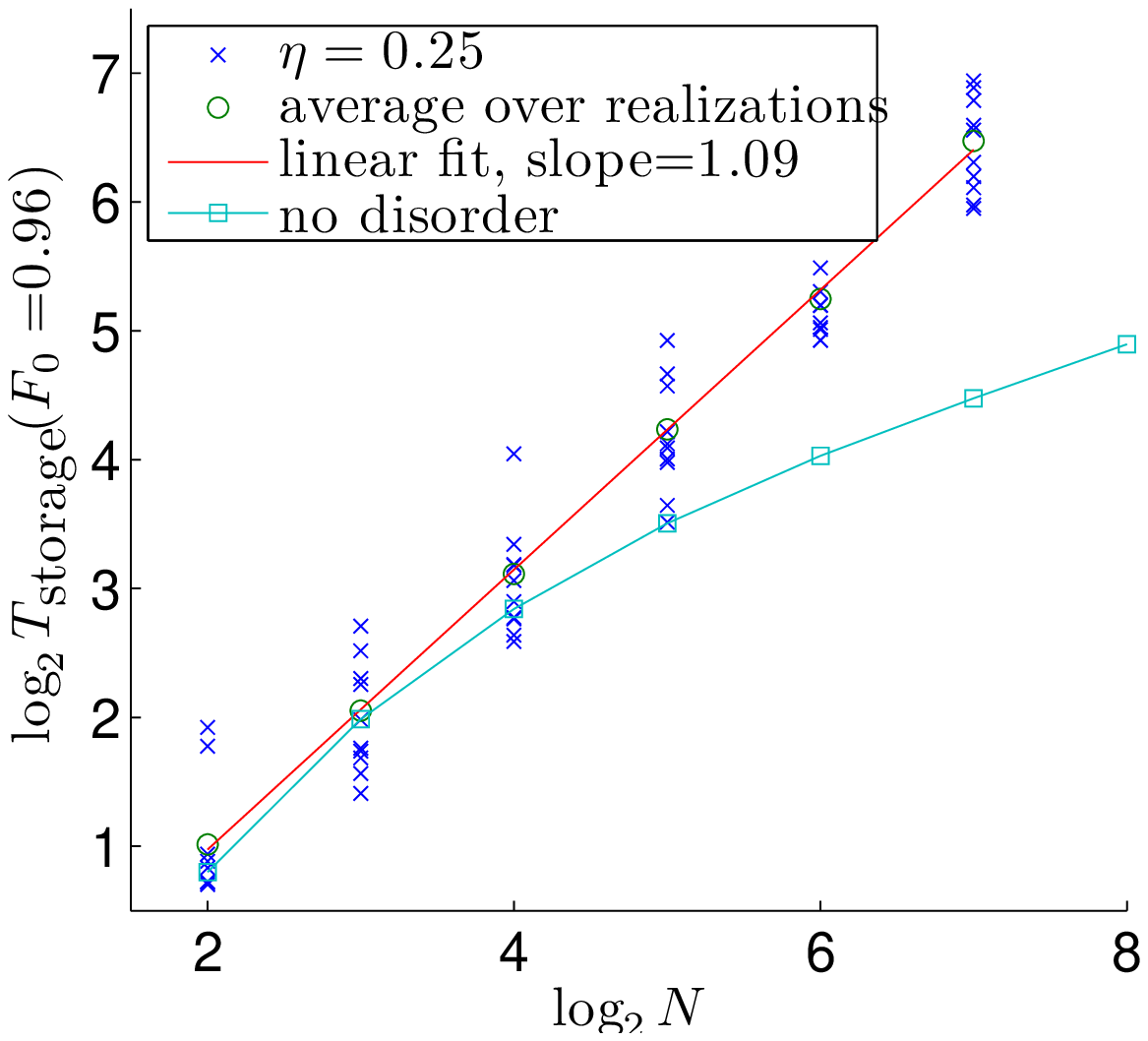,width=8.0cm}
   \label{storagetimerandom:subfig1}
 }
\caption{(Disorder) Fidelity and storage time for different randomness realizations and system sizes, for fixed perturbation- and randomness strength}
\end{figure}

Next we choose site-dependent chemical potentials~$\{\mu_j\}$ according to~\eqref{mu} with $\{x_j\}$ drawn independently and uniformly on $[-1,1]$.
To compare with the clean case, we use $N=12$ and the average chemical potential~$\mu=0.5$ to be in the regime~$\epsilon^2N>1$ of intermediate system size (in the clean case where $\epsilon=\mu$), see Section~\ref{subs:results}.  Fig.~\ref{figrandreal} shows the behavior of storage fidelity~$F_{\ket{g}}(t)$ without disorder~$(\eta=0)$  with disorder strength~$\eta=0.25$ for several disorder realizations.

Fig.~\ref{storagetimerandom:subfig1} shows that the (disorder-averaged) storage time scales roughly linearly with the system size in the presence of disorder. This is in stark contrast to the clean case (Fig.~\ref{figstortime:subfig2}). While this figure clearly shows a significant enhancement of storage times, it falls short of the exponential scaling of Theorem~\ref{thm:storage}. It appears that the regime where Theorem~\ref{thm:storage} applies is outside the range of what is numerically accessible.

\subsubsection*{Pseudorandom disorder}
Suppose now that disorder represents a controlled  external potential introduced on purpose to
enhance the storage time of the system.  Our numerical results demonstrate that for fixed parameters of the model the storage time varies strongly for different disorder realizations  lacking a self-averaging behavior. It shows that some disorder potentials are better than others in terms of their ability to suppress propagation of excitations and enhance the storage time.
A natural question is whether it is possible to choose a {\em deterministic} disorder potential that typically outperforms the random one in terms of the storage time when combined with an unknown homogeneous perturbation.

It is well known that certain deterministic potentials can give rise to localization of eigenvectors
in the Anderson model,  mimicking random disorder. Well-studied examples include
periodic potentials whose period is incommensurate with the one of the lattice, such as
Harper's potential~\cite{Ostlund83,Frohlich90} and pseudo-random potentials~\cite{Griniastyetal88,Brenner92}.

Here we focus on a potential generated by  iterations of the logistic map, namely,
$\mu_j=\mu+\eta(1-2y_j)$, where the sequence $\{y_j\}$ is defined by the initial
condition $y_1\in [0,1]$ and the rule
\begin{align}
y_{j+1}&=a\cdot y_j(1-y_j).\label{eq:logisticmap}
\end{align}
One can easily check that $y_j\in [0,1]$ for all $j$
provided that $0\le a\le 4$. The sequence $\{y_j\}$ is known to exhibit chaotic behavior for~$3.57\lesssim a< 4$ and almost all initial conditions.

With the hope of getting a long storage time, it is desirable to select values~$(y_1,a)$ giving rise to highly localized wavefunctions. To define an empirical measure of localization, we use the symmetric matrix~$H_s$ of Eq.~\eqref{eq:hslocalization}. This matrix has eigenvectors $\ket{\phi^{\pm}_j}$ with eigenvalues $\pm \lambda_j$, where  $\lambda_1 \leq \ldots \leq  \lambda_N$ are the Williamson eigenvalues of~$H_0+V$.  Omitting the eigenvectors~$\ket{\phi^\pm_1}$ corresponding to the boundary modes, we expect  the matrix entries of $T_{p,q}=\sum_{\alpha=2,s\in\{+,-\}}^N|\spr{p}{\phi^s_\alpha}\spr{\phi^s_\alpha}{q}|$ to decay exponentially away from the diagonal, $T_{p,q}\leq c\cdot e^{-|p-q|/\xi}$. Accordingly, we define $\tilde{\xi}(p)$ by least-square fitting a straight line to $(q,\log T_{p,q})$ such that
\begin{align*}
 T_{p,q}&\approx c(p)\cdot e^{-|p-q|/\tilde{\xi}(p)}\qquad\textrm{ for all }q\geq N\quad\textrm{if }p\leq N\\
 T_{p,q}&\approx c(p)\cdot e^{-|p-q|/\tilde{\xi}(p)}\qquad\textrm{ for all }q\leq N\quad\textrm{if }p> N\ .
\end{align*}
We then define the effective localization length as~$\xi_{\textrm{eff}}=\max_p \tilde{\xi}(p)$. This quantity can be computed in time~$O(N^3)$ from the sequence of on-site potentials~$\{\mu_j\}$.
\begin{figure}[H]
\centering
\begin{center}
\subfigure[This shows  $T_{\textrm{storage}}(F_0)$ compared to $\log\xi_{\textrm{eff}}$ for random and pseudorandom disorder. The rightmost box corresponds to a choice of pseudorandom disorder parameters equal to~$(y_1,a)=(0.2845,3.9914)$   (cf.~\eqref{eq:logisticmap}).
  As expected, choosing these parameters so as to minimize the effective localization length~$\xi_{\textrm{eff}}$  appears to be a good strategy for maximizing storage times. Here $N=64$, $\mu=0.5$, $\eta=0.25$ and $F_0=0.97$.]{
\epsfig{file=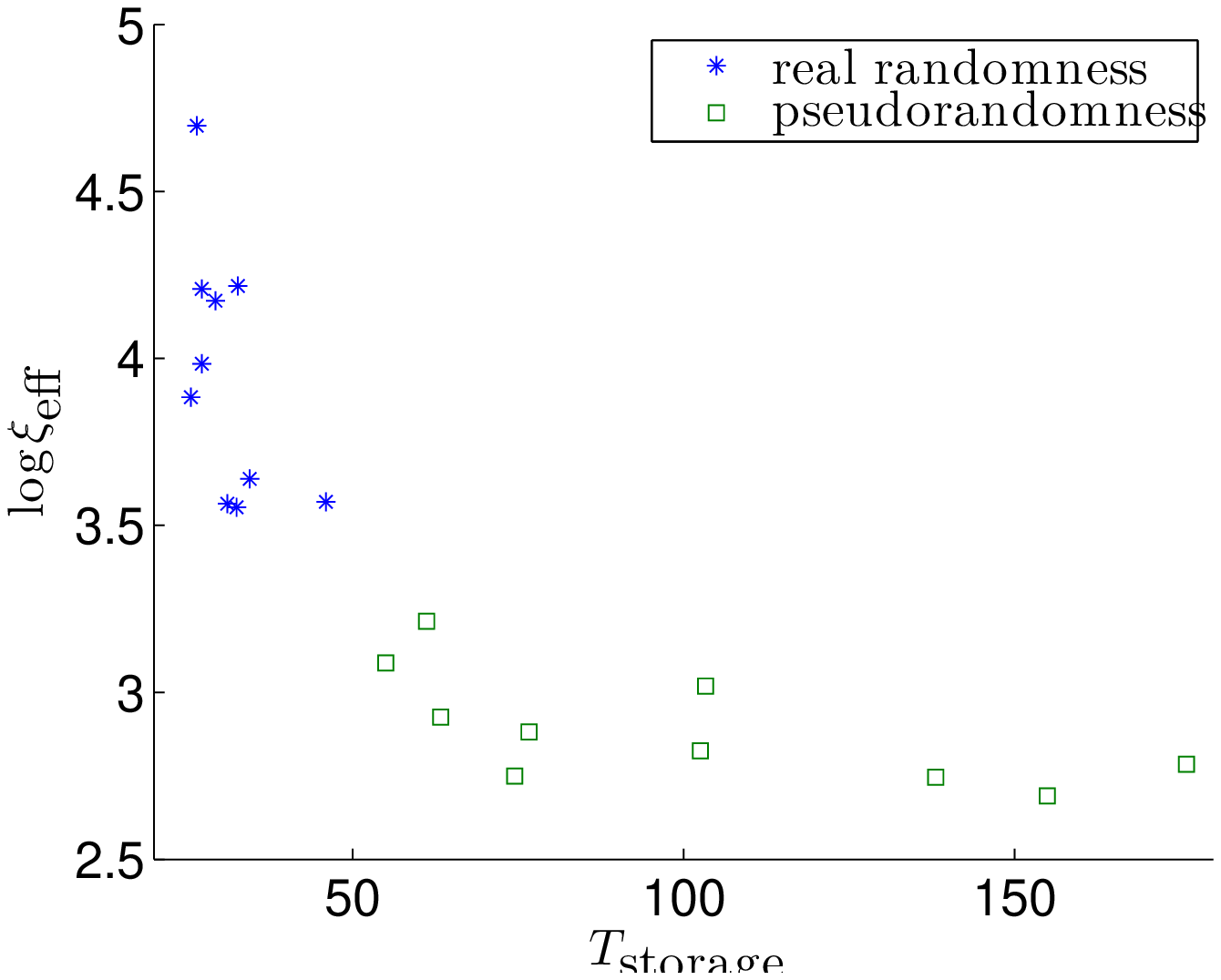,width=8.0cm}
   \label{figpseudo:subfig1}}
\subfigure[Plot  of $\log_2 \bar{T}$  for the average over disorder realizations and
of $\log_2 T_{\textrm{storage}}$
for certain pseudorandom sequences.]{
\epsfig{file=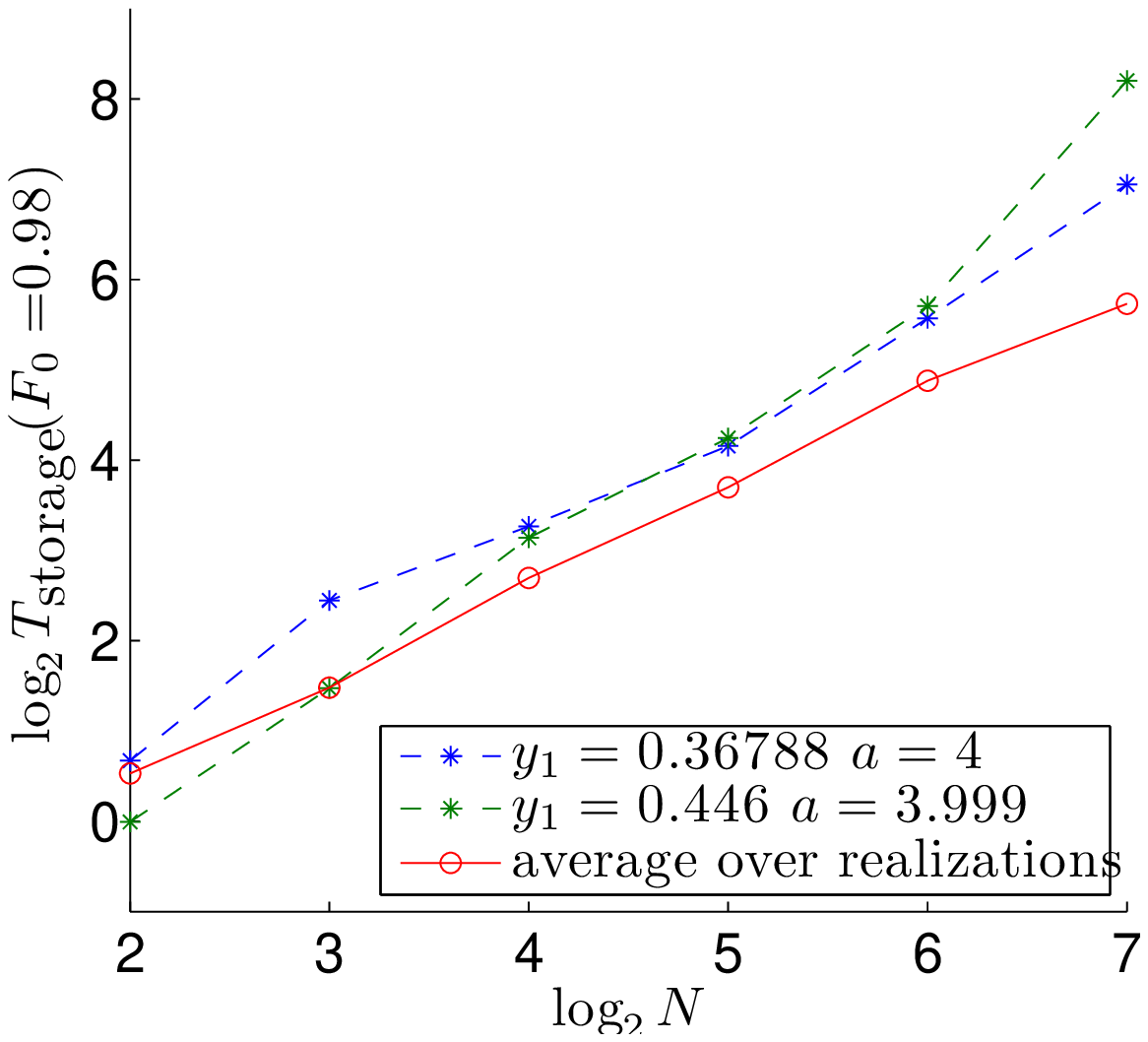,width=8.0cm}
   \label{figpseudo:subfig2}}
\end{center}
\label{figpseudo}
\caption{(Pseudo-random disorder) Enhancement of storage times over random disorder}
\end{figure}

We have computed~$\xi_{\textrm{eff}}$ and storage times for~$N=128$ and a discrete set of initial values~$(y_1,a)$. We expect initial values resulting in small values of~$\xi_{\textrm{eff}}$ to result in long storage times. Figure~\ref{figpseudo:subfig1} qualitatively confirms this expectation: it shows that~$\xi_{\textrm{eff}}$ indeed provides a crude measure for the quality of the quantum memory.  Figure~\ref{figpseudo:subfig2} shows the behavior of the storage time for
two different choices of $(y_1,a)$~and different system sizes. It shows that these pseudorandom sequences give better storage times than the average over truly random realizations.
\subsection*{Acknowledgments}
We would like to thank David DiVincenzo, Alexei Kitaev, Nate Lindner, Tobias Osborne and John Preskill for informative  discussions
and numerous useful comments concerning  Anderson localization.
SB thanks Daniel Loss for sharing his insights on disorder-enhanced stability of  the toric code
memory  that motivated this line of research. Part of this work was done while RK was at the Institute for Quantum Information, Caltech. We acknowledge partial support by the DARPA QuEST program under contract number~HR0011-09-C-0047.

\bibliographystyle{hplain}
\bibliography{q}

\end{document}